\documentclass[manyauthors]{fundam}

\usepackage{amsmath}
\usepackage{amssymb}
\usepackage{mathrsfs}
\usepackage{wasysym}
\usepackage{bbm}		
\usepackage{bm}			
\usepackage{mathtools}
	\mathtoolsset{showonlyrefs=true}		
\usepackage[inline]{enumitem}
\usepackage{nccmath}	
\usepackage{braket}		
\usepackage[normalem]{ulem}	
\usepackage{nicefrac}		
\usepackage{cancel}			
\usepackage{url} 
\usepackage{doi}
\usepackage{graphicx}
\usepackage{tikz}
\usetikzlibrary{shapes,arrows,automata,matrix,fit,patterns,positioning}
\usepackage{xfrac}
\usepackage{array}

\newtheorem{observation}[definition]{Observation}
\newtheorem{question}[definition]{Question}
\newtheorem{algorithmicproblem}{Algorithmic Problem}[section]

\newcommand{\exampleqed}{\ensuremath{\ocircle}\par}
\newcommand{\remarkqed}{\ensuremath{\Diamond}\par}

\newcommand{\ZZ}{\mathbb{Z}}			
\newcommand{\NN}{\mathbb{N}}			
\newcommand{\RR}{\mathbb{R}}			
\newcommand{\PP}{\mathbb{P}}	

\newcommand{\symb}[1]{\mathtt{#1}}		
\newcommand{\isdef}{\triangleq}			
\newcommand{\ee}{\mathrm{e}}			
\newcommand{\bigo}{O}					
\DeclarePairedDelimiter\abs{\lvert}{\rvert}		
\newcommand{\oo}{\circ}					

\newcommand{\diam}{\operatorname{\mathrm{diam}}}			

\newcommand{\xPr}{\operatorname{\mathbb{P}}}		
\newcommand{\xExp}{\operatorname{\mathbb{E}}}		

\newcommand{\Bern}{\mathcal{B}}				
\newcommand{\RV}[1]{\mathbf{#1}}			

\newcommand{\updensity}{\overline{d}}		

\newcommand{\collection}[1]{\mathcal{#1}}	
\newcommand{\family}[1]{\mathscr{#1}}	

\newcommand{\xclass}[1]{\mathbf{#1}}	
\newcommand{\classP}{\xclass{P}}		
\newcommand{\classNP}{\xclass{NP}}		

\makeatletter						
\let\@@pmod\pmod
\DeclareRobustCommand{\pmod}{\@ifstar\@pmods\@@pmod}
\def\@pmods#1{\mkern4mu({\operator@font mod}\mkern 6mu#1)}
\makeatother

\newcommand{\gkl}{\texttt{GKL}}
\newcommand{\toom}{\texttt{NEC-}\allowbreak\texttt{Maj}}
\newcommand{\MRIE}{%
	\texttt{Maj-}\allowbreak\texttt{Random-}\allowbreak\texttt{If-}\allowbreak\texttt{Equal}%
}

\newcommand{\BinA}{\{\symb{0},\symb{1}\}}

\newcommand{\Neighb}{\mathcal{N}}					
\newcommand{\Moore}{\mathcal{M}}					
\newcommand{\vonNeumann}{\mathcal{N}}				

\newcommand{\unif}[1]{\underline{#1}}	
\newcommand{\unifO}{\unif{\symb{0}}}	
\newcommand{\unifI}{\unif{\symb{1}}}	
\DeclareMathOperator*{\maj}{maj}		
\newcommand{\blank}{\diamond}			
\newcommand{\xstop}{\mathmakebox[0.75em][c]{\otimes}}		
\newcommand{\xtrac}{\mathmakebox[0.75em][c]{\bullet}}		
\newcommand{\xnone}{\mathmakebox[0.75em][c]{\circ}}		
\newcommand{\FPert}[2][]{
	\tilde#2\ifthenelse{\equal{#1}{}}{}{\langle#1\rangle}%
}

\newcommand{\Def}{\mathcal{D}}		
\newcommand{\DefNE}{\Def_{\mathrm{NE}}}
\newcommand{\CorNE}{\mathcal{C}_{\mathrm{NE}}}	
\newcommand{\Hom}{\mathcal{H}}
\newcommand{\Col}{\mathcal{C}}
\newcommand{\xspace}[1]{\mathcal{#1}}	
\newcommand{\PS}{\xspace{P}}				
\newcommand{\pa}{$\sfrac 1 2$}			
\newcommand{\incoming}{\mathrm{in}}			
\newcommand{\outgoing}{\mathrm{out}}		

\definecolor{lightgreen}{rgb}{0.7,1,0.7}
\definecolor{lightred}{rgb}{1,0.7,0.7}
\definecolor{lightblue}{rgb}{.7,0.7,1}
\newcommand{\sqaa}[2] 
{\draw[black, fill=lightred] (#1,#2) -- (#1+1,#2) -- (#1+1, #2+1) -- (#1, #2+1) -- cycle; \draw (#1+0.5, #2+0.5) node {$0$} ; }
\newcommand{\sqbb}[2] 
{\draw[black] (#1,#2) -- (#1+1,#2) -- (#1+1, #2+1) -- (#1, #2+1) -- cycle; \draw (#1+0.5, #2+0.5) node {$1$} ; }
\newcommand{\sqcc}[2] 
{\draw[black, fill=lightgreen] (#1,#2) -- (#1+1,#2) -- (#1+1, #2+1) -- (#1, #2+1) -- cycle; \draw (#1+0.5, #2+0.5) node {$2$} ; }
\newcommand{\sqdd}[2] 
{\draw[black, fill=lightblue] (#1,#2) -- (#1+1,#2) -- (#1+1, #2+1) -- (#1, #2+1) -- cycle; \draw (#1+0.5, #2+0.5) node {$3$} ; }

\newcommand{\sqee}[2] 
{\draw[black, fill=lightred] (#1,#2) -- (#1+1,#2) -- (#1+1, #2+1) -- (#1, #2+1) -- cycle; \draw (#1+0.5, #2+0.5) node {$\mathbf{0}$} ; }

\newcommand{\sqE}[4]{
\draw [color=black,line width=1mm,fill=#4] (#1,#2) rectangle ++(1,1);
\draw (#1+.5,#2+.5) node {#3} ;  }

\newcommand{\sqF}[3]{
\draw [color=gray,line width=1.5pt] (#1,#2) rectangle ++(1,1);
\draw (#1+.5,#2+.5) node {#3} ;  }

\newcommand{\sqZ}[2]{{\sqF{#1}{#2}0}}
\newcommand{\sqO}[2]{
\draw [color=black,line width=3pt] (#1+0.1,#2+0.1) rectangle ++(0.8,0.8);
\draw (#1+.5,#2+.5) node {1} ;  
}

\newcommand{\basisFour}[4]{
\sqE 0 0 {$\alpha$}{#1}
\sqE 0 1 {$\beta$}{#2}
\sqE 1 0 {$\delta$}{#4}
\sqE 1 1 {$\gamma$}{#3}
}

\newcommand{\crownFour}{
\sqF {-1} 0 {$a$}
\sqF {-1} 1 {$b$}
\sqF 0 2 {$c$}
\sqF 1 2 {$d$}
\sqF 2 0 {$f$}
\sqF 2 1 {$e$}
\sqF 0 {-1} {$h$}
\sqF 1 {-1} {$g$}
}

\newcommand{\basisLarge}{
\sqE 0 0 {}{white}
\sqE 0 1 {}{white}
\sqE 0 2 {}{white}
\sqE 0 3 {$b_1$}{white}
\sqE 1 0 {}{white}
\sqE 1 1 {}{white}
\sqE 1 2 {}{white}
\sqE 1 3 {$b_2$}{white}
\sqE 2 0 {$\hdots$}{white}
\sqE 2 1 {}{white}
\sqE 2 2 {}{white}
\sqE 2 3 {$\hdots$}{white}
\sqE 3 0 {$b_{\ell^2}$}{white}
\sqE 3 1 {}{white}
\sqE 3 2 {}{white}
\sqE 3 3 {}{white}
}

\newcommand{\crownLarge}{
\sqF {-1} 0 {$a_1$}
\sqF {-1} 1 {$a_2$}
\sqF {-1} 2 {$\vdots$}
\sqF {-1} 3 {$a_{\ell}$}
\sqF 0 4 {$a_{\ell+1}$}
\sqF 1 4 {$a_{\ell+2}$}
\sqF 2 4 {$\hdots$}
\sqF 3 4 {$a_{2\ell}$}
\sqF 4 3 {$a_{2\ell+1}$}
\sqF 4 2 {$a_{2\ell+2}$}
\sqF 4 1 {$\vdots$}
\sqF 4 0 {$a_{3\ell}$}
\sqF 3 {-1} {$a_{3\ell+1}$}
\sqF 2 {-1} {$a_{3\ell+2}$}
\sqF 1 {-1} {$\hdots$}
\sqF 0 {-1} {$a_{4\ell}$}
}

\newcommand{\twobytwo}[2]{
\fill[pattern=north east lines] (#1+1,#2+1)--(#1+1,#2+2)--(#1+2,#2+2)--(#1+2,#2+1)--cycle ;
\draw[color=green, line width=0.5mm] (#1,#2+3)--(#1+3,#2+3)--(#1+3,#2+1) ;
\draw[color=green, line width=0.5mm] (#1+2,#2+2)--(#1+3,#2+2) ;
\draw[color=green, line width=0.5mm] (#1+1,#2+2)--(#1+1,#2+3) ;
\draw[color=green, line width=0.5mm] (#1+2,#2+2)--(#1+2,#2+3) ;
\draw[color=green, line width=0.5mm] (#1+3,#2+1)--(#1+3,#2+3) ;
\draw[color=green, line width=0.5mm] (#1-1,#2+4)--(#1+3,#2+4)--(#1+3,#2+3)--(#1+4,#2+3)--(#1+4,#2+0) ;
\draw[color=green, line width=0.5mm] (#1+0,#2+3)--(#1+0,#2+4) ;
\draw[color=green, line width=0.5mm] (#1+1,#2+3)--(#1+1,#2+4) ;
\draw[color=green, line width=0.5mm] (#1+2,#2+3)--(#1+2,#2+4) ;
\draw[color=green, line width=0.5mm] (#1+3,#2+1)--(#1+4,#2+1) ;
\draw[color=green, line width=0.5mm] (#1+3,#2+2)--(#1+4,#2+2) ;
\draw[color=red, line width=1mm] (#1+1,#2+2)--(#1+2,#2+2)--(#1+2,#2+1) ; 
\draw[ultra thick] (#1,#2)--(#1+2,#2)--(#1+2,#2+2)--(#1,#2+2)--(#1,#2) ; 
}

\newcommand{\twobytwobis}[2]{
\fill[pattern=north east lines] (#1+1,#2+1)--(#1+1,#2+2)--(#1+2,#2+2)--(#1+2,#2+1)--cycle ;
\draw[color=green, line width=0.5mm] (#1,#2+3)--(#1+3,#2+3)--(#1+3,#2+1) ;
\draw[color=green, line width=0.5mm] (#1+2,#2+2)--(#1+3,#2+2) ;
\draw[color=green, line width=0.5mm] (#1+1,#2+2)--(#1+1,#2+3) ;
\draw[color=green, line width=0.5mm] (#1+2,#2+2)--(#1+2,#2+3) ;
\draw[color=green, line width=0.5mm] (#1+3,#2+1)--(#1+3,#2+3) ;
\draw[color=green, line width=0.5mm] (#1-1,#2+4)--(#1+3,#2+4)--(#1+3,#2+3)--(#1+4,#2+3)--(#1+4,#2+0) ;
\draw[color=green, line width=0.5mm] (#1+0,#2+3)--(#1+0,#2+4) ;
\draw[color=green, line width=0.5mm] (#1+1,#2+3)--(#1+1,#2+4) ;
\draw[color=green, line width=0.5mm] (#1+2,#2+3)--(#1+2,#2+4) ;
\draw[color=green, line width=0.5mm] (#1+3,#2+1)--(#1+4,#2+1) ;
\draw[color=green, line width=0.5mm] (#1+3,#2+2)--(#1+4,#2+2) ;
\draw[color=red, line width=1mm] (#1+1,#2+2)--(#1+2,#2+2) ; 
\draw[ultra thick] (#1,#2)--(#1+2,#2)--(#1+2,#2+2)--(#1,#2+2)--(#1,#2) ; 
}

\newcommand{\twobytwoter}[2]{
\fill[pattern=north east lines] (#1+1,#2+1)--(#1+1,#2+2)--(#1+2,#2+2)--(#1+2,#2+1)--cycle ;
\draw[color=green, line width=0.5mm] (#1,#2+3)--(#1+3,#2+3)--(#1+3,#2+1) ;
\draw[color=green, line width=0.5mm] (#1+2,#2+2)--(#1+3,#2+2) ;
\draw[color=green, line width=0.5mm] (#1+1,#2+2)--(#1+1,#2+3) ;
\draw[color=green, line width=0.5mm] (#1+2,#2+2)--(#1+2,#2+3) ;
\draw[color=green, line width=0.5mm] (#1+3,#2+1)--(#1+3,#2+3) ;
\draw[color=green, line width=0.5mm] (#1-1,#2+4)--(#1+3,#2+4)--(#1+3,#2+3)--(#1+4,#2+3)--(#1+4,#2+0) ;
\draw[color=green, line width=0.5mm] (#1+0,#2+3)--(#1+0,#2+4) ;
\draw[color=green, line width=0.5mm] (#1+1,#2+3)--(#1+1,#2+4) ;
\draw[color=green, line width=0.5mm] (#1+2,#2+3)--(#1+2,#2+4) ;
\draw[color=green, line width=0.5mm] (#1+3,#2+1)--(#1+4,#2+1) ;
\draw[color=green, line width=0.5mm] (#1+3,#2+2)--(#1+4,#2+2) ;
\draw[color=red, line width=1mm] (#1+2,#2+2)--(#1+2,#2+1) ; 
\draw[ultra thick] (#1,#2)--(#1+2,#2)--(#1+2,#2+2)--(#1,#2+2)--(#1,#2) ; 
}

\newcommand{\figOneFill}{
\sqE 0 0 {$\alpha$}{white}
\sqF {-1} 0 {$a$}
\sqF {0} 1 {$b$}
\sqF 1 0 {$c$}
\sqF  0 {-1} {$d$}
}

\newcommand{\figFourCase}{
\basisFour{white}{white}{white}{white}
\crownFour
}

\newcommand{\figEllCase}{
\basisLarge{white}
\crownLarge
}

\newcommand{\figNeighbFourCol}{
\foreach \x in {-1,...,4} \draw (\x,-1)--(\x,4) ;
\foreach \y in {-1,...,4} \draw (-1,\y)--(4,\y) ;

\fill[pattern=north east lines] (1,1)--(1,2)--(2,2)--(2,1)--cycle ;

\draw[color=green, line width=1mm] (-1,4)--(3,4)--(3,1) ;
\draw[color=green, line width=1mm] (0,3)--(4,3)--(4,0) ;
\draw[color=green, line width=1mm] (0,3)--(4,3)--(4,0) ;
\draw[color=green, line width=1mm] (3,1)--(4,1) ;
\draw[color=green, line width=1mm] (2,2)--(4,2) ;
\draw[color=green, line width=1mm] (0,3)--(0,4) ;
\draw[color=green, line width=1mm] (1,2)--(1,4) ;
\draw[color=green, line width=1mm] (2,2)--(2,4) ;
\draw[color=red, line width=1.5mm] (1,2)--(2,2)--(2,1) ; 
\draw[ultra thick] (0,0)--(2,0)--(2,2)--(0,2)--(0,0) ; 
}

\newcommand{\figBlockTwo}{
\draw[color=red, line width=1mm] (1,2)--(2,2) ; 
\draw[color=red, line width=1mm] (7,6)--(8,6) ; 
\draw[color=red, line width=1mm] (3,6)--(4,6) ; 
\draw[color=red, line width=1mm] (2,7)--(2,5) ; 
\foreach \x in {0,...,10} \draw (\x,0)--(\x,10) ;
\foreach \y in {0,...,10} \draw (0,\y)--(10,\y) ;
\twobytwo{2}{0} 
\twobytwobis{5}{1} 
\twobytwoter{1}{6} 
\twobytwo{5}{5} 
\draw[line width=0.5mm, style=dotted] (3,10)--(10,3) ;
}

\newcommand{\figNeighbFourColGeneral}{
\foreach \x in {-1,...,6} \draw (\x,-1)--(\x,6) ;
\foreach \y in {-1,...,6} \draw (-1,\y)--(6,\y) ;

\fill[pattern=north east lines] (2,2)--(2,3)--(3,3)--(3,2)--cycle ;

\draw[color=green, line width=1mm] (-1,6)--(5,6)--(5,1) ;
\draw[color=green, line width=1mm] (0,5)--(6,5)--(6,0) ;
\draw[color=green, line width=1mm] (1,4)--(6,4) ;
\draw[color=green, line width=1mm] (4,2)--(4,6) ;
\draw[color=green, line width=1mm] (0,5)--(0,6) ;
\draw[color=green, line width=1mm] (1,4)--(1,6) ;
\draw[color=green, line width=1mm] (2,3)--(2,6) ;
\draw[color=green, line width=1mm] (3,3)--(3,6) ;
\draw[color=green, line width=1mm] (5,1)--(6,1) ;
\draw[color=green, line width=1mm] (4,2)--(6,2) ;
\draw[color=green, line width=1mm] (3,3)--(6,3) ;

\draw[color=red, line width=1.5mm] (2,3)--(3,3)--(3,2) ; 
\draw[ultra thick] (0,0)--(3,0)--(3,3)--(0,3)--(0,0) ; 
}

\newcommand{\sqG}[3]{\sqF #1 #2 {$#3$}}
\newcommand{\crownMaj}{
\sqZ {-1} 0 
\sqF {-1} {1} {$\vdots$} 
\sqZ {-1} 2 
\sqZ {-1} 3 
\sqZ 0 4 
\sqZ 1 4 
\sqG 2 4 {\dots}
\sqZ 3 4 
\sqZ 4 3 
\sqZ 4 2 
\sqF {4} {1} {$\vdots$} 
\sqZ 4 0 
\sqZ 3 {-1}
\sqF 2 {-1} {$\dots$} 
\sqZ 1 {-1} 
\sqZ 0 {-1} 
\sqO 0 2 
\sqG 0 2 {}
\sqO 1 2 
\sqG 1 2 {}
\sqG 2 2 {\dots}
\sqO 3 2 
\sqG 3 2 {}
\sqG 0 3 {x_1}
\sqG 1 3 {x_2}
\sqG 2 3 {\ldots}
\sqG 3 3 {x_k}
\sqG 0 1 {}
\sqG 1 1 {}
\sqG 2 1 {}
\sqG 3 1 {}
\sqG 0 0 {}
\sqG 1 0 {}
\sqG 2 0 {}
\sqG 3 0 {}
}

\newcommand{\Earrb}[2]{
\draw[->, ultra thick, color=blue] (#1+0.1,#2)--(#1+0.9,#2) ;}

\newcommand{\Earrc}[2]{
\draw[->, ultra thick, color=cyan] (#1+0.1,#2)--(#1+0.9,#2) ;}

\newcommand{\Sarrb}[2]{
\draw[->, ultra thick, color=blue] (#1,#2-0.1)--(#1,#2-0.9) ;}

\newcommand{\Sarrc}[2]{
\draw[->, ultra thick, color=cyan] (#1,#2-0.1)--(#1,#2-0.9) ;}

\newcommand{\dwn}[2]{
{\draw[thick,->] (#1,#2+0.9)--(#1,0.1+#2) ;}
}
\newcommand{\upw}[2]{
{\draw[thick,->] (#1,#2+0.1)--(#1,0.9+#2) ;}
}
\newcommand{\rgt}[2]{
{\draw[thick,->] (#1+0.1,#2)--(#1+.9,#2) ;}
}
\newcommand{\lft}[2]{
{\draw[thick,->] (#1+0.9,#2)--(#1+.1,#2) ;}
}

\newcommand{\pc}[3]{
\ifnum#3=1
        \sqaa{#1}{#2};
\fi
\ifnum#3=2
        \sqbb{#1}{#2};
\fi
\ifnum#3=3
        \sqcc{#1}{#2};
\fi
}

\newcommand{\sq}[2] 
{\draw[black] (#1,#2) -- (#1+1,#2) -- (#1+1, #2+1) -- (#1, #2+1) -- cycle; }

\newcommand{\pcQ}[5]{
\pc 0 #5 #1 
\pc 1 #5 #2
\pc 2 #5 #3
\pc 3 #5 #4
}

\newcommand{\upd}[3]{
\ifnum#3=0
        \upw{#1}{#2};
\fi
\ifnum#3=1
        \dwn{#1}{#2};
\fi
}

\newcommand{\lgd}[3]{
\ifnum#3=0
        \lft{#1}{#2};
\fi
\ifnum#3=1
        \rgt{#1}{#2};
\fi
}

\newcommand{\updT}[4]{
\upd 1 #4 #1
\upd 2 #4 #2
\upd 3 #4 #3
}

\newcommand{\lrT}[4]{
\lgd #4 1 #1
\lgd #4 2 #2
\lgd #4 3 #3
}

\newcommand{\SVcaseError}{
\foreach \i in {0,3,5,8,11} \foreach \j in {0,3,6,9} {\sqaa{\i}{\j}} ;
\foreach \i in {1,4,9} \foreach \j in {1,4,7,10} {\sqaa{\i}{\j}} ;
\foreach \i in {2,5,7,10} \foreach \j in {2,5,8,11} {\sqaa{\i}{\j}} ;

\foreach \i in {0,3,8,11} \foreach \j in {2,5,8,11} {\sqbb{\i}{\j}} ;
\foreach \i in {1,4,9} \foreach \j in {0,3,6,9} {\sqbb{\i}{\j}} ;
\foreach \i in {2,5,7,10} \foreach \j in {1,4,7,10} {\sqbb{\i}{\j}} ;

\foreach \i in {0,3,8,11} \foreach \j in {1,4,7,10} {\sqcc{\i}{\j}} ;
\foreach \i in {1,4,9} \foreach \j in {2,5,8,11} {\sqcc{\i}{\j}} ;
\foreach \i in {2,5,7,10} \foreach \j in {0,3,6,9} {\sqcc{\i}{\j}} ;

\foreach \j in {0,3,7,10} {\sqaa{6}{\j}} ;
\foreach \j in {2,5,6,9,11} {\sqbb{6}{\j}} ;
\foreach \j in {1,4,8} {\sqcc{6}{\j}} ;

\foreach \x in {1,2,3,4,5,8,9,10,11} \foreach \y in {0,...,11}{\draw[thick, ->] (\x,\y+0.1)--(\x,0.9+\y) ;} ;
\foreach \x in {0,1,2,3,4,5,7,8,9,10,11} \foreach \y in {1,...,11}{\draw[thick, ->] (0.1+\x,\y)--(0.9+\x,\y) ;} ;

\foreach \x in {6} \foreach \y in {0,1,2,3,4,5,11} {\draw[->, thick] (\x,\y+0.1)--(\x,0.9+\y) ;} ; 
\foreach \x in {6} \foreach \y in {6,...,10} {\draw[->, line width=2pt] (\x,\y+0.9)--(\x,0.1+\y) ;} ; 
\foreach \x in {7} \foreach \y in {0,1,2,3,4,5,11} {\draw[->, line width=2pt] (\x,\y+0.9)--(\x,0.1+\y) ;} ; 
\foreach \x in {7} \foreach \y in {6,...,10} {\draw[->, thick] (\x,\y+0.1)--(\x,0.9+\y) ;} ; 

\foreach \x in {6} \foreach \y in {1,2,3,4,5,7,8,9,10} {\draw[thick, ->] (0.1+\x,\y)--(0.9+\x,\y) ;} ; 
\foreach \x in {6} \foreach \y in {11} {\draw[line width=2pt, ->] (0.9+\x,\y)--(0.1+\x,\y) ;} ; 

\fill[pattern=north east lines, style=semitransparent, even odd rule] (0,0) (1,1)--(1,11)--(11,11)--(11,1)--cycle (0,0)--(0,12)--(12,12)--(12,0)--cycle;
}

\newcommand{\figExampleSV}{
\begin{tikzpicture}[scale=0.7,>=stealth']
\pcQ 3 2 3 1 3
\pcQ 1 3 1 2 2
\pcQ 3 1 2 1 1
\pcQ 1 2 3 2 0
\updT 1 0 0 3
\updT 1 0 0 2
\updT 0 0 1 1
\updT 0 0 1 0
\lrT 1 0 1 3
\lrT 1 1 1 2
\lrT 1 1 1 1
\lrT 1 0 1 0
\end{tikzpicture}
}

\newcommand{\verC}{
\begin{tikzpicture}[>=stealth']
\sq{0}{0} ; \sq{1}{0} ; {\draw[ultra thick, ->] (1,0) -- (1,1) ;}
\draw (0.5,0.5) node {$q$} ;
\draw (1.5,0.5) node {$q+1$} ;
\begin{scope}[shift={(0,1.2)}]
\sq{0}{0} ; \sq{1}{0} ; {\draw[ultra thick, <-] (1,0) -- (1,1) ;}
\draw (0.5,0.5) node {$q$} ;
\draw (1.5,0.5) node {$q-1$} ;
\end{scope}
\end{tikzpicture}
}

\newcommand{\horC}{
\begin{tikzpicture}[>=stealth']
\sq{0}{0} ; \sq{0}{1} ; {\draw[ultra thick, ->] (0,1) -- (1,1) ;}
\draw (0.5,1.5) node {$q$} ;
\draw (0.5,0.5) node {$q+1$} ;
\end{tikzpicture}
\begin{tikzpicture}[>=stealth']
\sq{0}{0} ; \sq{0}{1} ; {\draw[ultra thick, <-] (0,1) -- (1,1) ;}
\draw (0.5,1.5) node {$q$} ;
\draw (0.5,0.5) node {$q-1$} ;
\end{tikzpicture}
}

\newcommand{\figSVcontour}{

\foreach \x in {0,...,12} \draw (\x,0)--(\x,12) ;
\foreach \y in {0,...,12} \draw (0,\y)--(12,\y) ;

\foreach \x in {5,8,9} 
\foreach \y in {11}
{\draw[->, line width=1.5pt] (\x,\y+0.9)--(\x,0.1+\y) ;}
\draw (9,12.5) node {$1$} ;
\draw (8,12.5) node {$2$} ;
\draw (5,12.5) node {$3$} ;

\foreach \x in {0} 
\foreach \y in {3,7}
{\draw[->, line width=1.5pt] (\x+0.1,\y)--(\x+0.9,\y) ;}
\draw (-0.5,7) node {$4$} ;
\draw (-0.5,3) node {$5$} ;

\foreach \x in {11} 
\foreach \y in {2,3,6,9}
{\draw[->, line width=1.5pt] (\x+0.1,\y)--(\x+0.9,\y) ;}
\draw (12.5,9) node {$1$} ;
\draw (12.5,6) node {$2$} ;
\draw (12.5,3) node {$3$} ;
\draw (12.5,2) node {$4$} ;

\foreach \x in {6} 
\foreach \y in {0}
{\draw[->, line width=1.5pt] (\x,\y+0.9)--(\x,0.1+\y) ;}
\draw (6,-0.5) node {$5$} ;

\Earrb{9}{11} \Earrb{10}{11} \Sarrb{11}{11} \Sarrb{11}{10} 
\Earrc{8}{11} \Sarrc{9}{11} \Earrc{9}{10} \Sarrc{10}{10} \Earrc{10}{9} \Sarrc{11}{9} \Sarrc{11}{8} \Sarrc{11}{7} 
\Earrb{5}{11} \Earrb{6}{11} \Earrb{7}{11} \Sarrb{8}{11} \Earrb{8}{10} \Sarrb{9}{10} \Earrb{9}{9} \Sarrb{10}{9} \Sarrb{10}{8} \Sarrb{10}{7} \Earrb{10}{6} \Sarrb{11}{6} \Sarrb{11}{5} \Sarrb{11}{4} 
\Earrc{1}{7} \Earrc{2}{7} \Earrc{3}{7} \Earrc{4}{7} \Earrc{5}{7} \Earrc{6}{7} \Earrc{7}{7} \Earrc{8}{7} \Sarrc{9}{7} \Earrc{9}{6} \Sarrc{10}{6} \Sarrc{10}{5} \Sarrc{10}{4}  \Earrc{10}{3} \Sarrc{11}{3}
\Earrb{1}{3} \Earrb{2}{3} \Earrb{3}{3} \Earrb{4}{3} \Earrb{5}{3} \Sarrb{6}{3} \Sarrb{6}{2}

\fill[pattern=north east lines, style=semitransparent, even odd rule] (0,0) (1,1)--(1,11)--(11,11)--(11,1)--cycle (0,0)--(0,12)--(12,12)--(12,0)--cycle;

}

\newcommand{\CorrectionOneDVersionTwo}{
\coordinate (A) at (0,0);
\newcommand{\xA}{0} ;
\newcommand{\yA}{0} ;
\coordinate (B) at (1,0);
\newcommand{\xB}{1} ;
\newcommand{\yB}{0} ;
\coordinate (C) at (2,-1);
\newcommand{\xC}{2} ;
\newcommand{\yC}{-1} ;
\coordinate (D) at (6,-2.5);
\newcommand{\xD}{6} ;
\newcommand{\yD}{-2.5} ;
\coordinate (E) at (6,-6);
\newcommand{\xE}{6} ;
\newcommand{\yE}{-6} ;
\coordinate (F) at (6,-3);
\newcommand{\xF}{6} ;
\newcommand{\yF}{-3} ;
\draw[overlay,help lines] (-1,0) -- (7,0);
\draw[ultra thick,decorate,decoration={snake,segment length=5pt,amplitude=1pt}] (A) -- (B);
\fill[gray!30] (A) -- (E) -- (F) -- cycle;
\fill[gray!30] (B) -- (C) -- (D) -- cycle;
\fill[pattern=crosshatch] (C) -- (F) -- (D) -- cycle;
\draw[very thick, ->-] (A) -- (E);
\draw[very thick, ->-] (A) -- (C);
\draw[very thick, ->-] (C) -- (F);
\draw[very thick, ->-] (B) -- (C);
\draw[very thick, ->-] (B) -- (D);
\draw[very thick, ->-] (F) -- (E);
\draw[very thick, dashed, ->-] (C) -- (D);
\draw[very thick, dashed, ->-] (D) -- (F);

\node[cross out, red, draw, very thick, minimum size=5pt, inner sep=0pt, outer sep=0pt] at (C) {};

\node at (0, 0.5) {};		
\tikzstyle{quadri}=[rectangle,draw,color=blue,fill=yellow!50,text=blue]
\tikzstyle{estun}=[->,thin, color=blue]
\begin{scope}[overlay]
\node[quadri] (FAD) at (0,-2) {Fading of the traces (speed 1)};
\draw[estun, shorten >=2pt] (FAD)to[bend left](0.8*\xA+0.2*\xE,0.8*\yA+0.2*\yE) ;
\draw[estun, shorten >=2pt] (FAD)to[bend left](0.4*\xB+0.6*\xC,0.4*\yB+0.6*\yC) ;
\node[quadri] (FRONT) at (2.5,0.5) {Front of correction trails (speed 2)};
\draw[estun,shorten >=1pt] (FRONT)to[bend left](0.6*\xB+0.4*\xD,0.6*\yB+0.4*\yD) ;
\draw[estun,shorten >=1pt] (FRONT)to[bend left](0.4*\xA+0.6*\xF,0.4*\yA+0.6*\yF) ;
\node[quadri] (STOP) at (5.5,-0.5) {Propagation of a stop signal (speed 4)};
\draw[estun, shorten >=2pt] (STOP)to[bend left](0.4*\xC+0.6*\xD,0.4*\yC+0.6*\yD) ;
\node[quadri] (xSTOP) at (0.6, -3) {Creation of a stop signal};
\draw[estun, shorten >=2pt] (xSTOP) to[bend right] (C);

\fill[gray!30, overlay] (-0.3,-4) rectangle (0.5,-4.4);
\fill[pattern=crosshatch, overlay] (-0.3,-4.5) rectangle (0.5,-4.9);
\node[anchor=west,overlay] at (0.55,-4.2) {trace};
\node[anchor=west,overlay] at (0.55,-4.7) {stop};
\end{scope}
}

\usepackage{tikz}

\usetikzlibrary{plotmarks}
\usetikzlibrary{snakes}
\usetikzlibrary{calc}
\usetikzlibrary{shapes}
\usetikzlibrary{arrows}
\usetikzlibrary{decorations.pathmorphing}
\usetikzlibrary{decorations.markings}
\usetikzlibrary{patterns}
\usetikzlibrary{matrix}
\usepackage{ifthen}
\usepackage{etoolbox}
\usepackage{calc}
\usepackage{xstring}	
\usepackage{graphicx}
\usepackage{rotating}

\tikzset{
	->-/.style={%
		decoration={markings,mark=at position #1 with {\arrow{>}}},%
		postaction={decorate}%
	},->-/.default=0.5%
}

\newdimen\cellsize
\newcommand{\cell}[2]{
	\IfEqCase{#1}{%
		{0}{\draw[very thin] #2 circle (\cellsize);}%
		{1}{\fill[black] #2 circle (\cellsize);}%
	}[%
		{%
		\fill[#1] #2 circle (\cellsize);
		\draw[very thin, gray] #2 circle (\cellsize);
		}%
	]
}

\newcommand{\xcell}[2]{
	\IfEqCase{#1}{%
		{0}{\fill[red] #2 circle (\cellsize);}%
		{1}{\fill[blue] #2 circle (\cellsize);}%
	}[%
		{%
		\fill[#1] #2 circle (\cellsize);
		\draw[very thin, gray] #2 circle (\cellsize);
		}%
	]
}

\newdimen\smallcellsize
\newcommand{\smallcell}[2]{
	\IfEqCase{#1}{%
		{0}{\draw[very thin] #2 circle (\smallcellsize);}%
		{1}{\fill[black] #2 circle (\smallcellsize);}%
	}[%
		{%
		\fill[#1] #2 circle (\smallcellsize);
		\draw[very thin, gray] #2 circle (\smallcellsize);
		}%
	]
}

\newdimen\mathcellsize
\newdimen\mathcellbase
\newcommand{\mathcell}[1]{
	\tikz[baseline=\mathcellbase]{%
		\IfEqCase{#1}{%
			{0}{\draw[very thin] (0,0) circle (\mathcellsize);}%
			{1}{\fill[black] (0,0) circle (\mathcellsize);}%
		}[%
			{%
			\fill[#1] (0,0) circle (\mathcellsize);
			\draw[very thin, gray] (0,0) circle (\mathcellsize);
			}%
		]
	}
}

\newdimen\tilesize
\newdimen\tilecorner
\newcommand{\ctile}[2]{
	\begin{scope}
		\clip #2 +(-\tilesize,-\tilesize) rectangle +(\tilesize,\tilesize);%
		\IfEqCase{#1}{%
			{0}{%
				\fill[white] #2 +(-\tilesize,-\tilesize) rectangle +(\tilesize,\tilesize);%
				\draw[very thin] #2 +(-\tilesize,-\tilesize) rectangle +(\tilesize,\tilesize);}%
			{1}{%
				\fill[black] #2 +(-\tilesize,-\tilesize) rectangle +(\tilesize,\tilesize);%
				\draw[very thin] #2 +(-\tilesize,-\tilesize) rectangle +(\tilesize,\tilesize);
			}%
		}[%
			{%
			\fill[#1] #2 +(-\tilesize,-\tilesize) rectangle +(\tilesize,\tilesize);%
			\draw[very thin] #2 +(-\tilesize,-\tilesize) rectangle +(\tilesize,\tilesize);%
			}%
		]
	\end{scope}
}

\newcommand{\colorededge}[3]{
	\begin{scope}[rounded corners=\tilecorner]
		\IfEqCase{#1}{%
			{l}{%
				\clip #3 -- +(-\tilesize,-\tilesize) -- +(-\tilesize,\tilesize) -- cycle;
				\fill[#2] #3 -- +(-\tilesize,-\tilesize) -- +(-\tilesize,\tilesize) -- cycle;
				\draw[very thin] #3 -- +(-\tilesize,-\tilesize) -- +(-\tilesize,\tilesize) -- cycle;
			}
			{r}{%
				\clip #3 -- +(\tilesize,-\tilesize) -- +(\tilesize,\tilesize) -- cycle;
				\fill[#2] #3 -- +(\tilesize,-\tilesize) -- +(\tilesize,\tilesize) -- cycle;
				\draw[very thin] #3 -- +(\tilesize,-\tilesize) -- +(\tilesize,\tilesize) -- cycle;
			}
			{d}{%
				\clip #3 -- +(-\tilesize,-\tilesize) -- +(\tilesize,-\tilesize) -- cycle;
				\fill[#2] #3 -- +(-\tilesize,-\tilesize) -- +(\tilesize,-\tilesize) -- cycle;
				\draw[very thin] #3 -- +(-\tilesize,-\tilesize) -- +(\tilesize,-\tilesize) -- cycle;
			}
			{u}{%
				\clip #3 -- +(-\tilesize,\tilesize) -- +(\tilesize,\tilesize) -- cycle;
				\fill[#2] #3 -- +(-\tilesize,\tilesize) -- +(\tilesize,\tilesize) -- cycle;
				\draw[very thin] #3 -- +(-\tilesize,\tilesize) -- +(\tilesize,\tilesize) -- cycle;
			}
		}[%
			\PackageError{colorededge}{Undefined tile side}{See the definition!}
		]
	\end{scope}
}

\newcommand{\wtile}[5]{
	\begin{scope}
		\colorededge{l}{#1}{#5}%
		\colorededge{r}{#2}{#5}%
		\colorededge{d}{#3}{#5}%
		\colorededge{u}{#4}{#5}%
		\draw #5 +(-\tilesize,-\tilesize) rectangle +(\tilesize,\tilesize);%
	\end{scope}
}

\newcommand{\collarrow}{%
	\relax\ifmmode%
		\tikz[>=stealth',baseline=\mathcellbase] \draw[->] (0.5\tilesize,0) -- (-0.5\tilesize,0);%
	\else%
		\tikz[>=stealth'] \draw[->] (0.5\tilesize,0) -- (-0.5\tilesize,0);%
	\fi
}
\newcommand{\colrarrow}{%
	\relax\ifmmode%
		\tikz[>=stealth',baseline=\mathcellbase] \draw[->] (-0.5\tilesize,0) -- (0.5\tilesize,0);%
	\else%
		\tikz[>=stealth'] \draw[->] (-0.5\tilesize,0) -- (0.5\tilesize,0);%
	\fi
}
\newcommand{\coldarrow}{%
	\relax\ifmmode%
		\,\tikz[>=stealth',baseline=\mathcellbase] \draw[->] (0,0.5\tilesize) -- (0,-0.5\tilesize);\,%
	\else%
		\tikz[>=stealth'] \draw[->] (0,0.5\tilesize) -- (0,-0.5\tilesize);%
	\fi
}
\newcommand{\coluarrow}{%
	\relax\ifmmode%
		\,\tikz[>=stealth',baseline=\mathcellbase] \draw[->] (0,-0.5\tilesize) -- (0,0.5\tilesize);\,%
	\else%
		\tikz[>=stealth'] \draw[->] (0,-0.5\tilesize) -- (0,0.5\tilesize);%
	\fi
}

\newcommand{\colltick}{%
	\relax\ifmmode%
		\tikz[>=stealth',baseline=\mathcellbase] {%
			\draw (0.5\tilesize,0) -- (-0.5\tilesize,0);%
			\draw (-0.3\tilesize,-0.3\tilesize) -- (-0.3\tilesize,0.3\tilesize);%
		}	
	\else%
		\tikz[>=stealth'] {%
			\draw (0.5\tilesize,0) -- (-0.5\tilesize,0);%
			\draw (-0.3\tilesize,-0.3\tilesize) -- (-0.3\tilesize,0.3\tilesize);%
		}
	\fi
}
\newcommand{\colrtick}{%
	\relax\ifmmode%
		\tikz[>=stealth',baseline=\mathcellbase] {%
			\draw (0.5\tilesize,0) -- (-0.5\tilesize,0);%
			\draw (0.3\tilesize,-0.3\tilesize) -- (0.3\tilesize,0.3\tilesize);%
		}	
	\else%
		\tikz[>=stealth'] {%
			\draw (0.5\tilesize,0) -- (-0.5\tilesize,0);%
			\draw (0.3\tilesize,-0.3\tilesize) -- (0.3\tilesize,0.3\tilesize);%
		}
	\fi
}
\newcommand{\coldtick}{%
	\relax\ifmmode%
		\tikz[>=stealth',baseline=\mathcellbase] {%
			\draw (0,0.5\tilesize) -- (0,-0.5\tilesize);%
			\draw (-0.3\tilesize,-0.3\tilesize) -- (0.3\tilesize,-0.3\tilesize);%
		}
	\else%
		\tikz[>=stealth'] {%
			\draw (0,0.5\tilesize) -- (0,-0.5\tilesize);%
			\draw (-0.3\tilesize,-0.3\tilesize) -- (0.3\tilesize,-0.3\tilesize);%
		}
	\fi
}
\newcommand{\colutick}{%
	\relax\ifmmode%
		\tikz[>=stealth',baseline=\mathcellbase] {%
			\draw (0,0.5\tilesize) -- (0,-0.5\tilesize);%
			\draw (-0.3\tilesize,0.3\tilesize) -- (0.3\tilesize,0.3\tilesize);%
		}
	\else%
		\tikz[>=stealth'] {%
			\draw (0,0.5\tilesize) -- (0,-0.5\tilesize);%
			\draw (-0.3\tilesize,0.3\tilesize) -- (0.3\tilesize,0.3\tilesize);%
		}
	\fi
}

\newcommand{\wtilelab}[5]{
	\begin{scope}[shift={#5}]
		\draw[very thin] (-\tilesize,-\tilesize) rectangle (\tilesize,\tilesize);%
		\node at (-\tilesize,0) {#1};
		\node at (\tilesize,0) {#2};
		\node at (0,-\tilesize) {#3};
		\node at (0,\tilesize) {#4};
	\end{scope}
}

\newcommand{\wpatchblank}[1]{%
	\wtilelab{}{}{}{}{#1};
}
\newcommand{\wpatchLlu}[2][]{
	\ifthenelse{\equal{#1}{}}{\wtilelab{}{\colrarrow}{\colrarrow}{}{#2};}{%
		\IfEqCase{#1}{%
			{0}{\wtilelab{}{\colrarrow}{\colrarrow}{}{#2};}
			{1}{\wtilelab{\coldarrow}{}{\coldarrow}{}{#2};}
			{2}{\wtilelab{\collarrow}{}{}{\collarrow}{#2};}
			{3}{\wtilelab{}{\coluarrow}{}{\coluarrow}{#2};}
		}
	}
}
\newcommand{\wpatchLcu}[2][]{
	\ifthenelse{\equal{#1}{}}{\wtilelab{\colrarrow}{\colrarrow}{\colrarrow}{}{#2};}{%
		\IfEqCase{#1}{%
			{0}{\wtilelab{\colrarrow}{\colrarrow}{\colrarrow}{}{#2};}
			{1}{\wtilelab{\coldarrow}{}{\coldarrow}{\coldarrow}{#2};}
			{2}{\wtilelab{\collarrow}{\collarrow}{}{\collarrow}{#2};}
			{3}{\wtilelab{}{\coluarrow}{\coluarrow}{\coluarrow}{#2};}
		}
	}
}
\newcommand{\wpatchLru}[2][]{
	\ifthenelse{\equal{#1}{}}{
		\wtilelab{\colrarrow}{}{}{\coldarrow}{#2};
		\begin{scope}[shift={#2}]
			\draw (-0.5\tilesize,-1.3\tilesize)
				-- (-0.5\tilesize,-0.7\tilesize) -- (0.7\tilesize,0.5\tilesize)
				-- (1.3\tilesize,0.5\tilesize);
		\end{scope}
	}{%
		\IfEqCase{#1}{%
			{0}{%
				\wtilelab{\colrarrow}{}{}{\coldarrow}{#2};
				\begin{scope}[shift={#2}]
					\draw (-0.5\tilesize,-1.3\tilesize)
						-- (-0.5\tilesize,-0.7\tilesize) -- (0.7\tilesize,0.5\tilesize)
						-- (1.3\tilesize,0.5\tilesize);
				\end{scope}
			}
			{1}{%
				\wtilelab{}{\collarrow}{}{\coldarrow}{#2};
				\begin{scope}[shift={#2}]
					\draw (-1.3\tilesize,0.5\tilesize)
						-- (-0.7\tilesize,0.5\tilesize) -- (0.5\tilesize,-0.7\tilesize)
						-- (0.5\tilesize,-1.3\tilesize);
				\end{scope}
			}
			{2}{%
				\wtilelab{}{\collarrow}{\coluarrow}{}{#2};
				\begin{scope}[shift={#2}]
					\draw (0.5\tilesize,1.3\tilesize)
						-- (0.5\tilesize,0.7\tilesize) -- (-0.7\tilesize,-0.5\tilesize)
						-- (-1.3\tilesize,-0.5\tilesize);
				\end{scope}
			}
			{3}{%
				\wtilelab{\colrarrow}{}{\coluarrow}{}{#2};
				\begin{scope}[shift={#2}]
					\draw (1.3\tilesize,-0.5\tilesize)
						-- (0.7\tilesize,-0.5\tilesize) -- (-0.5\tilesize,0.7\tilesize)
						-- (-0.5\tilesize,1.3\tilesize);
				\end{scope}
			}
		}
	}	
}
\newcommand{\wpatchLlc}[3][]{
	\wtilelab{}{}{\colrarrow}{\colrarrow}{#2};
	\begin{scope}[shift={#2}]
		\ifthenelse{\equal{#3}{}}{%
		}{%
			\path (\tilesize,-\tilesize) -- node[right=-2pt] {$(#3,\bullet)$} (\tilesize,\tilesize);
		}
	\end{scope}
}
\newcommand{\wpatchLcc}[3][]{
	\wtilelab{}{}{\colrarrow}{\colrarrow}{#2};
	\begin{scope}[shift={#2}]
		\ifthenelse{\equal{#3}{}}{%
			\path (-\tilesize,-\tilesize) -- node {$\bullet$} (-\tilesize,\tilesize);
		}{%
			\path (-\tilesize,-\tilesize) -- node[left=-2pt] {$(#3,\bullet)$} (-\tilesize,\tilesize);
			\path (\tilesize,-\tilesize) -- node[right=-2pt] {$(#3,\bullet)$} (\tilesize,\tilesize);
		}
	\end{scope}
}
\newcommand{\wpatchLrc}[3][]{
	\wtilelab{}{}{}{}{#2};
	\begin{scope}[shift={#2}]
		\ifthenelse{\equal{#3}{}}{%
			\path (-\tilesize,-\tilesize) -- node {$\bullet$} (-\tilesize,\tilesize);
			\path (\tilesize,-\tilesize) -- node[right=-2pt] {$#3$} (\tilesize,\tilesize);
		}{%
			\path (-\tilesize,-\tilesize) -- node[left=-2pt] {$\bullet$} (-\tilesize,\tilesize);
			\path (\tilesize,-\tilesize) -- node[right=-2pt] {$#3$} (\tilesize,\tilesize);
		}
		\draw (-0.5\tilesize,-1.3\tilesize) -- (-0.5\tilesize,1.3\tilesize);
	\end{scope}
}
\newcommand{\wpatchLld}[2][]{
	\wtilelab{}{\colrarrow}{}{\colrarrow}{#2};
}
\newcommand{\wpatchLcd}[2][]{
	\wtilelab{\colrarrow}{\colrarrow}{}{\colrarrow}{#2};
}
\newcommand{\wpatchLrd}[2][]{
	\wtilelab{\colrarrow}{}{\coluarrow}{}{#2};
	\begin{scope}[shift={#2}]
		\draw (-0.5\tilesize,1.3\tilesize)
			-- (-0.5\tilesize,0.7\tilesize) -- (0.7\tilesize,-0.5\tilesize)
			-- (1.3\tilesize,-0.5\tilesize);
	\end{scope}
}

\newcommand{\globalpatchinginput}{%
	\fill[pattern=north east lines, pattern color=blue] (-6,-6) rectangle (10,10);
	\fill[white] (0,0) rectangle (4,4);
	
	\draw[very thick] (-6,-6) rectangle (10,10);
	\draw[very thick] (0,0) rectangle (4,4);
	\draw[dashed] (-3,-3) rectangle (7,7);
	\draw[dashed] (-4,-4) rectangle (8,8);
	
	\node at (2,2) {$S_n$};
	
	\begin{scope}[overlay,>=stealth',shorten <=1pt]
		\draw[<->,thin] (-6.5,-6) -- node[left=1.5ex,anchor=center,rotate=90]
			{\scriptsize $\beta(n)$} (-6.5,-4);
		\draw[<->,thin] (-6.5,-4) -- node[left=1.5ex,anchor=center,rotate=90]
			{\scriptsize $1$} (-6.5,-3);
		\draw[<->,thin] (-6.5,-3) -- node[left=1.5ex,anchor=center,rotate=90]
			{\scriptsize $\alpha(n)$} (-6.5,0);
		\draw[<->,thin] (-6.5,0) -- node[left=1.5ex,anchor=center,rotate=90]
			{\scriptsize $n$} (-6.5,4);
		
		\draw[<->,thin] (-6,-6.5) -- node[below] {\scriptsize $\beta(n)$} (-4,-6.5);
		\draw[<->,thin] (-4,-6.5) -- node[below] {\scriptsize $1$} (-3,-6.5);
		\draw[<->,thin] (-3,-6.5) -- node[below] {\scriptsize $\alpha(n)$} (0,-6.5);
		\draw[<->,thin] (0,-6.5) -- node[below] {\scriptsize $n$} (4,-6.5);
	\end{scope}
}

\newcommand{\globalpatchingoutput}{%
	\fill[pattern=north east lines, pattern color=blue] (-4,-4) rectangle (8,8);
	\fill[white] (-3,-3) rectangle (7,7);
	\fill[pattern=north west lines, pattern color=red] (-3,-3) rectangle (7,7);
	
	\draw[dashed] (0,0) rectangle (4,4);
	\draw[dashed] (-3,-3) rectangle (7,7);
	\draw[very thick] (-4,-4) rectangle (8,8);
	
	\node at (2,2) {$S_n$};
	
	\begin{scope}[overlay,>=stealth',shorten <=1pt]
		\draw[<->,thin] (-4.5,-4) -- node[left=1.5ex,anchor=center,rotate=90]
			{\scriptsize $1$} (-4.5,-3);
		\draw[<->,thin] (-4.5,-3) -- node[left=1.5ex,anchor=center,rotate=90]
			{\scriptsize $\alpha(n)$} (-4.5,0);
		\draw[<->,thin] (-4.5,0) -- node[left=1.5ex,anchor=center,rotate=90]
			{\scriptsize $n$} (-4.5,4);
		
		\draw[<->,thin] (-4,-4.5) -- node[below] {\scriptsize $1$} (-3,-4.5);
		\draw[<->,thin] (-3,-4.5) -- node[below] {\scriptsize $\alpha(n)$} (0,-4.5);
		\draw[<->,thin] (0,-4.5) -- node[below] {\scriptsize $n$} (4,-4.5);
	\end{scope}
}

\newcommand\nn{5}
\newcommand\nnl{4}
\newcommand{\squaretilinginstance}{%
	\draw (-\tilesize,-\tilesize) rectangle +(2*\nn*\tilesize,2*\nn*\tilesize);
	\foreach \j in {1,...,\nnl} {%
		\draw ($(-\tilesize,-\tilesize)+(-0.3\tilesize,2*\j*\tilesize)$) -- +(0.6\tilesize,0);
		\draw ($(2*\nn*\tilesize,0)+(-\tilesize,-\tilesize)+(-0.3\tilesize,2*\j*\tilesize)$) -- +(0.6\tilesize,0);
	}
	\foreach \i in {1,...,\nnl} {%
		\draw ($(-\tilesize,-\tilesize)+(2*\i*\tilesize,-0.3\tilesize)$) -- +(0,0.6\tilesize);
		\draw ($(0,2*\nn*\tilesize)+(-\tilesize,-\tilesize)+(2*\i*\tilesize,-0.3\tilesize)$) -- +(0,0.6\tilesize);
	}
	\foreach \j/\jlab in {0/1,1/2,\nnl/n} {%
		\path ($(-\tilesize,-\tilesize)+(0,2*\j*\tilesize)$) -- node[left=-2pt] {$a_{\jlab}$} +(0,2\tilesize);
	}
	\foreach \j/\jlab in {0/1,1/2,\nnl/n} {%
		\path ($(2*\nnl*\tilesize,2*\nnl*\tilesize)+(\tilesize,\tilesize)+(0,-2*\j*\tilesize)$) -- node[right=-2pt] {$c_{\jlab}$} +(0,-2\tilesize);
	}
	\foreach \i/\ilab in {0/1,1/2,\nnl/n} {%
		\path ($(0,2*\nnl*\tilesize)+(-\tilesize,\tilesize)+(2*\i*\tilesize,0)$) -- node[above=-2pt] {$b_{\ilab}$} +(2\tilesize,0);
	}
	\foreach \i/\ilab in {0/1,1/2,\nnl/n} {%
		\path ($(2*\nnl*\tilesize,0)+(\tilesize,-\tilesize)+(-2*\i*\tilesize,0)$) -- node[below=-2pt] {$d_{\ilab}$} +(-2\tilesize,0);
	}
}

\newcommand{\squaretilinginstancetransformed}{%
	
	\fill[blue!10] ($(-\tilesize,-\tilesize)+(-4,-4)$) rectangle ($(\tilesize,\tilesize)+(8,8)$);
	\fill[white] ($(-\tilesize,-\tilesize)+(-3,-3)$) rectangle ($(\tilesize,\tilesize)+(7,7)$);
	
	\fill[red!10] ($(-\tilesize,-\tilesize)+(-1,-1)$) rectangle ($(\tilesize,\tilesize)+(5,5)$);
	\fill[white] ($(-\tilesize,-\tilesize)+(0,0)$) rectangle ($(\tilesize,\tilesize)+(4,4)$);

	\foreach \j/\jlab in {0/1,1/2,\nnl/n} {%
		\path ($(-\tilesize,-\tilesize)+(0,2*\j*\tilesize)$) -- node[left=-2pt] {$a_{\jlab}$} +(0,2\tilesize);
	}
	\foreach \j/\jlab in {0/1,1/2,\nnl/n} {%
		\path ($(2*\nnl*\tilesize,2*\nnl*\tilesize)+(\tilesize,\tilesize)+(0,-2*\j*\tilesize)$) -- node[right=-2pt] {$c_{\jlab}$} +(0,-2\tilesize);
	}
	\foreach \i/\ilab in {0/1,1/2,\nnl/n} {%
		\path ($(0,2*\nnl*\tilesize)+(-\tilesize,\tilesize)+(2*\i*\tilesize,0)$) -- node[above=-2pt] {$b_{\ilab}$} +(2\tilesize,0);
	}
	\foreach \i/\ilab in {0/1,1/2,\nnl/n} {%
		\path ($(2*\nnl*\tilesize,0)+(\tilesize,-\tilesize)+(-2*\i*\tilesize,0)$) -- node[below=-2pt] {$d_{\ilab}$} +(-2\tilesize,0);
	}	
	
	\foreach \i in {-7,...,0,5,6,...,12} {%
		\draw[help lines] ($(-\tilesize,-\tilesize)+(2*\i*\tilesize,-14\tilesize)$)
			-- ($(-\tilesize,\tilesize)+(2*\i*\tilesize,22\tilesize)$);
	}
	\foreach \i in {1,...,4} {%
		\draw[help lines] ($(-\tilesize,-\tilesize)+(2*\i*\tilesize,-14\tilesize)$)
			-- ($(-\tilesize,\tilesize)+(2*\i*\tilesize,-2\tilesize)$);
		\draw[help lines] ($(-\tilesize,-\tilesize)+(2*\i*\tilesize,10\tilesize)$)
			-- ($(-\tilesize,\tilesize)+(2*\i*\tilesize,22\tilesize)$);
	}
	\foreach \j in {-7,...,0,5,6,...,12} {%
		\draw[help lines] ($(-\tilesize,-\tilesize)+(-14\tilesize,2*\j*\tilesize)$)
			-- ($(\tilesize,-\tilesize)+(22\tilesize,2*\j*\tilesize)$);
	}
	\foreach \j in {1,...,4} {%
		\draw[help lines] ($(-\tilesize,-\tilesize)+(-14\tilesize,2*\j*\tilesize)$)
			-- ($(\tilesize,-\tilesize)+(-2\tilesize,2*\j*\tilesize)$);
		\draw[help lines] ($(-\tilesize,-\tilesize)+(10\tilesize,2*\j*\tilesize)$)
			-- ($(\tilesize,-\tilesize)+(22\tilesize,2*\j*\tilesize)$);
	}

	\begin{scope}[overlay,>=stealth']
		\draw[<->,thin] ($(-\tilesize,-\tilesize)+(-7,-7.5)$) -- node[below] {$\beta(n)$} ($(\tilesize,-\tilesize)+(-5,-7.5)$);
		\draw[<->,thin] ($(-\tilesize,-\tilesize)+(-3,-7.5)$) -- node[below] {$\alpha(n)$} ($(\tilesize,-\tilesize)+(-1,-7.5)$);
		\draw[<->,thin] ($(-\tilesize,-\tilesize)+(-7.5,-7)$) -- node[left=1.5ex,anchor=center,rotate=90] {$\beta(n)$} ($(-\tilesize,\tilesize)+(-7.5,-5)$);
		\draw[<->,thin] ($(-\tilesize,-\tilesize)+(-7.5,-3)$) -- node[left=1.5ex,anchor=center,rotate=90] {$\alpha(n)$} ($(-\tilesize,\tilesize)+(-7.5,-1)$);
	\end{scope}
	
	\foreach \i in {-7,...,-2} {
		\path ($(2*\i*\tilesize,-2\tilesize)$) -- node {\colrarrow} +(2\tilesize,0);
		\path ($(2*\i*\tilesize,10\tilesize)$) -- node {\colrarrow} +(2\tilesize,0);
		\foreach \j in {0,...,5} {
			\path ($(-\tilesize,-\tilesize)+(2*\i*\tilesize,2*\j*\tilesize)$) -- node {\colrarrow} +(2\tilesize,0);
		}
		\foreach \j in {0,...,4} {
			\fill ($(\tilesize,0)+(2*\i*\tilesize,2*\j*\tilesize)$) circle (2pt);
		}
	}
	
	\foreach \i in {6,...,11} {
		\path ($(2*\i*\tilesize,-2\tilesize)$) -- node {\collarrow} +(-2\tilesize,0);
		\path ($(2*\i*\tilesize,10\tilesize)$) -- node {\collarrow} +(-2\tilesize,0);
		\foreach \j in {0,...,5} {
			\path ($(\tilesize,-\tilesize)+(2*\i*\tilesize,2*\j*\tilesize)$) -- node {\collarrow} +(-2\tilesize,0);
		}
		\foreach \j in {0,...,4} {
			\fill ($(-\tilesize,0)+(2*\i*\tilesize,2*\j*\tilesize)$) circle (2pt);
		}
	}

	\foreach \j in {-7,...,-2} {
		\path ($(-2\tilesize,2*\j*\tilesize)$) -- node {\coluarrow} +(0,2\tilesize);
		\path ($(10\tilesize,2*\j*\tilesize)$) -- node {\coluarrow} +(0,2\tilesize);
		\foreach \i in {0,...,5} {
			\path ($(-\tilesize,-\tilesize)+(2*\i*\tilesize,2*\j*\tilesize)$) -- node {\coluarrow} +(0,2\tilesize);
		}
		\foreach \i in {0,...,4} {
			\fill ($(0,\tilesize)+(2*\i*\tilesize,2*\j*\tilesize)$) circle (2pt);
		}
	}
	
	\foreach \j in {6,...,11} {
		\path ($(-2\tilesize,2*\j*\tilesize)$) -- node {\coldarrow} +(0,-2\tilesize);
		\path ($(10\tilesize,2*\j*\tilesize)$) -- node {\coldarrow} +(0,-2\tilesize);
		\foreach \i in {0,...,5} {
			\path ($(-\tilesize,\tilesize)+(2*\i*\tilesize,2*\j*\tilesize)$) -- node {\coldarrow} +(0,-2\tilesize);
		}
		\foreach \i in {0,...,4} {
			\fill ($(0,-\tilesize)+(2*\i*\tilesize,2*\j*\tilesize)$) circle (2pt);
		}
	}

	\draw (-2.5\tilesize,-1.3\tilesize)
		-- (-2.5\tilesize,9.3\tilesize) -- (-1.3\tilesize,10.5\tilesize)
		-- (9.3\tilesize,10.5\tilesize) -- (10.5\tilesize,9.3\tilesize)
		-- (10.5\tilesize,-1.3\tilesize) -- (9.3\tilesize,-2.5\tilesize)
		-- (-1.3\tilesize,-2.5\tilesize) -- cycle;
}

\newdimen\endpointsize
\newcommand{\wtpathsR}[1]{%
	\begin{scope}[shift={#1}]
		\draw[very thin] (-\tilesize,-\tilesize) rectangle (\tilesize,\tilesize);%
		\draw[very thick] (0,0) -- (\tilesize,0);
		\fill (0,0) circle (\endpointsize);
	\end{scope}
}
\newcommand{\wtpathsU}[1]{%
	\begin{scope}[shift={#1}]
		\draw[very thin] (-\tilesize,-\tilesize) rectangle (\tilesize,\tilesize);%
		\draw[very thick] (0,0) -- (0,\tilesize);
		\fill (0,0) circle (\endpointsize);
	\end{scope}
}
\newcommand{\wtpathsL}[1]{%
	\begin{scope}[shift={#1}]
		\draw[very thin] (-\tilesize,-\tilesize) rectangle (\tilesize,\tilesize);%
		\draw[very thick] (0,0) -- (-\tilesize,0);
		\fill (0,0) circle (\endpointsize);
	\end{scope}
}
\newcommand{\wtpathsD}[1]{%
	\begin{scope}[shift={#1}]
		\draw[very thin] (-\tilesize,-\tilesize) rectangle (\tilesize,\tilesize);%
		\draw[very thick] (0,0) -- (0,-\tilesize);
		\fill (0,0) circle (\endpointsize);
	\end{scope}
}
\newcommand{\wtpathsRU}[1]{%
	\begin{scope}[shift={#1}]
		\draw[very thin] (-\tilesize,-\tilesize) rectangle (\tilesize,\tilesize);%
		\draw[very thick,rounded corners=1pt] (\tilesize,0) -- (0,0) -- (0,\tilesize);
	\end{scope}
}
\newcommand{\wtpathsUL}[1]{%
	\begin{scope}[shift={#1}]
		\draw[very thin] (-\tilesize,-\tilesize) rectangle (\tilesize,\tilesize);%
		\draw[very thick,rounded corners=1pt] (0,\tilesize) -- (0,0) -- (-\tilesize,0);
	\end{scope}
}
\newcommand{\wtpathsLD}[1]{%
	\begin{scope}[shift={#1}]
		\draw[very thin] (-\tilesize,-\tilesize) rectangle (\tilesize,\tilesize);%
		\draw[very thick,rounded corners=1pt] (-\tilesize,0) -- (0,0) -- (0,-\tilesize);
	\end{scope}
}
\newcommand{\wtpathsDR}[1]{%
	\begin{scope}[shift={#1}]
		\draw[very thin] (-\tilesize,-\tilesize) rectangle (\tilesize,\tilesize);%
		\draw[very thick,rounded corners=1pt] (0,-\tilesize) -- (0,0) -- (\tilesize,0);
	\end{scope}
}
\newcommand{\wtpathsRL}[1]{%
	\begin{scope}[shift={#1}]
		\draw[very thin] (-\tilesize,-\tilesize) rectangle (\tilesize,\tilesize);%
		\draw[very thick,rounded corners=1pt] (\tilesize,0) -- (-\tilesize,0);
	\end{scope}
}
\newcommand{\wtpathsUD}[1]{%
	\begin{scope}[shift={#1}]
		\draw[very thin] (-\tilesize,-\tilesize) rectangle (\tilesize,\tilesize);%
		\draw[very thick,rounded corners=1pt] (0,\tilesize) -- (0,-\tilesize);
	\end{scope}
}
\newcommand{\wtpathsX}[1]{%
	\begin{scope}[shift={#1}]
		\draw[very thin] (-\tilesize,-\tilesize) rectangle (\tilesize,\tilesize);%
		\draw[very thick,rounded corners=1pt] (0,\tilesize) -- (0,-\tilesize);
		\draw[very thick,rounded corners=1pt] (\tilesize,0) -- (-\tilesize,0);
	\end{scope}
}
\newcommand{\wtpaths}[2]{%
	\IfEqCase{#1}{%
		{0}{%
			\wtpathsR{#2}
		}
		{1}{%
			\wtpathsU{#2}
		}
		{2}{%
			\wtpathsL{#2}
		}
		{3}{%
			\wtpathsD{#2}
		}
		{4}{%
			\wtpathsRU{#2}
		}
		{5}{%
			\wtpathsUL{#2}
		}
		{6}{%
			\wtpathsLD{#2}
		}
		{7}{%
			\wtpathsDR{#2}
		}
		{8}{%
			\wtpathsRL{#2}
		}
		{9}{%
			\wtpathsUD{#2}
		}
		{10}{%
			\wtpathsX{#2}
		}
	}
}

\def\pathsexampleconfig{{%
	{ 9,  7,  8,  6,  0,  2,  9,  7,  2,  9},
	{ 4, 10,  6,  9,  7,  2,  9,  9,  7,  5},
	{ 6,  1,  9,  9,  9,  0, 10,  5,  9,  7},
	{ 4,  2,  9,  9,  9,  7, 10,  8,  5,  4},
	{ 7,  8,  5,  4,  5,  1,  9,  7,  8,  6},
	{ 9,  7,  8,  8,  8,  8,  5,  9,  7,  5},
	{ 5,  9,  7,  8,  8,  2,  3,  4,  5,  3}
}}
\newcommand{\pathsexample}{%
	\foreach \j in {0,...,6} {
		\foreach \i in {0,...,9} {
			\pgfmathparse{\pathsexampleconfig[\j][\i]}
			\wtpaths{\pgfmathresult}{(\i,-\j)}
		}
	}
}

\newcommand{\kca}{red}
\newcommand{\kcb}{green}
\newcommand{\kcc}{blue}
\newcommand{\kcd}{yellow}
\newcommand{\kce}{gray}

\newcommand{\kctile}[2]{
	\IfEqCase{#1}{%
		{1}{%
			\wtile{\kca}{\kcb}{\kcc}{\kcb}{#2}
		}
		{2}{%
			\wtile{\kca}{\kcc}{\kcb}{\kcb}{#2}
		}
		{3}{%
			\wtile{\kcb}{\kcc}{\kcc}{\kcb}{#2}
		}
		{4}{%
			\wtile{\kcb}{\kca}{\kcb}{\kca}{#2}
		}
		{5}{%
			\wtile{\kcc}{\kca}{\kcc}{\kca}{#2}
		}
		{6}{%
			\wtile{\kcc}{\kcb}{\kcb}{\kca}{#2}
		}
		{7}{%
			\wtile{\kcd}{\kcd}{\kca}{\kcd}{#2}
		}
		{8}{%
			\wtile{\kcd}{\kcd}{\kcb}{\kcc}{#2}
		}
		{9}{%
			\wtile{\kcd}{\kce}{\kca}{\kcb}{#2}
		}
		{10}{%
			\wtile{\kcd}{\kce}{\kcd}{\kcb}{#2}
		}
		{11}{%
			\wtile{\kce}{\kce}{\kca}{\kcd}{#2}
		}
		{12}{%
			\wtile{\kce}{\kce}{\kcb}{\kcc}{#2}
		}
		{13}{%
			\wtile{\kce}{\kcd}{\kcb}{\kcb}{#2}
		}
	}
}

\newcommand{\ama}{red}
\newcommand{\amb}{green}
\newcommand{\amc}{blue!30}
\newcommand{\amd}{yellow}
\newcommand{\ame}{gray}
\newcommand{\amf}{blue!70!red}

\newcommand{\amtile}[2]{
	\IfEqCase{#1}{%
		{1}{
			\wtile{\ama}{\amb}{\amb}{\ama}{#2}
		}
		{2}{%
			\wtile{\amc}{\amd}{\amd}{\amc}{#2}
		}
		{3}{%
			\wtile{\amd}{\ame}{\ame}{\amd}{#2}
		}
		{4}{%
			\wtile{\amf}{\amc}{\amc}{\amf}{#2}
		}
		{5}{%
			\wtile{\amc}{\ame}{\amd}{\amd}{#2}
		}
		{6}{%
			\wtile{\amc}{\amc}{\amd}{\amf}{#2}
		}
		{7}{%
			\wtile{\amd}{\amd}{\ame}{\amc}{#2}
		}
		{8}{%
			\wtile{\amf}{\amd}{\amc}{\amc}{#2}
		}
		{9}{%
			\wtile{\amb}{\ama}{\amc}{\ame}{#2}
		}
		{10}{%
			\wtile{\amb}{\ama}{\amf}{\amd}{#2}
		}
		{11}{%
			\wtile{\ama}{\ama}{\amd}{\ame}{#2}
		}
		{12}{%
			\wtile{\amb}{\amb}{\amf}{\amc}{#2}
		}
		{13}{%
			\wtile{\amd}{\amf}{\ama}{\amb}{#2}
		}
		{14}{%
			\wtile{\ame}{\amc}{\ama}{\amb}{#2}
		}
		{15}{%
			\wtile{\amc}{\amf}{\amb}{\amb}{#2}
		}
		{16}{%
			\wtile{\ame}{\amd}{\ama}{\ama}{#2}
		}
	}
}

\newcommand{\xxamtile}[2]{
	\begin{scope}[shift={#2},rotate=-90]		
		\amtile{#1}{(0,0)}
	\end{scope}
}

\newdimen\blanksize
\newdimen\ltcellsize
\def\ledrappiertoomconfig{{%
	{0, 0, 0, 0, 0, 0, 0, 0, 0},
	{0, 0, 0, 0, 0, 0, 0, 0, 0},
	{0, 0, 0, 3, 1, 1, 2, 0, 0},
	{0, 0, 0, 0, 1, 0, 1, 0, 0},
	{0, 0, 0, 0, 0, 1, 1, 0, 0},
	{0, 0, 0, 0, 0, 0, 3, 0, 0},
	{0, 0, 0, 0, 0, 0, 0, 0, 0},
	{0, 0, 0, 0, 0, 0, 0, 0, 0},
	{0, 0, 0, 0, 0, 0, 0, 0, 0}
}}
\newcommand{\ledrappiertoom}{%
	\foreach \j in {0,...,8} {%
		\foreach \i in {0,...,8} {
			\pgfmathparse{int(\ledrappiertoomconfig[\j][\i])}
			\edef\x{\pgfmathresult}
			\IfEqCase{\x}{%
				{0}{%
					\node at (\i,\j) {$\symb{0}$};
				}
				{1}{%
					\fill[gray!50] (\i,\j) +(-\ltcellsize,-\ltcellsize) rectangle +(\ltcellsize,\ltcellsize);%
					\node at (\i,\j) {$\symb{1}$};
				}%
				{2}{%
					\fill[red!70] (\i,\j) +(-\ltcellsize,-\ltcellsize) rectangle +(\ltcellsize,\ltcellsize);%
					\node at (\i,\j) {$\symb{1}$};
				}%
				{3}{%
					\fill[yellow] (\i,\j) +(-\ltcellsize,-\ltcellsize) rectangle +(\ltcellsize,\ltcellsize);%
					\node at (\i,\j) {$\symb{1}$};
				}%
			}
		}
	}
}

\newcommand{\tikzmark}[1]{\tikz[remember picture] \node[coordinate] (#1) {#1};}

 \usepackage{hyperref}

\begin{document}


\setcounter{page}{27}
\publyear{22}
\papernumber{2103}
\volume{185}
\issue{1}

    \finalVersionForARXIV


\title{Self-stabilisation of Cellular Automata on Tilings}

\author{Nazim Fat\`es\\
	Universit\'e de Lorraine\\
    CNRS, Inria, LORIA\\
	F-54000 Nancy, France\\
	nazim.fates@loria.fr
	\and
	Ir\`ene Marcovici\\
	Universit\'e de Lorraine\\
    CNRS, Inria, IECL\\
	F-54000 Nancy, France\\
	irene.marcovici@univ-lorraine.fr
	\and
	Siamak Taati\thanks{The work of ST was partially supported by NWO grant 612.001.409.}\thanks{Address for
                 correspondence: Department of  Mathematics, 	American University of Beirut,
	                                Beirut, Lebanon. \newline \newline
          \vspace*{-6mm}{\scriptsize{Received February 2021; \ accepted January 2022.}}}
       \\
	Department of Mathematics\\
 	American University of Beirut\\
	Beirut, Lebanon\\
	siamak.taati@gmail.com
}

\maketitle

\runninghead{N.~Fat\`es et al.}{Self-stabilisation of CA on Tilings}

\vspace*{-6mm}
\begin{abstract}
	Given a finite set of local constraints, we seek a cellular automaton (i.e., a local and uniform algorithm) that self-stabilises on the configurations that satisfy these constraints.
	More precisely, starting from a finite perturbation of a valid configuration, the cellular automaton must eventually fall back into the space of valid configurations where it remains still.
	We allow the cellular automaton to use extra symbols, but in that case, the extra symbols can also appear in the initial finite perturbation.
	For several classes of local constraints (e.g., $k$-colourings with $k\neq 3$, and North-East deterministic constraints), we provide efficient self-stabilising cellular automata with or without additional symbols that wash out finite perturbations in linear or quadratic time, but also show that there are examples of local constraints for which the self-stabilisation problem is inherently hard.
	We note that the optimal self-stabilisation speed is the same for all local constraints that are isomorphic to one another.
	We also consider probabilistic cellular automata rules and show that in some cases, the use of randomness simplifies the problem.
	In the deterministic case, we show that if finite perturbations are corrected in linear time, then the cellular automaton self-stabilises even starting from a random perturbation of a valid configuration, that is, when errors in the initial configuration occur independently with a sufficiently low density.

\end{abstract}


\section*{Introduction}
\label{sec:intro}
\addcontentsline{toc}{section}{Introduction}

While all living organisms possess some ability to stabilise or repair themselves when subjected to perturbations or attacks, artificial systems rarely have such an ability.
In particular, in systems designed in engineering and computer science, a small local perturbation (e.g., due to noise or tampering by an adversary) can often propagate throughout the system
leading to a total or partial devastation of the behaviour of the system.
The inevitability of such perturbations has lead to the study of systems which, in addition to their normal functionality, have the \emph{self-stabilisation} property.
A self-stabilising system has the capacity to re-enter a set of ``legal'' or ``desirable'' states once the system has been taken out of its normal behaviour by an external perturbation.

The concept of self-stabilisation in computational processes was first introduced in 1970s by Dijkstra, who presented examples of networks of finite-state automata with a non-trivial self-stabilisation property~\cite{Dij74}.  Since then, self-stabilisation has been widely studied in the context of distributed computing (see e.g.~\cite{Dol00,AlDeDuPe19}).  In the current paper, we explore the question of self-stabilisation in the context of cellular automata.

In a cellular automaton (CA), the components of the system, the \emph{cells}, are arranged regularly on an infinite $d$-dimensional lattice.  The cells are identical finite-state automata that interact locally and change their states synchronously.  The overall state of the cells is referred to as a \emph{configuration} of the CA.
We will assume that the set of legal configurations of the system is specified with a finite number of local constraints, which must be satisfied at every position.  As a prototypical example, one may consider the \emph{colouring} constraints: each cell can have any of a finite number of colours, and the legal configurations are those in which every two adjacent cells have different colours.
More generally, we think of the legal states as \emph{tilings} of the lattice with a finite number of tile types (identified with the states of the cells) satisfying local matching constraints.
In the language of symbolic dynamics, the set of legal configurations is simply a shift space of finite type (SFT).
We will clarify the terminology further in the following section.

\medskip
We require our CA to have the following form of self-stabilisation:
\begin{enumerate}[label={(\arabic*)}]
\item Starting from a configuration that deviates from a legal configuration only on a finite region, the CA must evolve back to a legal configuration in a finite number of steps.
\item Starting from a legal configuration, the CA must remain unchanged.
\end{enumerate}
We note that in our definition, the CA has no functionality other than to keep the constraints satisfied.  Depending on the local constraints, even this simplified notion of self-stabilisation can be quite challenging to achieve.
The difficulty is that the cells are indistinguishable and the information available to each cell is limited to the state of its close neighbours.  Since the cells are not aware of their own absolute position or the position and extent of the error region, it is thus for example not possible to correct the error region by starting from the upper-left corner and then proceeding sequentially.

Self-stabilisation can be understood as a weak form of fault-tolerance, in which the perturbations occur only at the beginning.  Stronger forms of fault-tolerance have been studied in the setting of cellular automata, although aside from a few strong proof-of-concept constructions, the field remains wide open.
Around the same time as Dijkstra, Toom found a class of CA which self-stabilise even in presence of sufficiently weak \emph{temporal} noise~\cite{Too74,Too80} (see Examples~\ref{exp:toom} and~\ref{exp:toom:noise} below).  G\'acs and Reif exploited Toom's simplest example (i.e., the NEC-majority rule) to construct a three-dimensional CA which, in presence of noise, can perform universal computation reliably~\cite{GaRe88}.  Subsequently, G\'acs was able to construct a sophisticated one-dimensional CA capable of reliable universal computation~\cite{Gac86,Gac01}.
In our setting, Toom's NEC-majority CA solves the self-stabilisation problem for the constraint that adjacent cells must have the same colour.

Various other problems studied in the setting of cellular automata can be related to self-stabilisation.
For instance, the \emph{density classification} problem~\cite{BuFaMaMa13,Oli14} can be formulated as a problem of self-stabilisation where the system needs to return to a homogeneous configuration (all-zero or all-one) with the additional requirement that the colour which appears less frequently in the initial configuration is the one which has to be wiped off.  It is known that Toom's NEC-majority CA solves this problem, at least when the system starts from a biased Bernoulli random configuration~\cite{BuFaMaMa13}.
Another example is the \emph{global synchronisation} problem, which can again be understood as a self-stabilisation problem with the homogeneous configurations as the legal states, with the requirement that, in its legal state, the system must oscillate rather than remain unchanged~\cite{Ric17,Fat19}.

\medskip
This article has grown out of a conference paper in which some of our results were presented~\cite{FaMaTa19}.
The scope of the current paper is however more general and contains various new results. The structure of the paper is as follows:
\begin{itemize}
\itemsep=0.95pt
\item In Section~\ref{sec:setting}, we introduce the terminology and notation.
\item In Section~\ref{sec:1d}, we present a general construction for self-stabilising one-dimensional SFTs.
\item As is the case for many other problems regarding cellular automata and tilings, the self-stabilisation problem in two and higher dimensions is significantly more complex than in one dimension.
Section~\ref{sec:2d} is dedicated to the two-dimensional case, where we present several constructions of self-stabilising CA depending on the type of the constraints.
Here, the example of $k$-colourings serves as a running example, as it has different levels of difficulty depending on the parameter~$k$.  While the cases of $k=2$ and $k\geq 5$ admit relatively simple solutions that self-stabilise in linear time, our solution for $k=4$ self-stabilises in quadratic time, and we could not find any solution whatsoever for the case $k=3$.
We also provide a linear-time solution for the case of deterministic SFTs.  Deterministic SFTs encompass a relatively rich family of constraints, including some which admit only non-periodic legal configurations.
\item In Section~\ref{sec:PCA}, we investigate the self-stabilisation of probabilistic CA, and show that, in some cases, access to randomness simplifies the self-stabilisation problem.  For instance, our probabilistic solution for $k$-colourings with $k\geq 5$ works in logarithmic time rather than linear time, and furthermore, has the same symmetries as the colouring constraints.
\item After having explored the ``algorithmic'' aspects of self-stabilisation, we turn to the question of ``complexity'' in Section~\ref{sec:complexity}.  We show that for some choices of the legal constraints, the self-stabilisation problem is inherently hard (i.e., requires super-polynomial stabilisation time, unless $\classP=\classNP$).  We also show that ``isomorphic'' constraints (i.e., isomorphic SFTs) admit solutions with roughly the same stabilisation times.
\item Section~\ref{sec:random-noise} is about a different notion of self-stabilisation in which the initial perturbations are random rather than finite.  We show that if a (deterministic) CA self-stabilises from finite perturbations in linear time, then it also self-stabilises from sufficiently weak Bernoulli random perturbations.  The more interesting question of self-stabilisation in presence of temporal noise (as in the case of Toom's CA) is left open.
\item The article ends with some remarks and open questions in Section~\ref{sec:discussion}.
\end{itemize}

\section{Terminology and notations}
\label{sec:setting}

\subsection{Tilings}

\paragraph{Configurations and patterns.}
Let $\Sigma$ be an \emph{alphabet}, that is, a finite set of \emph{symbols}, and let $d\geq 1$.
An assignment $x\colon\ZZ^d\to\Sigma$ is called a \emph{configuration} of the \emph{lattice}~$\ZZ^d$.
The symbol $x_i$ is the \emph{state} (or \emph{colour}) of \emph{cell}~$i$.
A configuration is said to be \emph{homogeneous} if all the cells are in the same state. Given $c\in\Sigma$, we denote by $\unif{c}$ the homogeneous configuration in which all cells have symbol~$c$.

The restriction of the configuration $x\in\Sigma^{\ZZ^d}$ to a set $A\subseteq\ZZ^d$ is denoted by $x_A$.  A \emph{pattern} is an assignment $p\colon A\to\Sigma$ with finite \emph{shape} $A\subseteq\ZZ^d$, i.e., a partial configuration with a finite domain. We denote by $\Sigma^\#$ the set of all patterns.

A sequence $x^{(1)}, x^{(2)}, \ldots$ of configurations is said to \emph{converge} to another configuration $x$ if the state of each cell in~$x^{(n)}$ eventually fixates at the value of the same cell in~$x$, that is, if for every $i\in\ZZ^d$ there exists an~$n_i$ such that $x^{(n)}_i=x_i$ for all $n\geq n_i$.
This is the notion of convergence in the product topology on~$\Sigma^{\ZZ^d}$.  The space~$\Sigma^{\ZZ^d}$ with the product topology is compact and metrizable.

The \emph{shift} by $k\in\ZZ^d$ is the map $\sigma^k\colon\Sigma^{\ZZ^d}\to\Sigma^{\ZZ^d}$ defined by $\forall i\in \ZZ^d, \sigma^k(x)_i\isdef x_{k+i}$.  Every shift is continuous in the product topology.\vspace*{-2mm}

\paragraph{Shift spaces of finite type.} A \emph{shift space of finite type (SFT)} is a set $X\subseteq\Sigma^{\ZZ^d}$ of configurations identified by a finite number of local constraints.  More specifically, let $\collection{F}\subseteq\Sigma^\#$ be a finite set of finite patterns, which we refer to as the \emph{forbidden patterns}. The set of configurations $x\in\Sigma^{\ZZ^d}$ that avoid the patterns in~$\collection{F}$ (i.e. $\sigma^k(x)_A\notin\collection{F}$ for every pattern $p\colon A\to\Sigma$ of $\collection{F}$ and every $k\in\ZZ^d$) is called an SFT and is denoted by $X_{\collection{F}}$. Every SFT is closed (hence compact) in the product topology, and is invariant under every shift.

Observe that the choice of the defining forbidden sets $\collection{F}$ is not unique, and in our discussion, we occasionally need to consider distinct collections defining the same SFT.
The smallest integer $m$ for which $X$ can be identified by a collection of forbidden patterns with shape $S_m\isdef\{0,1,\ldots,m-1\}^d$ is referred to as the \emph{interaction range} of $X$.

A pattern (or partial configuration) $p:A\to\Sigma$ is said to be \emph{globally admissible} in $X$ if $p=x_A$ for some $x\in X$, and is said to be \emph{locally admissible} with respect to~$\collection{F}$ if it has no occurrence of the patterns from~$\collection{F}$, that is, if $\sigma^k(p)_B\notin\collection{F}$ for all $k\in\ZZ^d$ and finite $k+B\subseteq A$.  Note that in general, a locally admissible pattern does not need to be globally admissible.\vspace*{-2mm}

\paragraph{Tiling spaces (or nearest-neighbour SFTs).}  A \emph{tiling space} (or \emph{nearest-neighbour SFT}) is an SFT defined by a collection of \emph{nearest-neighbour} forbidden patterns, that is to say, patterns whose shapes consist of exactly two adjacent cells. Formally, let $e_1,e_2,\ldots,e_d$ denote the standard basis vectors  in~$\RR^d$.  A nonempty set $X\subseteq\Sigma^{\ZZ^d}$ is a ($d$-dimensional) tiling space if there exist functions $v_1,v_2,\ldots,v_d\colon\Sigma^2\to\{0,1\}$ such that
\begin{align}
	X &= \left\{x\in \Sigma^{\ZZ^d} : \forall c\in \ZZ^d, \forall i\in\{1,2,\ldots,d\}, v_i(x_{c},x_{c+e_i})=1\right\} \;.
\end{align}

\begin{example}[Homogeneous space]
\label{ex:homog}
We denote by $\Hom_2=\{\unifO,\unifI\}$ the $d$-dimensional SFT containing only the two homogeneous configurations
$\unifO,\unifI\in\{\symb{0}, \symb{1}\}^{\ZZ^d} $.
This can be seen as the tiling space defined by the functions $v_1\isdef\cdots\isdef v_d\isdef v$ where $v(a,b)\isdef 1 $ if $a=b$, and $0$ if $a\not=b$.
\hfill\exampleqed
\end{example}

\begin{example}[$k$-colourings]
A \emph{$k$-colouring} of the lattice~$\ZZ^d$ is an assignment of colours from $\Sigma\isdef\{0,1,\ldots,k-1\}$ to each position in such a way that the adjacent positions have different colours.
The set of all $k$-colourings $\Col_k$ is a tiling space identified by the functions $v_1\isdef\cdots\isdef v_d\isdef v$ where $v(a,b)\isdef 1$ if $a\neq b$, and $0$ if $a=b$.\\
\phantom{.}\hfill\exampleqed
\end{example}

\begin{example}[Hard-core]
\label{ex:hardcore}
The $d$-dimensional \emph{hardcore} tiling space on the set of symbols $\Sigma=\{\symb{0},\symb{1}\}$ is defined by the function $v_1\isdef\cdots\isdef v_d\isdef v$ where $v(a,b)\isdef 1$ if and only if $(a,b)\neq(\symb{1},\symb{1})$.
\hfill\exampleqed
\end{example}

\begin{example}[Wang tiles]
A general family of tiling spaces are those defined by Wang tiles.  A \emph{Wang tile} is a unit square with coloured edges (see Figure~\ref{fig:ammann}).  The colours indicate the matching rules for tiling: two tiles placed next to each other must have the same colour on their touching edges.  A finite collection~$\Theta$ of Wang tiles identifies a two-dimensional tiling space $X\subseteq\Theta^{\ZZ^2}$, consisting of all \emph{valid tilings} (i.e., configurations that respect the matching rule).  In other words, $X$ is defined by the functions $v_1,v_2:\Theta^2\to\{0,1\}$ where $v_1(a,b)\isdef 1$ if and only if the right edge of~$a$ has the same colour as the left edge of~$b$, and $v_2(a,b)\isdef 1$ if and only if the top edge of~$a$ has the same colour as the bottom edge of~$b$.
\hfill\exampleqed
\end{example}

\subsection{Perturbations of configurations}

\paragraph{Finite perturbations.} For two configurations $x,y\in \Sigma^{\ZZ^d}$, we denote by $\Delta(x,y)\isdef\{i\in\ZZ^d: x_i\neq y_i\}$ the set of cells at which $x$ and $y$ disagree.
A \emph{finite perturbation} of a configuration $x\in\Sigma^{\ZZ^d}$ in $\Sigma^{\ZZ^d}$ is a configuration $\tilde{x}\in\Sigma^{\ZZ^d}$ such that $\Delta(x,\tilde{x})$ is finite.

The \emph{diameter} of a finite set $A\subseteq\ZZ^d$, denoted by $\diam(A)$, is the smallest $m\in\NN$ such that $A$ fits in a hypercube of size $m$, that is, $A\subseteq i+[0,m)^d$ for some $i\in\ZZ^d$.
For two configurations $x,y\in\Sigma^{\ZZ^d}$, the diameter of $\Delta(x,y)$ is denoted by $\delta(x,y)$.

Given an SFT $X\subseteq\Sigma^{\ZZ^d}$, we denote by $\FPert[\Sigma]{X}$ the set of finite perturbations of the elements of $X$ in~$\Sigma^{\ZZ^d}$, that is $\FPert[\Sigma]{X}\isdef\big\{y\in \Sigma^{\ZZ^d} : \exists x\in X, \delta(x,y)<\infty\big\}$.

Let us stress that the set $\FPert[\Sigma]{X}$ depends on the choice of the alphabet $\Sigma$.  A larger alphabet $\Sigma'\supseteq\Sigma$ would lead to a larger set $\FPert[\Sigma']{X}$ of finite perturbations.  When the choice of the alphabet is clear from the context, we will simply use the notation $\FPert{X}$ as a shortcut.\vspace*{-2mm}

\paragraph{Case of tiling spaces.} When considering an element $x\in\FPert{X}$, we will often examine the set of cells where the constraints of the SFT are not respected. We will say that such cells are \emph{defective}, or \emph{have defects}. In the specific case of tiling spaces (or nearest-neighbour SFT), we introduce different notions of defects.

\medskip
For a configuration $x\in\Sigma^{\ZZ^d}$, a cell $c\in\ZZ^d$ is said to have a \emph{defect in direction $e_i$} (with respect to $v_i$) if $v_i(x_{c},x_{c+e_i})=0$. It has a  \emph{defect in direction $-e_i$} if $v_i(x_{c-e_i},x_{c})=0$. In the two-dimensional case, we will also use the terminology \emph{E-defect, W-defect, N-defect, S-defect} instead of respectively defect in direction $e_1, -e_1, e_2 -e_2$. The set of cells having a defect is then defined by
\begin{align}
	\Def(x) &\isdef
		\{c\in\ZZ^d : \text{$\exists e\in\{\pm e_1,\ldots, \pm e_d\}$,
			$c$ has a defect in direction $e$}\} \;.
\end{align}
A cell $c\in\ZZ^d$ is said to be \emph{defect-free} if it does not belong to $\Def(x)$, meaning that it obeys the local constraints in the $2d$ directions.

\medskip
Note that in somes cases, even if a configuration contains only very few defects, it is necessary to modify a much larger set of cells in order to reach a valid configuration. More precisely, for some tiling spaces $X$, neither the cardinality nor the diameter of $\Def(\tilde x)$ gives much information about $\delta(\tilde x,X)$.

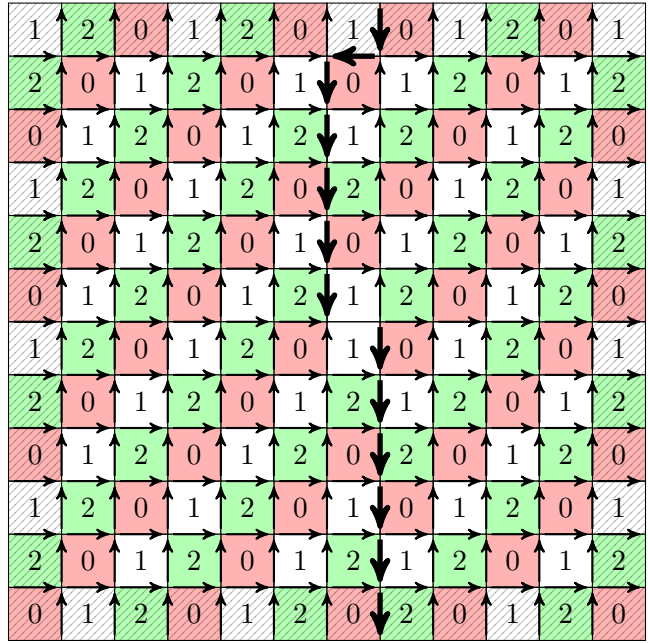
\begin{figure}[!b]
\vspace*{1mm}
	\begin{center}
		\begin{tikzpicture}[scale=0.7,>=stealth']
			\SVcaseError
		\end{tikzpicture}
	\end{center}\vspace*{-4mm}
\caption{A configuration in $\FPert{\Col_3}$ that has only two defects but is nonetheless far from every element of~$\Col_3$.
The arrows pointing South and East are emphasized.}
\label{fig:SWerror}
\end{figure}

\begin{example}[$3$-colourings]
\label{ex:3col} Let $\Col_3$ be the set of two-dimensional $3$-colourings. For any integer $n\geq 1$, there exists a configuration $\tilde x\in\FPert{\Col_3}$ such that $\Def(x)$ contains only two adjacent cells, and $\delta(\tilde x,\Col_3)\geq n$.
In other words, in order to correct a single defect, one may have to modify the state of cells that are arbitrarily far.
The construction of such a configuration is illustrated in Figure~\ref{fig:SWerror}, using the connection between the set of $3$-colourings and the six-vertex model.
This connection and a further discussion of $3$-colourings will be presented in Section~\ref{sec:discussion:3col}.
\hfill\exampleqed
\end{example}

A symbol $\alpha\in\Sigma$ is called a \emph{safe} symbol for an SFT $X\subseteq\Sigma^{\ZZ^d}$ if for every $x\in X$ and each $k\in\ZZ^d$, the configuration $\tilde{x}$ obtained from~$x$ by replacing $x_k$ with $\alpha$ is again in~$X$.  In case of a tiling space, this means that $v_i(\alpha,\sigma)=v_i(\sigma,\alpha)=1$ for all $\sigma\in\Sigma$ and $i\in\{1,2,\ldots,d\}$.
For example, for the hardcore tiling space of Example~\ref{ex:hardcore}, the symbol $\symb{0}$ is a safe symbol.
If a tiling space has a safe symbol, then the phenomenon discussed above for $3$-colourings cannot occur, since one can always update the cells having a defect with a safe symbol in order to recover a valid configuration.

\subsection{Self-stabilising cellular automata}
\label{sec:setting:self-stab}

\paragraph{Cellular automata.} A \emph{cellular automaton (CA)} is a dynamical system on $\Sigma^{\ZZ^d}$ obtained by repeated parallel updating of the symbols on the lattice using a local rule. More specifically, given a finite set $\Neighb\subseteq\ZZ^d$ and a map $f\colon\Sigma^\Neighb\to\Sigma$, we can define a mapping $F\colon\Sigma^{\ZZ^d}\to\Sigma^{\ZZ^d}$ by $F(x)_i \isdef f\big(\sigma^i(x)_{\Neighb}\big)$. This is the \emph{global map} of the CA defined with \emph{local rule}~$f$ and \emph{neighbourhood}~$\Neighb$.	

The neighbourhood can always be chosen to have the form $\Neighb=\{-r,\ldots,r\}^d$, in which case we say that the CA has \emph{neighbourhood radius}~$r$.  The neighbourhood $\Moore\isdef\{-1,0,1\}^d$ is called the \emph{Moore} neighbourhood, and the neighbourhood $\vonNeumann=\{0,\pm e_1,\ldots, \pm e_d\}$ is referred to as the \emph{von~Neumann} neighbourhood.

For a neighbourhood~$\Neighb$ and a set $A\subseteq\ZZ^d$, we write $\Neighb(A)$ and $A+\Neighb$ interchangeably.
For an integer $t\geq 1$, we also introduce the set $\Neighb^t=\smash{\underbrace{\Neighb+\ldots+\Neighb}_{\text{ $t$ times}}}$. If the CA $F$ has neighbourhood~$\Neighb$, then $F^t$ is a CA of neighbourhood~$\Neighb^t$.\vspace{-2mm}

\paragraph{Self-stabilisation.} We say that a CA $F\colon\Sigma^{\ZZ^d}\to\Sigma^{\ZZ^d}$ \emph{stabilises} an SFT $X\subseteq\Sigma^{\ZZ^d}$ from finite perturbations if
\begin{enumerate}[label={\roman*)}]
	\item (\emph{consistency}) the configurations of $X$ are fixed points, that is, $F(x)=x$ for every $x\in X$,
	\item (\emph{attraction}) finite perturbations of the elements of $X$ evolve to $X$ in finitely many steps, that is, for every $\tilde{x}\in\FPert[\Sigma]{X}$, there exists a time $t\in\NN$ such that $F^t(\tilde{x})\in X$.
\end{enumerate}
The first such $t$ is called the \emph{stabilisation time} (or the \emph{recovery time}) starting from $\tilde{x}$.  We say that $F$ stabilises $X$ from finite perturbations \emph{in time $\tau(n)$} if for each $n\in\NN$, the largest stabilisation time among all the configurations $\tilde{x}$ with $\delta(\tilde{x},X)=n$ is $\tau(n)$.
We remark that the above notion of stabilisation makes sense even if $X$ is an arbitrary set of configurations, and not only when it is an SFT.  We restrict ourselves to the scenario in which $X$ is an SFT, because local and uniform constraints appear more natural in the present context.

Note that the above definition does not rule out the possibility that $\Sigma$ has extra symbols in addition to those appearing in the elements of~$X$.  In other words, it is possible that $X\subseteq\Gamma^{\ZZ^d}$ for some $\Gamma\subsetneq\Sigma$.
Let us emphasize that, according to the above definition, in order for $F$ to stabilise $X$, it is necessary that all finite perturbations of $X$ in $\Sigma^{\ZZ^d}$ (not just those in~$\Gamma^{\ZZ^d}$) evolve to~$X$.
Even without this requirement, the problem of finding a cellular automaton that corrects finite perturbation on a given SFT remains non-trivial.  We will come back to this in Section~\ref{sec:discussion:3col} and in Problem~\ref{q:open:general-solution} of Section~\ref{sec:discussion:other-open}.

\medskip
Following an observation made earlier, note that if a tiling space $X\subseteq\Sigma^{\ZZ^d}$ has a safe symbol, then one can easily construct a CA that stabilises $X$ from finite perturbations, without extra symbols. Indeed, suppose that $\alpha$ is a safe symbol of $X$. Then, the map $F\colon\Sigma^{\ZZ^d}\to\Sigma^{\ZZ^d}$ defined by
\begin{align}\label{eq:safe-symbol:stabiliser}
	\forall c\in\ZZ^d, \quad
		F(x)_c \isdef
		\begin{cases}
			\alpha	& \text{if $c\in\Def(x)$,} \\
			x_c 	& \text{otherwise,}
		\end{cases}
\end{align}
is a CA with neighbourhood $\Neighb=\{0,\pm e_i\}$, and it stabilises $X$ in one step.

\medskip
The aim of the current article is to present self-stabilising CAs for other families of SFTs, for which finding a CA achieving the stabilisation from finite perturbations is a non-trivial problem.

\section{One-dimensional case}
\label{sec:1d}

In this section, we focus on the problem of self-stabilisation for a one-dimensional SFT $X\subseteq \Sigma^{\ZZ}$. As a specific case, let us present a well-known example of a CA stabilising the one-dimensional tiling space $\Hom_2=\{\unif{\symb{0}},\unif{\symb{1}}\}$ in linear time, without extra symbols.

\begin{example}[GKL]
\label{exp:gkl-mtraffic}
The G\'acs-Kurdyumov-Levin (GKL) cellular automaton is the CA $\gkl\colon\BinA^{\ZZ}\to\BinA^{\ZZ}$ with neighbourhood $\Neighb=\{-3,-1,0,1,3\}$ defined for any $x\in\BinA^{\ZZ}$ and $k\in\ZZ$ by
\begin{align}
	\gkl(x)_k &\isdef
		\begin{cases}
		\maj(x_k,x_{k+1},x_{k+3}) & \text{if $x_k=1$,} \\
		\maj(x_k,x_{k-1},x_{k-3}) & \text{if $x_k=0$.}\\
		\end{cases}
\end{align}
It is known that this CA is both a \emph{$\symb{0}$-eroder} and a \emph{$\symb{1}$-eroder}, which precisely means that it stabilises~$\Hom_2$ from finite perturbations~\cite{GaKuLe78}. Furthermore, the stabilisation occurs in linear time~\cite{GoMa92}. Another slightly simpler CA having the same property was proposed by Kari and Le~Gloannec, under the name of \emph{modified traffic}~\cite{KaLG12}.
\hfill\exampleqed
\end{example}

More generally, we will prove that (essentially) every one-dimensional SFT has a CA that stabilises it from finite perturbations in linear time.

\medskip
Let $X\subseteq \Sigma^{\ZZ}$ be a one-dimensional SFT. Then, there exists an integer $k\geq 1$ such that $X$ can be described by a set $\collection{F}$ of forbidden words of length $k+1$.  In this case, we say that $X$ is a \emph{$k$-step} SFT. Indeed, the constraints can be represented by a transition matrix~$A$ indexed by $\Sigma^k\times\Sigma^k$ such that for two words $v=v_1\cdots v_k$ and $w=w_2\cdots w_{k+1}$ of length $k$,
\begin{align}
	A(v_1\cdots v_k,w_2 \cdots w_{k+1}) &\isdef
	\begin{cases}
		1 & \text{if $v_2=w_2$, \ldots, $v_k=w_k$ and $v_1v_2\cdots v_kw_{k+1}\not\in\collection{F}$,} \\
		0 & \text{otherwise.}
	\end{cases}
\end{align}
Note that if $k=1$, then $X$ is a tiling space, that is, a nearest-neighbour SFT.
 \eject
For an integer $n\geq 1$, we denote by $L_n(X)$ the set of words of length $n$ occuring in $X$, that is,
\begin{align}
	L_n(X)\isdef\big\{w\in\Sigma^n : \exists x\in X,\; x_{1}\cdots x_n=w_1\cdots w_n\big\} \;.
\end{align}
The set $L(X)\isdef\bigcup_{n\in\NN} L_n(X)$ is called the \emph{language} of $X$.

\begin{remark} In the specific case when $X$ is a \emph{mixing} one-dimensional SFT (meaning that there exists $n_0\geq 0$ such that for every two words $u,v\in L(X)$ and every $n\geq n_0$, there exists a word $w\in L_n(X)$ such that $uwv\in L(X)$) and has at least one homogeneous configuration, then a result of Maass implies that the self-stabilisation problem is trivial for $X$~\cite[Theorem 3.2]{Ma95}. Indeed, this result states that there exists a CA $F\colon\Sigma^\ZZ\to\Sigma^\ZZ$ (i.e., without extra symbols) such that
\begin{enumerate}
\itemsep=0.95pt
\item $F$ is the identity map on $X$,
\item $F(\Sigma^\ZZ)=X$.
\end{enumerate}
In particular, we achieve self-stabilisation in one step.
\hfill\remarkqed
\end{remark}

For $u,v\in L_k(X)$, we write $u\sim v$ if there is a word $w\in \Sigma^*$ such that $uwv\in L(X)$.
We say that $X$ is \emph{non-wandering} if $\sim$ is an equivalence relation, meaning that the transition graph defined by the adjacency matrix $A$ consists of strongly connected components, with no connections between them.

\begin{example}[Wandering vs.\ non-wandering]
The SFT $\{\symb{0}^\ZZ,\symb{1}^\ZZ\}\subseteq\{\symb{0},\symb{1}\}^\ZZ$
is non-wandering, but the SFT consisting in $\symb{0}^\ZZ$, $\symb{1}^\ZZ$
and the translations of $\cdots\symb{0}\symb{0}\symb{0}\symb{1}\symb{1}\symb{1}\cdots$
is not.
\hfill\exampleqed
\end{example}

\begin{example}[A non-wandering SFT]\label{ex:red} The $1$-step SFT on the alphabet $\Sigma=\{\symb{0},\symb{1},\symb{2},\symb{3},\symb{4}\}$ with the transition matrix and graph illustrated in Figure~\ref{fig:1dSFT:red} is non-wandering.
We will use this as a running example to illustrate the constructions of this section.
\hfill\exampleqed
\end{example}

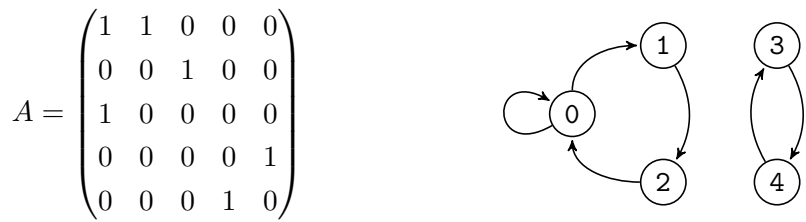
\begin{figure}[!h]
\vspace*{-4mm}
	\begin{center}
		\begin{tabular}{cc}
		\begin{minipage}{0.4\textwidth}
			\centering
			\begin{tikzpicture}[->,>=stealth',baseline=(current bounding box.center)]
				\node {$A=\begin{pmatrix}
					1 & 1 & 0 & 0 & 0 \\
					0 & 0 & 1 & 0 & 0 \\
					1 & 0 & 0 & 0 & 0 \\
					0 & 0 & 0 & 0 & 1 \\
					0 & 0 & 0 & 1 & 0
					\end{pmatrix}$};
			\end{tikzpicture}
		\end{minipage}
		&
		\begin{minipage}{0.4\textwidth}
			\centering
			\begin{tikzpicture}[->,>=stealth',shorten >=1pt,semithick]
				\tikzstyle{every state}=[minimum size=17pt,inner sep=0pt]
				
				\node[state] (S0) at (-1.2,0) {$\symb{0}$};
				\node[state] (S1) at (0,0.9) {$\symb{1}$};
				\node[state] (S2) at (0,-0.9) {$\symb{2}$};
				\node[state] (S3) at (1.5,0.9) {$\symb{3}$};
				\node[state] (S4) at (1.5,-0.9) {$\symb{4}$};
				
				\path (S0) edge [out=-150,in=150,loop,overlay] (S0);
				\path (S0) edge [out=90,in=180] (S1);
				\path (S1) edge [out=-60,in=60] (S2);
				\path (S2) edge [out=180,in=-90] (S0);
				\path (S3) edge [out=-60,in=60] (S4);
				\path (S4) edge [out=120,in=-120] (S3);
			\end{tikzpicture}
		\end{minipage}
		\end{tabular}
	\end{center}\vspace*{-6mm}
\caption{An example of a non-wandering SFT. (left) transition matrix (right) transition graph.}
\label{fig:1dSFT:red}\vspace*{-2mm}
\end{figure}

\begin{theorem}[Self-stabilisation in one dimension]
\label{prop:robust:1d}
For every non-wandering one-dimensional SFT $X\subseteq\Sigma^\ZZ$, there exists a CA $F\colon{\Sigma'}^\ZZ\to{\Sigma'}^\ZZ$ with ${\Sigma'}\supseteq\Sigma$ that stabilises $X$ from finite perturbations in linear time.
\end{theorem}

\begin{remark}
Ilkka T\"orm\"a has communicated with us an alternative construction showing that \emph{every} one-dimensional SFT is stabilised by a CA with \emph{no extra symbol}~\cite{Tor15}.
His construction uses long markers made of forbidden patterns, allowing to avoid the introduction of extra symbols.
\hfill\remarkqed
\end{remark}

The rest of this section is devoted to the proof of Theorem~\ref{prop:robust:1d}.  Thus, for the rest of this section, we will assume that $X$ is a $k$-step non-wandering SFT.

We denote by $m$ the smallest non-negative integer such that for every $u,v\in L_k(X)$ with $u\sim v$, we have $uwv\in L(X)$ for some word $w$ of length at most $m$.

In the case of the tiling space of Example~\ref{ex:red}, recall that $k=1$, and one can check that $m=2$. Indeed, for $u=v=1$ or $u=v=2$, one needs a word $w$ of length at least $2$ in order to have $uwv\in L(X)$, while in all the other cases, a word of length $0$ or $1$ is sufficient.

\medskip
For a configuration $y\in \Sigma^\ZZ$, we denote by
\begin{align}
	\Def(y) &\isdef \{i\in\ZZ: y_{[i-k,i]}\notin L_{k+1}(X)\}
\end{align}
the set of cells at which a \emph{defect} occurs.
Note that if $y$ is a (finite) perturbation of a configuration $x\in X$,
then $\Def(y)\subseteq\Delta(x,y)+\{0,1,\ldots,k\}$,
(i.e., every defect on $y$ is on the right within distance $k$ from an element of $\Delta(x,y)$),
but $\Def(y)$ could be much smaller than $\Delta(x,y)$.

\medskip
Let us consider again Example~\ref{ex:red}.
Let $y\in\Sigma^{\ZZ}$ be the configuration below.
The set $\Def(y)$ contains only two elements, which correspond to the positions marked by a cross. But in order to recover a valid configuration, we need to change the values of at least seven cells (the ones taking values $3$ and~$4$).

\begin{align}
	\begin{array}{cccccccccccccccccccc}
	\hline
	\cdots&\symb{0}&\symb{1}&\symb{2}&\symb{0}&\symb{1}&\symb{3}&\symb{4}&\symb{3}&\symb{4}&\symb{3}&\symb{4}&\symb{3}&\symb{2}&\symb{0}&\symb{1}&\symb{2}&\symb{0}&\symb{0}&\symb{0}\cdots\\
	\hline
	& & & & & &\times & & & & & & & \times & & & & & &
	\end{array}
\end{align}

\paragraph{A sequential correction process.}
We first describe a particular \emph{sequential} procedure for
correcting finite ``islands'' of defects on $X$.
This procedure will not be a cellular automaton itself, but will be used in the construction of the cellular automaton.
Applied on a configuration~$y$, the procedure involves updating
the symbols on $y$ one by one, from left to right, starting from
a cell on the left of the leftmost element of $\Def(y)$.
Each update is performed according to the same local rule, which we call the \emph{patching rule}.
If $y$ is in $X$, then the procedure does not modify any symbol in $y$.
If $y$ is a finite perturbation of a configuration $x\in X$,
then the procedure eventually turns $y$ into a configuration $z\in X$,
before reaching few cells to the right of the rightmost element of $\Def(y)$.
The existence of an appropriate patching rule relies on
the fact that $X$ is non-wandering.

More specifically, the patching rule will be a function $g\colon\Sigma^{2k+m}\to \Sigma$.
In order to update the symbol at position $i$ on $y$,
we replace it with $g(y_{[i-k,i+m+k)})$.
The sequential updating of
$y$ from a cell~$i$ to a cell $j\geq i$
proceeds by first updating the symbol at cell $i$,
then updating the symbol at cell~$i+1$, and so forth until
we update the symbol at cell~$j$.

\medskip
The patching rule $g$ is constructed as follows.
\begin{itemize}
	\item For $u\in L_k(X)$ and $q\in \Sigma^{m+k}$,
		we let $r\in [0,m]$ be the smallest index (if exists) such that
		$u w q_{[r+1,m+k]}\in L(X)$ for some $w\in \Sigma^r$.
		If no such $r$ exists, we choose an arbitrary $a\in \Sigma$
		such that $ua\in L(X)$ and set $g(uq)\isdef a$.
		If $r$ exists, we choose a corresponding $w\in \Sigma^r$
		and set $g(uq)\isdef w_1$ if $r>0$ and $g(uq)\isdef q_1$ if $r=0$.
	\item For $u\in\Sigma^k\setminus L_k(X)$ and $q\in \Sigma^{m+k}$,
		we simply set $g(uq)\isdef q_1$.
		(This case is not used during a sequential correction.)
\end{itemize}

Let us see on some examples how we can construct a patching rule in the context of the tiling space of Example~\ref{ex:red}, where $k=1$ and $m=2$. If $u\in\{\symb{0},\symb{1},\symb{2}\}$ and $q_3\in\{\symb{3},\symb{4}\}$, then, there is no index $r\in [0,2]$ as above. So, we can simply set $g(uq_1q_2q_3)\isdef (u+1) \bmod 3$. Consider now for example $g(\symb{0432})$. Then $r=2$, and for $w=01$, we have $\symb{0012}\in L(X)$. So, we can set $g(\symb{0432})=\symb{0}$, which is in fact the only possible choice here. For $g(\symb{0320})$, we have $r=1$, and again there is no other choice than to set $g(\symb{0320})=\symb{1}$.

Intuitively, when applied at the leftmost defect $i\in \Def(y)$,
the patching rule updates $y_i$ based on the (roughly) shortest
patch that would remove the defects on $y$.
The following simple lemma formulates the main property of the patching rule.

\begin{lemma}[Sequential process]
	Let $y$ be a finite perturbation of a configuration $x\in X$.
	Let $[a,b]$ be an interval containing~$\Def(y)$.
	Then, the sequential updating of $y$ from $a$ to $b+m$ using the above-constructed
	patching rule~$g$ ends with an element of~$X$.
\end{lemma}
\begin{proof}
	For $i\in[a,b+m]$, let $y^{(i)}$ denote the configuration obtained
	during the sequential updating procedure right after updating cell $i$,
	and set $y^{(a-1)}\isdef y$ for consistency.
	
	Observe that $y'\isdef y^{(b-1)}$ has the following property:
	the finite words in $y'_{(-\infty,b-1]}$ and $y'_{[b,\infty)}$
	are all in $L(X)$.
	Since $y'_{[b-k,b-1]}$ and $y'_{[i,i+k)}$ for $i\geq b$
	are in the same transitive component of~$X$
	(here we are using the non-wandering property),
	there must be an integer $r\in[0,m]$
	such that $y'_{[b-k,b-1]}w y'_{[b+r,b+r+k)}\in L(X)$
	for some $w\in \Sigma^r$.  Now, the construction of $g$ ensures that
	$y^{(b+r-1)}$ is in $X$.
\end{proof}

Let us observe on Figure~\ref{fig:Seqfirststeps} how the patching rule described for Example~\ref{ex:red} corrects the configuration $y$ represented above, when applied from left to right. The smallest interval $[a,b]$ containing the defects is marked on the first configuration. In order to update the symbol at position $i$, we replace it by $g(y_{i-1}y_iy_{i+1}y_{i+2})$. The successive cells that are updated are underlined. The last configuration that is shown belongs to the SFT, so that afterwards, the patching rule does not introduce any new change.

\begin{figure}[!h]\small
$$
\begin{array}{ccccccccccccccccccccc}
\hline
\cdots&\symb{0}&\symb{1}&\symb{2}&\symb{0}&\symb{1}&\tikzmark{ExRegA}\uline{\symb{3}}&\symb{4}&\symb{3}&\symb{4}&\symb{3}&\symb{4}&\symb{3}&\symb{2}\tikzmark{ExRegB}&\symb{0}&\symb{1}&\symb{2}&\symb{0}&\symb{0}&\symb{0}&\cdots\\
\hline
\\
\hline
\cdots&\symb{0}&\symb{1}&\symb{2}&\symb{0}&\symb{1}&\symb{2}&\uline{\symb{4}}&\symb{3}&\symb{4}&\symb{3}&\symb{4}&\symb{3}&\symb{2}&\symb{0}&\symb{1}&\symb{2}&\symb{0}&\symb{0}&\symb{0}&\cdots\\
\hline
\\
\hline
\cdots&\symb{0}&\symb{1}&\symb{2}&\symb{0}&\symb{1}&\symb{2}&\symb{0}&\uline{\symb{3}}&\symb{4}&\symb{3}&\symb{4}&\symb{3}&\symb{2}&\symb{0}&\symb{1}&\symb{2}&\symb{0}&\symb{0}&\symb{0}&\cdots\\
\hline
\\
\hline
\cdots&\symb{0}&\symb{1}&\symb{2}&\symb{0}&\symb{1}&\symb{2}&\symb{0}&\symb{1}&\uline{\symb{4}}&\symb{3}&\symb{4}&\symb{3}&\symb{2}&\symb{0}&\symb{1}&\symb{2}&\symb{0}&\symb{0}&\symb{0}&\cdots\\
\hline
\\
\hline
\cdots&\symb{0}&\symb{1}&\symb{2}&\symb{0}&\symb{1}&\symb{2}&\symb{0}&\symb{1}&\symb{2}&\uline{\symb{3}}&\symb{4}&\symb{3}&\symb{2}&\symb{0}&\symb{1}&\symb{2}&\symb{0}&\symb{0}&\symb{0}&\cdots\\
\hline
\\
\hline
\cdots&\symb{0}&\symb{1}&\symb{2}&\symb{0}&\symb{1}&\symb{2}&\symb{0}&\symb{1}&\symb{2}&\symb{0}&\uline{\symb{4}}&\symb{3}&\symb{2}&\symb{0}&\symb{1}&\symb{2}&\symb{0}&\symb{0}&\symb{0}&\cdots\\
\hline
\\
\hline
\cdots&\symb{0}&\symb{1}&\symb{2}&\symb{0}&\symb{1}&\symb{2}&\symb{0}&\symb{1}&\symb{2}&\symb{0}&\symb{0}&\uline{\symb{3}}&\symb{2}&\symb{0}&\symb{1}&\symb{2}&\symb{0}&\symb{0}&\symb{0}&\cdots\\
\hline
\\
\hline
\cdots&\symb{0}&\symb{1}&\symb{2}&\symb{0}&\symb{1}&\symb{2}&\symb{0}&\symb{1}&\symb{2}&\symb{0}&\symb{0}&\symb{1}&\symb{2}&\symb{0}&\symb{1}&\symb{2}&\symb{0}&\symb{0}&\symb{0}&\cdots\\
\hline
\end{array}
$$
\begin{tikzpicture}[remember picture,overlay]
	\draw[decorate, decoration={brace,amplitude=3pt},thick]  ([yshift=1.2em]ExRegA) -- ([yshift=1.2em]ExRegB);
\end{tikzpicture}\vspace*{-3mm}
\caption{An example of evolution of the sequential correcting process in the context of the SFT of Example~\ref{ex:red}.  Time goes downwards.}\label{fig:Seqfirststeps}\vspace*{-5mm}
\end{figure}
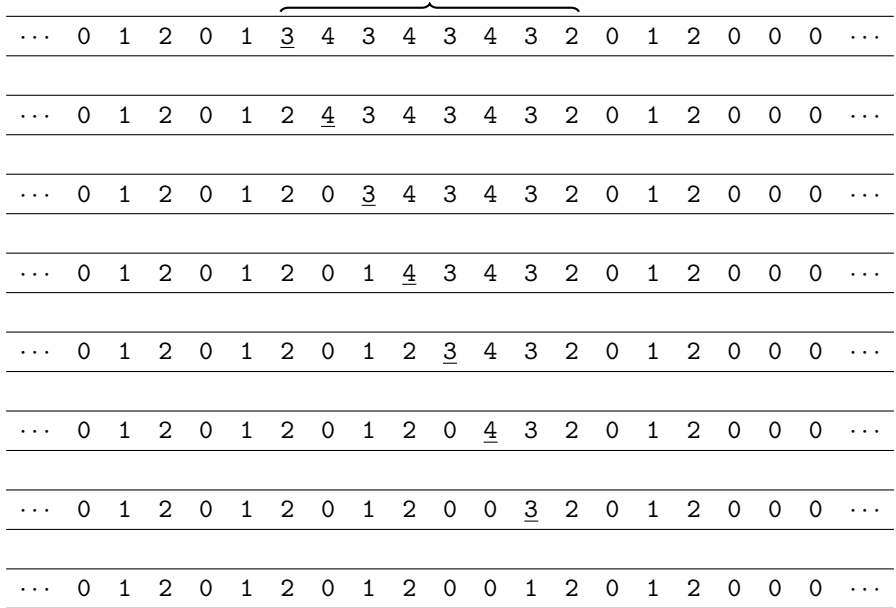

\smallskip
Let us observe that the above procedure may fail if the SFT is wandering.
For instance, if we modify the SFT of Example~\ref{ex:red} so as to allow transition from $\symb{1}$ to $\symb{3}$, then the configuration
\begin{align}
	\begin{array}{cccccccccccccccccccc}
	\hline
\cdots&\symb{0}&\symb{1}&\symb{2}&\symb{0}&\symb{1}&\symb{3}&\symb{4}&\symb{3}&\symb{4}&\symb{3}&\symb{4}&\symb{3}&\symb{2}&\symb{0}&\symb{1}&\symb{2}&\symb{0}&\symb{0}&\symb{0}\cdots\\
	\hline
	& & & & & & & & & & & & & \times & & & & & &
	\end{array}
\end{align}
will contain only one defect.  As before, in order to correct this configuration, we have no choice but to modify every $\symb{3}$ and $\symb{4}$.  However, in this case, starting from any position, the sequential process described above can only modify the cells to the right of the rightmost $\symb{3}$, hence it will not succeed.

\paragraph{From the sequential rule to the stabilising CA.}
We now construct a CA $T\colon\Sigma'^\ZZ\to \Sigma'^\ZZ$ with an extended alphabet $\Sigma'\supseteq \Sigma$
that corrects finite islands of defects on $X$ in linear time.
The CA uses the patching rule $g$ to correct defects sequentially
from left to right.  Since the CA cannot a priori
identify the leftmost defect, it instead applies the patching rule
simultaneously everywhere that locally looks like the leftmost defect.
Therefore, if there are several far apart defect regions,
there will be a correction trail initiated from the left of each of them.
These correction trails need not be consistent with one another.
We use suitable signals to make sure that the leftmost trail of correction is ``dominant'',
and the other ones do not continue forever to the right, corrupting the original configuration.
To this end, each correction trail leaves a trace (using extra symbols)
that is slowly faded away on its own.
If a correction trail coming from the left meets a trace in front of it,
it sends a fast signal ahead (again using extra symbols) to stop the correction trail
that has left that trace.  So, the leftmost trail eventually stops all the trails
in front of it and goes on to correct the entire island of defects.
We must of course make sure that this scenario works
even if the defects in the initial configuration involve symbols from the extended alphabet.

\medskip
The extended alphabet will be $\Sigma'\isdef \Sigma\times\{\xnone,\xtrac,\xstop\}$,
in which $\{\xnone\}\times \Sigma$ is identified with $\Sigma$.
The symbol~$\xtrac$ represents the \emph{trace} and $\xstop$ signifies the \emph{stop} signal.
For convenience, we identify the configurations with alphabet $\Sigma'$
with pairs $(y,\alpha)$ where $y\in \Sigma^\ZZ$ and $\alpha\in\{\xnone,\xtrac,\xstop\}^\ZZ$.
The configuration $(y,\alpha)$ has a \emph{defect} at cell $i$ if
either $y$ has a defect at $i$, or $i$ contains a signal symbol $\xtrac$ or $\xstop$.
Extending the notation $\Def(y)$, we denote the set of defects on $(y,\alpha)$ by
\begin{align}
	\Def(y,\alpha) &\isdef \{i\in\ZZ: \text{$i\in \Def(y)$ or $\alpha_i\neq\xnone$}\} \;.
\end{align}
We also define the set
\begin{align}
	\Def_0(y,\alpha) &\isdef
		\big\{i\in \Def(y,\alpha):y_{[i-k,i-1]}\in L_k(X)\big\}
\end{align}
whose elements we interpret as the cells that are the leftmost elements of a defect island.
Observe that the patching rule $g$ can only affect the cells in $\Def_0(y,\alpha)$.

The CA $T$ will be constructed as a composition
$T_2(T_gT_1^2T_0)^2$
of four CA maps $T_g,T_0,T_1,T_2\colon\Sigma'^\ZZ\allowbreak\to \Sigma'^\ZZ$.
This composition will ensure that the fading of the traces
is half as slow as the correction speed, while the stop signals propagate
twice as fast as the correction speed.

\medskip
\noindent\emph{Patching}.
The map $T_g$ is defined by $T_g(y,\alpha)\isdef(y',\alpha')$, where
\begin{align}
	(y'_i, \alpha'_i) &\isdef
	\begin{cases}
		\left(g(y_{[i-k,i+m+k)}),\xtrac\right)\qquad
			&\text{if $\alpha_{i-1}\neq\xstop$ and $i\in \Def_0(y,\alpha)$,} \\
		\left(y_i,\alpha_i\right)
			&\text{otherwise.}
	\end{cases}
\end{align}
It simply applies the patching rule on the elements of $\Def_0(y,\alpha)$
(if no stop signal on the left) and leaves a trace behind.
It also erases any stop symbol which sees no stop symbol on its left neighbour,
replacing it with a trace symbol.

\medskip

\noindent\emph{Generation of stop signals}.
The map $T_0$ is responsible for generating stop signals,
and is defined by $T_0(y,\alpha)\isdef(y,\alpha')$ (i.e., no change on $y$), where
\begin{align}
	\alpha'_i &\isdef \begin{cases}
		\xstop\qquad
			&\text{if $\alpha_{i}=\xtrac$ and $i\in \Def(y)$,} \\
		\alpha_i
			&\text{otherwise.}
	\end{cases}
\end{align}

\medskip

\noindent\emph{Propagation of stop signals}.
The propagation of the stop signals is governed by the map $T_1$,
which is defined by $T_1(y,\alpha)\isdef(y,\alpha')$ (i.e., no change on $y$), where
\begin{align}
	\alpha'_i &\isdef \begin{cases}
		\xstop\qquad
			&\text{if $\alpha_{i-1}=\xstop$ and $\alpha_i=\xtrac$,} \\
		\alpha_i
			&\text{otherwise.}
	\end{cases}
\end{align}

\medskip

\noindent\emph{Fading of the traces}.
Finally, the map $T_2$ handles the fading of the traces
and is defined by $T_2(y,\alpha)\isdef(y,\alpha')$ (i.e., no change on $y$), where
\begin{align}
	\alpha'_i &\isdef \begin{cases}
		\xnone\qquad
			&\text{if $\alpha_{i-1}=\xnone$ and $\alpha_i=\xtrac$,} \\
		\alpha_i
			&\text{otherwise.}
	\end{cases}
\end{align}

In Figure~\ref{fig:CAfirststeps},
we illustrate the operation of the cellular automaton in the context of Example~\ref{ex:red}.
The defective cells are underlined.
Note that in this example $\Def_0=\Def$.
Namely, since $k=1$, we always have $y_{[i-k,i-1]}=y_{i-1}\in\Sigma=L_1(X)$.
Two consecutive configurations correspond to one application of $T_gT_1^2T_0$.
Every two steps, the map $T_2$ is also applied, erasing one trace symbol on the left of each correction region. Figure~\ref{fig:illCAOneD} illustrates the evolution of the CA in a longer time span. When the front of a correction trail meets the fading trace of another correction trail in front of it, a stop signal is created. This stop signal travels faster than the correction trails, and hence quickly catches up with all the correction trails in front of it.  As a result, every correction trail which is not initiated by the leftmost defect is eventually stopped.

\begin{figure}[!h]
\vspace{2mm}
\centering
\hspace*{-2mm}\includegraphics[height=11cm]{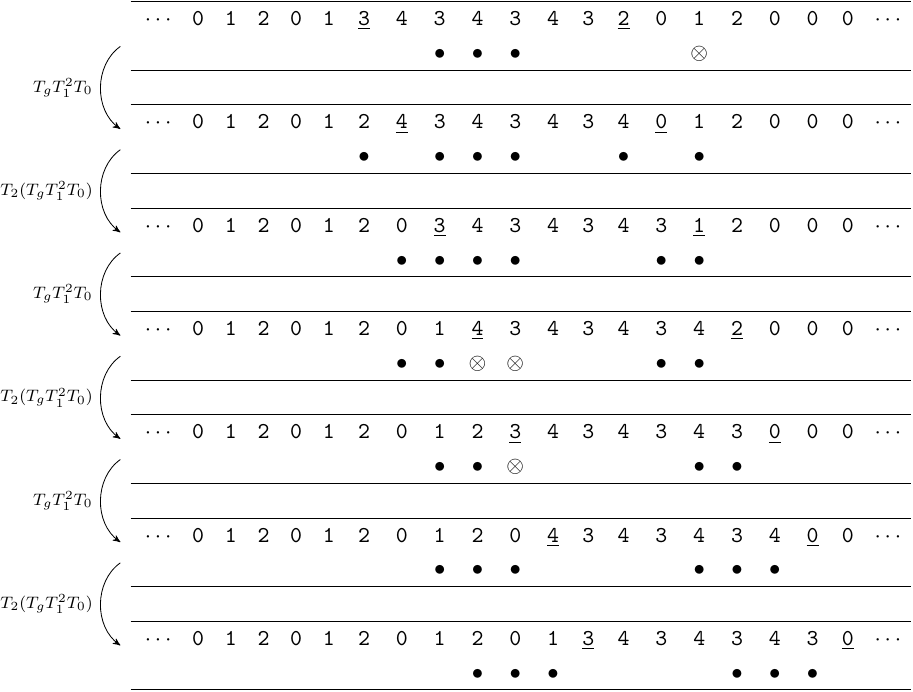}\vspace*{-2mm}
\caption{An example of evolution of the celllular automaton in the context of the SFT of Example~\ref{ex:red}.  Time goes downwards. Note that in this example, the initial perturbation already contains symbols from the extended alphabet~$\Sigma'$.}\label{fig:CAfirststeps}\vspace*{-1mm}
\end{figure}

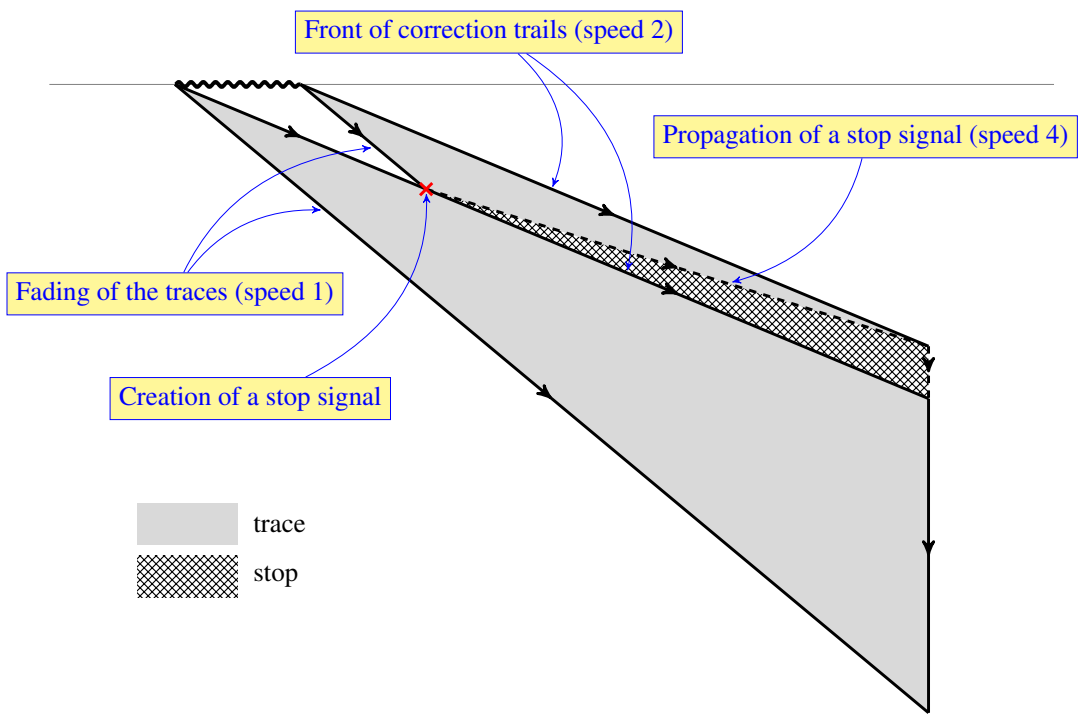
\begin{figure}[!h]
	\begin{center}
	\scalebox{0.92}{
   \begin{tikzpicture}[xscale=1.8,yscale=1.5,>=stealth']
		\CorrectionOneDVersionTwo
	\end{tikzpicture} }
	\end{center}\vspace*{-4mm}
\caption{Illustration showing the behaviour of the different signals composing the cellular automaton.
Time goes downwards.}\label{fig:illCAOneD}\vspace*{-2mm}
\end{figure}

\subsubsection*{Proof of Theorem~\ref{prop:robust:1d}:}
	We verify that the CA $T$ constructed above
	corrects finite islands of defects on~$X$ in linear time.

\medskip	
	Let $(y,\alpha)\colon\ZZ\to \Sigma'$ be a finite perturbation of
	a configuration $x\in X$
	(or more explicitly, a finite perturbation of $(x,\xnone^\ZZ)$).
	Let us call a cell $i\in \Def(y,\alpha)$
	\begin{itemize}
\itemsep=0.95pt
		\item \emph{active} on $(y,\alpha)$ if $i\in \Def_0(y,\alpha)$ and $\alpha_{i-1}\neq\xstop$,
		\item \emph{frozen} if $\alpha_{i-1}=\xstop$, and
		\item \emph{fading} if $\alpha_i=\xtrac$.
	\end{itemize}
	The \emph{tail} of an active cell $i$ is the longest (possibly empty) interval $[i-l,i-1]$
	of fading cells.  A tail whose leftmost element is frozen is said to \emph{be freezing}.
	An active cell together with its tail is a correction \emph{trail}.
	Observe that the leftmost element of $\Def(y,\alpha)$ is either active, or fading and \linebreak non-frozen.
	
\medskip
	The CA works intuitively as follows.
	Let $[a,b]\subseteq\ZZ$ be the smallest interval containing $\Delta((x,\xnone^\ZZ),\allowbreak(y,\alpha))$.
	The leftmost trail moves with speed \emph{at least} $2$ to the right,
	while its tail is erased with speed $1$ from the left.
	The rightmost trail moves with speed \emph{at most} $2$ to the right,
	occasionally becoming frozen when the following trail reaches its tail.
	In summary, at time $t>0$,
	\begin{itemize}
\itemsep=0.95pt
		\item the leftmost active cell is on the right of $a+2t$,
		\item the leftmost fading cell is on the right of $a+t$,
		\item the rightmost active cell is on the left of $b+2t$,
		\item the rightmost non-freezing tail is on the left of $b+t$.
	\end{itemize}
	It follows that before time $t_0\isdef b-a$, the leftmost trail overpasses
	the rightmost non-freezing tail, and causes it to freeze.
	From that moment on, it takes at most $(b-a)/4$ steps before
	every cell in front of the leftmost trail is inactive.
	At time $t_1\isdef t_0 + \lfloor(b-a)/4\rfloor < \frac{5}{4}(b-a)$,
	the CA is then in a configuration $(y',\alpha')$
	with the following properties.
	\begin{itemize}
		\item The configuration $(y',\alpha')$ is a finite perturbation of
			$(x,\xnone^\ZZ)$ with
			$\Delta(x,y')\subseteq [a+2t_1,b+2t_1]$ and
			$\Delta((x,\xnone^\ZZ),(y',\alpha'))\subseteq [a+t_1,b+2t_1]$.
		\item The only active cell on $(y',\alpha')$ is on the right of
			$a+2t_1$ and every element of $\Def(y',\alpha')$
			to the left of it is fading and non-frozen.
	\end{itemize}
	From time $t_1$ onward, the CA essentially applies
	the patching rule, removing the defects on $y$
	in no more than $b-a+m$ steps.
	At time $t_2\isdef t_1 + b-a+m$, the CA
	is in a configuration $(y'',\alpha'')$ in which
	\begin{itemize}
\itemsep=0.95pt
		\item $y''$ has no defect,
		\item $\alpha''$ contains no stop symbol $\xstop$
			and its trace symbols $\xtrac$ are contained in region
			$[a+t_2,a+2t_2]$.
	\end{itemize}
	Eventually, in at most $(a+2t_2)-(a+t_2)=t_2$ more steps, every trace symbol fades away,
	and we arrive at a configuration in $X$
	before time $t_3\isdef t_2+t_2\leq \frac{9}{4}(b-a)+m$. \QED

\section{Two-dimensional case}
\label{sec:2d}

We now consider the problem of stabilising a two-dimensional SFT. In this section, we will focus on some specific families of tiling spaces, for which different strategies can be impletemented in order to achieve self-stabilisation efficiently. The first cases we consider are inspired from the study of self-stabilisation for $k$-colourings~\cite{FaMaTa19}. Namely, Sections~\ref{sub:finite},~\ref{sub:single_cell}, and~\ref{sub:ell-fillable} treat and extend respectively the cases of $2$-colourings, $k$-colourings with $k\geq 5$, and $4$-colourings. Finally, Section~\ref{sec:2d:deterministic} is based on the same construction as in Section~\ref{sec:1d}, allowing us to handle the case of \emph{deterministic} SFTs.
Most of the constructions in this section can be easily generalized to higher dimensions, but for concreteness, we focus on the two-dimensional case.

As a first example, let us present a well-known CA stabilising the two-dimensional tiling space $\Hom_2=\{\unif{\symb{0}},\unif{\symb{1}}\}$ in linear time and without extra symbols.

\begin{figure}[!b]
\begin{center}
\begin{tabular}{c c c c}
\includegraphics[width=0.21\linewidth]{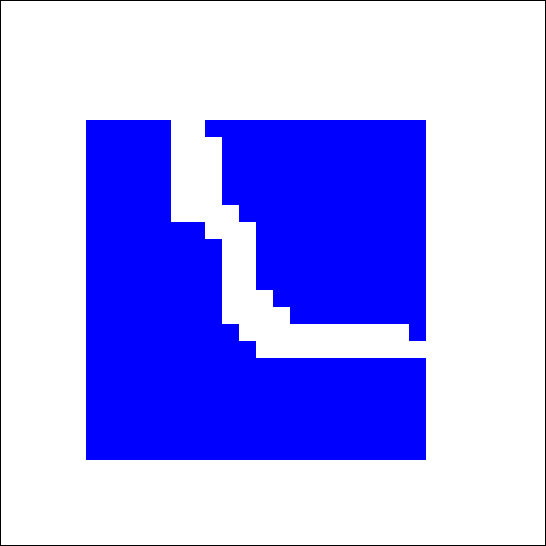} &
\includegraphics[width=0.21\linewidth]{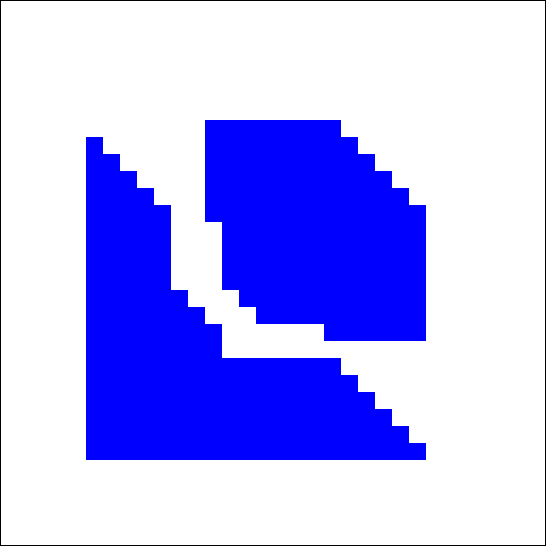} &
\includegraphics[width=0.21\linewidth]{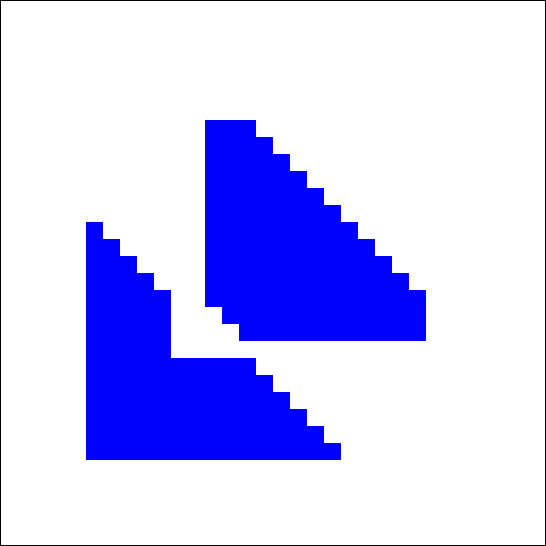} &
\includegraphics[width=0.21\linewidth]{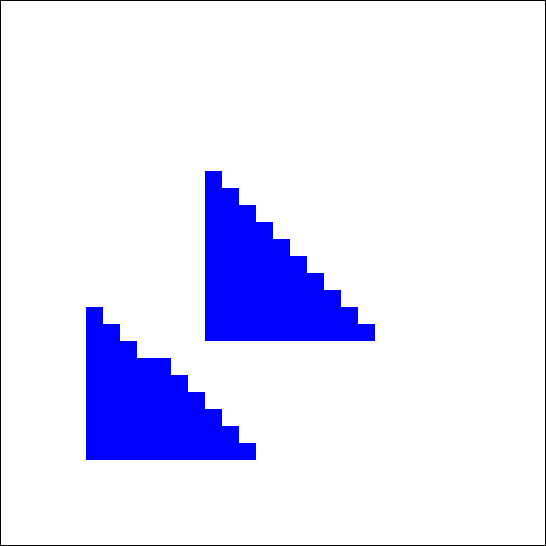} \\
$ t = 0 $ & $ t= 5$ & $ t= 10 $ & $ t =15  $
\end{tabular}
\end{center}\vspace*{-3mm}
\caption{Illustration of the evolution of Toom's CA $\toom$.}
\label{diag:toom}
\end{figure}

\begin{example}[Toom's North-East-Center majority rule]\label{exp:toom}
Toom's (deterministic) majority cellular automaton is the CA
$\toom\colon\BinA^{\ZZ^2}\to\BinA^{\ZZ^2}$ with neighbourhood $\Neighb=\{0,e_1,e_2\}$
defined, for any $x\in\BinA^{\ZZ^2}$ and $k\in\ZZ^2 $, by
\begin{align}
	\toom(x)_k &\isdef \maj(x_k,x_{k+e_1},x_{k+e_2}) \;,
\end{align}
where $\maj$ denotes the majority function, outputting the symbol which is in majority among the input symbols.

It is known that this two-dimensional CA is both a \emph{$\symb{0}$-eroder} and a \emph{$\symb{1}$-eroder}, which is equivalent to the CA stabilising $\Hom_2=\{\unifO,\unifI\}$ from finite perturbations. Furthermore, the stabilisation occurs in linear time~\cite{Too80,ToVaStMiKuPi90}.
\hfill\exampleqed
\end{example}


As we shall discuss in Section~\ref{sec:complexity}, it seems hopeless to be able to construct, for any two-dimensional SFT, a CA that stabilises it from finite perturbations in linear time, or even in polynomial time. However, in the following of this section, we will present several classes of SFTs for which we are able provide CA that stabilise in linear or quadratic time. We first focus on colourings and similar tiling spaces, then we treat the case of two-dimensional deterministic tiling spaces.

\subsection{Finite SFTs}
\label{sub:finite}

In this section, we treat the case where the SFT contains only a finite number of configurations.
Observe that the configurations of such a finite SFT are necessarily spatially periodic.

\begin{example}[Finite SFTs]
The set $\Hom_2=\{\unif{\symb{0}},\unif{\symb{1}}\}\subseteq\{\symb{0},\symb{1}\}^{\ZZ^2}$ is a finite tiling space. The set $\Col_2$ of all $2$-colourings of $\ZZ^2$ is also a finite tiling space, since it contains only two configurations, namely the odd and even chequerboard configurations.
\hfill\exampleqed
\end{example}

Let us consider an arbitrary finite SFT $X\subseteq\Sigma^{\ZZ^2}$. Then, for each configuration $x\in X$, there exist integers $n_1,n_2\geq 1$ (the \emph{horizontal} and \emph{vertical} periods of~$x$) such that $\sigma^{n_1e_1}(x)=\sigma^{n_2e_2}(x)=x$.
Let $N_1$ (respectively, $N_2$) denote the least common multiple of the horizontal (resp., vertical) periods of all the configurations in~$X$.  Since $X$ is finite, $N_1$ and $N_2$ are finite.
Then, for all $x\in X$, we have $\sigma^{N_1e_1}(x)=\sigma^{N_2e_2}(x)=x$.
We define a CA $F$ on $\Sigma^{\ZZ^2}$ by
\begin{align}
	F(x)_{k} &\isdef \maj (x_{k},x_{k+N_1e_1},x_{k+N_2e_2}) \;,
\end{align}
where the majority function $\maj(a,b,c)$ assigns to three symbols $a,b,c$ the symbol which is most common among~$a,b,c$, with the convention that when $a,b,c$ are distinct, one can choose arbitrarily the value of the function.
Observe that $F$ simply consists in applying Toom's majority rule on each sub-lattice generated by $N_1e_1$ and $N_2e_2$.

\begin{proposition}[Self-stabilisation of finite SFTs]
Let $X\subseteq{\Sigma^{\ZZ^2}}$ be a finite SFT. Then, the CA $F\colon\Sigma^{\ZZ^2}\to\Sigma^{\ZZ^2}$ defined above stabilises $X$ from finite perturbations in linear time.
\end{proposition}

\begin{proof}  It is clear from the definition that $F(x)=x$ for every $x\in X$.
For each $n\in\NN$, define the triangular set $T_n=\{k\in\ZZ^2 : k_1+k_2\leq n,\; k_1,k_2\geq 0\}$.
Let $x\in X$ and take $y$ such that $\Delta(x,y)$ is finite.
By translating $x$ and $y$ if needed, we can assume without loss of generality that the difference set $\Delta(x,y)$ is included in the triangle $T_n$ for some $n$. It is then easy to verify that $\Delta\big(x,F(y)\big)\subseteq T_{n-1}$.
Indeed, for every cell outside $T_n$, the local rule does not modify the state, whereas for the cells $k\in \ZZ^2$ which are inside $T_n$ and satisfy $k_1+k_2 =n$, we have $F(y)_{k}=x_{k}$, since $x_{k}=x_{k+Ne_1}=x_{k+Ne_2}=y_{k+Ne_1}=y_{k+Ne_2}$.
Iterating $F$ we obtain $\Delta\big(x,F^t(y)\big)\subseteq T_{n-t}$ for each $t\geq 0$.
Consequently, as time goes by, the set of disagreements becomes smaller.
In particular, for $t=n+1$, we obtain $\Delta\big(x,F^{n+1}(y)\big)\subseteq T_{-1}=\varnothing$, hence $F^{n+1}(y)=x\in X$.  This means that the configuration $y$ has been corrected in at most $n+1$ steps.
\end{proof}

\begin{remark}
The above result clearly extends to dimensions $d\geq 2$. One can simply apply Toom's majority rule on each (two-dimensional) sub-lattice generated by $N_1e_1$ and $N_2e_2$, and in the proof, replace $T_n$ by the triangular prism $T_n=\{k\in\ZZ^d : k_1+k_2\leq n,\; k_1,k_2\geq 0\}$.
\hfill\remarkqed
\end{remark}

\subsection{Single-cell fillable tiling spaces}
\label{sub:single_cell}

We say that a two-dimensional tiling space is \emph{single-cell fillable} if there exists a map $\psi\colon\Sigma^4\to \Sigma$ such that, for any possible choice $(a,b,c,d)\in \Sigma^4$ of symbols surrounding a cell (see Figure~\ref{fig:single-fillability}), assigning the value $\psi(a,b,c,d)$ to the central cell ensures that it is defect-free. More specifically, the value $\alpha=\psi(a,b,c,d)$ satisfies $v_1(a,\alpha)=v_1(\alpha,c)=1$ and $v_2(\alpha,b)=v_2(d,\alpha)=1$.
One can observe that a tiling space is single-cell fillable if and only if the functions $v_1$ and $v_2$ can be chosen in such a way that every locally admissible pattern is globally admissible~\cite{MaPa15}.

\begin{figure}[!h]
\vspace*{2mm}
\begin{center}
\scalebox{0.96}{\begin{tikzpicture}
	\figOneFill
\end{tikzpicture} }
\end{center}\vspace*{-4mm}
\caption{Illustration of the notion of single-cell fillability.}
\label{fig:single-fillability}
\end{figure}
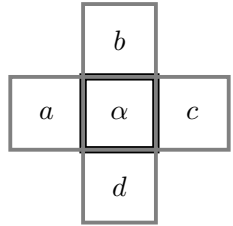

Note that any tiling space having a safe symbol is trivially single-cell fillable. The following is a more interesting family of single-cell fillable tiling spaces.

\begin{example}[Single-cell fillable]
For $k\geq 5$, the space $\Col_k$ of two-dimensional $k$-colourings is single-cell fillable.

\medskip
More generally, let $G=(V,E)$ be a finite undirected graph (possibly with self-loops).  Let $X_G\subseteq V^{\ZZ^d}$ denote the set of all graph homomorphisms from $\ZZ^d$ (with nearest-neigbour edges) to~$G$.  In other words, $x\in X_G$ if and only if $x_i x_j\in E$ for every two adjacent cells $i,j\in\ZZ^d$.  The set $X_G$ is clearly a tiling space, defined by the functions $v_1\isdef\cdots\isdef v_d\isdef v$ where $v(a,b)\isdef 1$ if and only if $ab\in E$.  Note that if $G$ is the complete graph on $k$ vertices (without self-loops), then $X_G$ coincides (up to renaming of the vertices) with the space of $k$-colourings~$\Col_k$.  It is easy to see that $X_G$ is single-cell fillable if and only if $G$ has the following property:  for any subset $A\subseteq V$ with $\abs{A}=2d$, there exists a vertex $b\in V$ such that $ab\in E$ for every $a\in A$.  In particular, $X_G$ is single-cell fillable if the minimum degree of the vertices in~$G$ is larger than $\big(1-(2d)^{-1}\big)\abs{V}$.
See Refs.~\cite{BrPa17,ChMa18} for more on the homomorphism spaces.
\hfill\exampleqed
\end{example}

Let $X$ be a single-cell fillable tiling space. We introduce the following terminology for defects. We say that a cell $(i,j)$ has a \emph{NE-defect} if it has a N-defect or an E-defect (or both).
For $x\in \Sigma^{\ZZ^2}$, we denote by $\DefNE(x)$ the set of cells having a NE-defect, that is:
\begin{align}
	\DefNE(x) &\isdef \big\{(i,j)\in\ZZ^2 : \text{$v_1(x_{c},x_{c+e_1})=0$ or $v_2(x_{c},x_{c+e_2})=0$}\big\} \;.
\end{align}

Let $ \psi\colon\Sigma^4 \to \Sigma $ be a function which assigns, to a choice $(a,b,c,d)\in \Sigma^4$ of symbols surrounding a cell, a symbol $\psi(a,b,c,d)$ for the central cell ensuring that it is defect-free.
We define a CA $F$ on $\Sigma^{\ZZ^2}$ by
\useshortskip
\begin{align}
	\forall c \in \ZZ^2, \quad
		F(x)_{c} \isdef
		\begin{cases}
			\psi(x_{c-e_1},x_{c-e_2}, x_{c+e_1},x_{c+e_2}) & \text{if $c\in\DefNE(x)$,} \\
			x_{c} & \text{otherwise.}
		\end{cases}
\end{align}

\begin{proposition}[Self-stabilisation of single-cell fillable tiling spaces]
Let $X\subseteq{\Sigma^{\ZZ^2}}$ be a single-cell fillable tiling space. Then, the CA $F\colon\Sigma^{\ZZ^2}\to\Sigma^{\ZZ^2}$ defined above stabilises $X$ from finite perturbations in linear time.
\end{proposition}

\begin{proof} It is clear from the definition that $F(x)=x$ for every $x\in X$. Let us now take $x\in \FPert{X}$.

First, observe that if $c\notin \DefNE(x),$ then $c\notin\DefNE(F(x))$, so that the set of NE-defects can only decrease under the action of $F$. Let us indeed take $c\notin \DefNE(x)$. By definition of $F$, the value of cell~$c$ is not modified when applying $F$, and if the value of cell $c+e_1$ (resp. $c+e_2$) is modified, that is, if $F(x)_{c+e_1}\not=x_{c+e_1}$ (resp. $F(x)_{c+e_2}\not=x_{c+e_2}$), then the new value of $c+e_1$ (resp. $c+e_2$) is chosen in such a way that $v_1\big(F(x)_c,F(x)_{c+e_1}\big)=1$ (resp. $v_2\big(F(x)_c,F(x)_{c+e_2}\big)=1$), so that $c\notin\DefNE(F(x))$.

Second, if $c\in\DefNE(x)$ is such that $c+e_1,c+e_2\notin\DefNE(x)$, then $c\notin \DefNE(F(x))$, so that the set of NE-defects is progressively eroded, from the NE to the SW. More formally, we can assume without loss of generality that there exists an integer $n\geq 0$ such that $\DefNE(x)\subseteq T_n$, where $T_n=\{(i,j)\in\ZZ^2 : i+j\leq n, \; i,j\geq 0\}$, and one can check that after $t$ steps, we have $\DefNE(F^t(x))\subseteq T_{n-t}$. Thus, after $n+1$ steps, we have $\DefNE(F^{n+1}(x))=\varnothing$, meaning that the configuration is fully corrected: $F^{n+1}(x)\in X$.
\end{proof}

\begin{remark}
The result extends naturally to $d$-dimensional single-cell fillable tiling spaces, with ${d\geq 2}$.
\hfill\remarkqed
\end{remark}

\subsection{Strongly $\ell$-fillable tiling spaces (with $\ell\geq 2$)}\label{sub:ell-fillable}

We say that a tiling space is \emph{strongly $\ell$-fillable} if there exists a map $\psi\colon\Sigma^{4\ell}\to \Sigma^{\ell^2}$ such that, for any possible choice $(a_1,\ldots,a_{4\ell})\in \Sigma^{4\ell}$ of symbols surrounding an $\ell$-square (that is, an $\ell\times\ell$ block of cells), assigning the values $(b_1,\ldots,b_{\ell^2})=\psi(a_1,\ldots,a_{4\ell})$ to the inner cells of the $\ell$-square ensures that each cell of the $\ell$-square is defect-free (see Figure~\ref{fig:EllCase}). Note that here, we do not assume any further condition on $(a_1,\ldots,a_{4\ell})\in \Sigma^{4\ell}$. We refer to the article by Alon et al.~\cite{AlBrChMaSp19} for a similar but weaker condition of $\ell$-fillability.

\begin{figure}[!h]\small
\begin{center}
\scalebox{0.95}{\begin{tikzpicture}
	\figEllCase
\end{tikzpicture} }
\end{center}\vspace*{-3mm}
\caption{Illustration of the notion of strong $\ell$-fillability.}
\label{fig:EllCase}\vspace*{-1mm}
\end{figure}
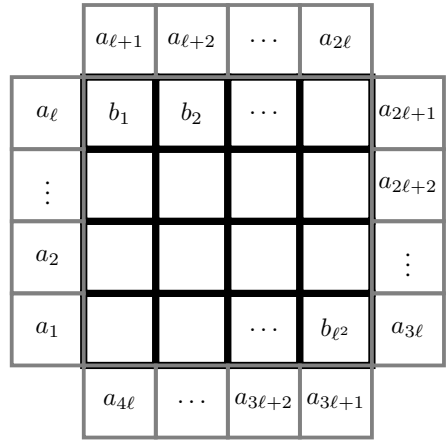

\begin{example}[Strongly $2$-fillable]
The set $\Col_4$ of two-dimensional $4$-colourings is not single-cell fillable.  We claim that $\Col_4$ is strongly $2$-fillable. To prove this, we need to show that for any possible choice $(a,b,c,d,e,f,g,h)\in \Sigma^8$ of symbols surrounding a $2$-square (see Figure~\ref{fig:lem:fourcol}), there exists a choice $(\alpha,\beta,\gamma,\delta)\in \Sigma^4$ for the cells of the $2$-square such that the four cells of the $2$-square are defect-free.

\begin{figure}[!h]
\vspace*{2mm}
\begin{center}
\scalebox{0.95}{\begin{tikzpicture}
	\figFourCase
\end{tikzpicture} }
\end{center}\vspace*{-4mm}
\caption{Illustration of the notion of strong $2$-fillability.}
\label{fig:lem:fourcol}\vspace*{-1mm}
\end{figure}
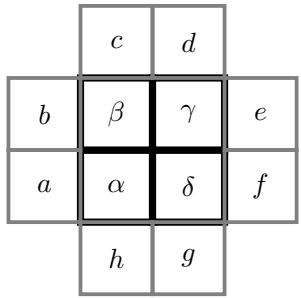

If $\{a,d,e,h\}\subsetneq \Sigma$, then we can choose a colour from $\Sigma\setminus\{a,d,e,h\}$ and assign it to both $\alpha$ and $\gamma$. We are then sure that we can find suitable colours for the two remaining cells, since each of these two cells is surrounded by at most three different colours.
In the same way, if $\{b,c,f,g\}\subsetneq \Sigma$, we can find a valid pattern.

Let us now assume that $\{a,d,e,h\}=\{b,c,f,g\}=\Sigma$. Without loss of generality, we can assume that $a=0, h=1, d=2, e=3$. The set of allowed colours for $\alpha$ is then $\{2,3\}$, and the set of allowed colours for $\gamma$ is $\{0,1\}$. If the allowed colours for $\beta$ and $\delta$ are $\{0,1\}$ and $\{2,3\}$ respectively, then a valid pattern is given by $(\alpha,\beta,\gamma,\delta)=(2,0,1,3)$. If the allowed colours for $\beta$ and $\delta$ are $\{0,2\}, \{1,3\}$ respectively, then a valid pattern is given by $(\alpha,\beta,\gamma,\delta)=(2,0,1,3)$. The other cases are analogous.~\hfill\exampleqed
\end{example}

\begin{example}[Strongly $2$-fillable]
Consider the set $\Theta$ of all Wang tiles whose edges are coloured either black or white, except the one with four white edges.
It is easy to see that the space of valid tilings with tiles from~$\Theta$ is strongly $2$-fillable but not single-cell fillable.
\hfill\exampleqed
\end{example}

\begin{example}[Strongly $2$-fillable]
\label{exp:paths}
Consider the set of Wang tiles depicted in Figure~\ref{fig:exp:paths}.  The decorations symbolise the edge colours, hence there are three possible colours: ``horizontal line'', ``vertical line'' and ``none''.
Let $X$ denote the space of all valid tilings.  It can be verified that $X$ is strongly $2$-fillable but not single-cell fillable.
\hfill\exampleqed
\end{example}

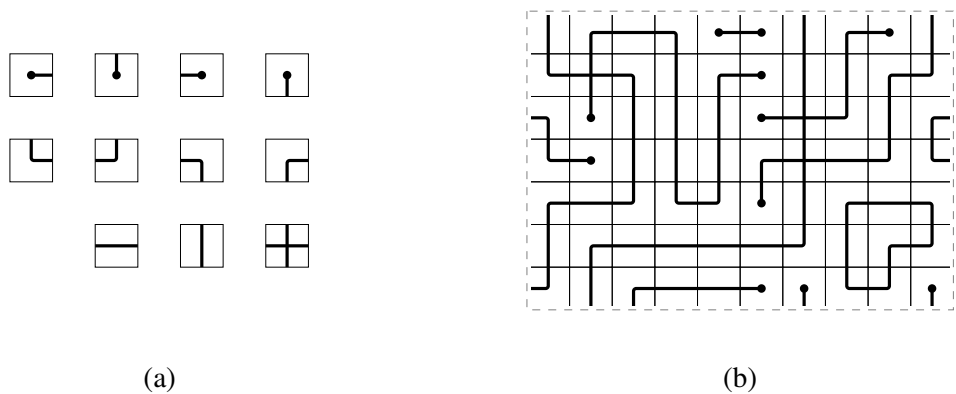
\begin{figure}[!h]
	\begin{center}
		\begin{tabular}{cc}
			\begin{minipage}{0.45\textwidth}
				\centering
				\begin{tikzpicture}[x=40pt,y=40pt,scale=0.8,every node/.style={scale=0.82}]
					\wtpathsR{(-1.5,1)}  \wtpathsU{(-0.5,1)}  \wtpathsL{(0.5,1)}  \wtpathsD{(1.5,1)}
					\wtpathsRU{(-1.5,0)} \wtpathsUL{(-0.5,0)} \wtpathsLD{(0.5,0)} \wtpathsDR{(1.5,0)}
					                     \wtpathsRL{(-0.5,-1)}\wtpathsUD{(0.5,-1)}\wtpathsX{(1.5,-1)}
				\end{tikzpicture}
			\end{minipage}
			&
			\begin{minipage}{0.45\textwidth}
				\centering
				\begin{tikzpicture}[x=20pt,y=20pt,scale=0.8,every node/.style={scale=0.82}]
					\begin{scope}
						\draw[help lines,dashed,gray] (-\tilesize+0\tilesize,\tilesize-0\tilesize) rectangle
							(19\tilesize-0\tilesize,-13\tilesize+0\tilesize);
						\clip (-\tilesize+0.2\tilesize,\tilesize-0.2\tilesize) rectangle
							(19\tilesize-0.2\tilesize,-13\tilesize+0.2\tilesize);
						\pathsexample
					\end{scope}
				\end{tikzpicture}
			\end{minipage}
			\\ [6em]
			(a) & (b)
		\end{tabular}
	\end{center}\vspace*{-5mm}
\caption{Illustration of Example~\ref{exp:paths}.  (a) The tile set.  (b) Example of a valid tiling.}
\label{fig:exp:paths}
\end{figure}

For $\ell=1$, strong $\ell$-fillability corresponds to single-cell fillability.  Therefore, in what follows, we will assume that $\ell\geq 2$.
The construction that we present below is slightly different from the one of the earlier article~\cite{FaMaTa19}, allowing for a simpler proof.  The proof from the latter article had some inaccuracies which are now avoided.

\medskip
Let us design a CA that corrects the finite perturbations of a strongly $\ell$-fillable tiling space~$X$.
Let $\psi\colon\Sigma^{4\ell}\to \Sigma^{\ell^2}$ be a function that maps some $(a_1,\ldots,a_{4\ell})\in \Sigma^{4\ell}$ to an element of $\Sigma^{\ell^2}$, such that the pattern formed by these values
is a defect-free pattern.
The aim is to apply $\psi$ on a collection of non-overlapping, non-adjacent $\ell$-squares containing defects so as to reduce the number of defects.
More specifically, the CA must (locally) select some $\ell$-squares containing defects in such a way that the following two conditions are satisfied:
\begin{enumerate}
\itemsep=0.95pt
\item Every two selected $\ell$-squares are at distance at least~$1$ (they do not overlap or touch each other).
\item If the configuration contains a defect, then at least one $\ell$-square containing a defect is selected.
\end{enumerate}
To this end, let us first identify a set of cells that will play the role of the top-right corners of the selected $\ell$-squares.
Given a configuration $x\in \Sigma^{\ZZ^2}$, let us denote again the set of cells having a NE-defect by
$\DefNE(x)=\big\{c\in\ZZ^2 : \text{$v_1(x_{c},x_{c+e_1})=0$ or $v_2(x_{c},x_{c+e_2})=0$}\big\}$.
For the sake of clarity, we first define a stabilising CA in the case $\ell=2$, and then treat the general case.

\paragraph{Case $\ell=2$.} We say that a cell $c\in \ZZ^2$ is a \emph{NE-corner} if $c\in\DefNE(x)$ and
\begin{align}
&c-e_1+e_2,c-2e_1+2e_2 \not\in\DefNE(x)\;,\\
&c+2e_1-e_2,c+e_1, c+e_2, c-e_1+2e_2 \not\in\DefNE(x)\;,\\
&c+2e_1,c+e_1+e_2, c+2e_2 \not\in\DefNE(x)\;, \\
&c+2e_1+e_2, c+e_1+2e_2 \not\in\DefNE(x)\;.
\end{align}
See Figure~\ref{fig:NeighbFourCol} for an illustration of the definition. We denote by $\CorNE(x)$ the set of NE-corners in a configuration $x\in \Sigma^{\ZZ^2}$.

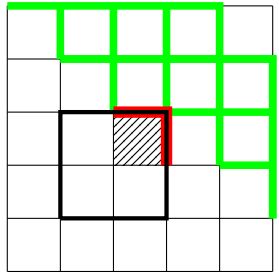
\begin{figure}[!h]\small
\begin{center}
\begin{tikzpicture}[scale=0.7]
	\figNeighbFourCol
\end{tikzpicture}
\end{center}\vspace*{-3.5mm}
\caption{Illustration of the notion of NE-corner in the case $\ell=2$.  The central cell is a NE-corner if at least one of the red lines (North or East or both) contains a defect and all the green lines are defect-free.}
\label{fig:NeighbFourCol}\vspace{1mm}
\end{figure}

We define a CA $F$ by the following rule: if a cell $c=(i,j)\in \ZZ^2$ is a NE-corner, then apply $\psi$ to the $2$-square whose NE-corner is $c$, that is, replace
the symbols of the cells $(i-1,j-1),(i-1,j),(i,j),(i,j-1)$ by $\psi(a,b,\ldots,h)$, where $a=x_{i-2,j-1}, b=x_{i-2,j},\ldots,h=x_{i-1,j-2}$ (see Figure~\ref{fig:lem:fourcol}). Let us first observe that the CA $F$ given by this rule is well-defined. Indeed, by definition of a NE-corner, one can check that there are no two adjacent NE-corners, vertically or horizontally, or in diagonal.
Consequently, at each time step, the $2$-squares that are updated do not overlap. In addition, note that the definition of NE-corners ensures that the $2$-squares that are updated cannot share a common edge, so that the property 1 above is satisfied.\vspace*{-1mm}

\paragraph{General case ($\ell\geq 2$).} In the general case, we modify the notion of NE-corner as follows. We say that a cell $c\in \ZZ^2$ is a \emph{NE-corner} if $c\in\DefNE(x)$ and for any $-\ell\leq i,j\leq \ell$,
\begin{align}
	\text{if $1\leq i+j\leq \ell-1$ or $[\text{$i+j=0$ and $j>i$}]$, then $c+ie_1+je_2\not\in\DefNE(x)$.}
\end{align}
See Figure~\ref{fig:NeighbFourColGeneral} for an illustration in the case $\ell=3$. We denote again by $\CorNE(x)$ the set of NE-corners in a configuration $x\in \Sigma^{\ZZ^2}$.

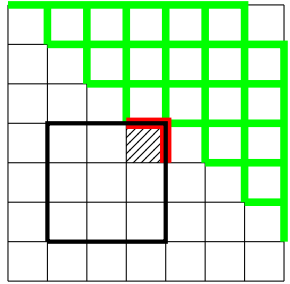
\begin{figure}[!h]
\begin{center}
\begin{tikzpicture}[scale=0.52]
	\figNeighbFourColGeneral
\end{tikzpicture}
\end{center}\vspace*{-2mm}
\caption{Illustration of the notion of NE-corner in the case $\ell=3$.  The central cell is a NE-corner if at least one of the red lines (North or East or both) contains a defect and all the green lines are defect-free.}
\label{fig:NeighbFourColGeneral}
\end{figure}

Again, we define a CA $F$ by the following rule: if a cell $c=(i,j)\in \ZZ^2$ is a NE-corner, then apply $\psi$ to the $\ell$-square whose NE-corner is $c$. By a similar argument as for $\ell=2$, $F$ is well-defined, and in addition, the $\ell$-squares that are updated cannot share a common edge, so that the property 1 above is satisfied.

We are now able to state the result, whose proof now consists in proving that the property 2 above is also satisfied.

\begin{proposition}[Self-stabilisation of strongly fillable tiling spaces]
\label{prop:ell-fillable}
Let $X\subseteq{\Sigma^{\ZZ^2}}$ be a strongly $\ell$-fillable tiling space, with $\ell\geq 2$.
Then, the CA $F\colon\Sigma^{\ZZ^2}\to\Sigma^{\ZZ^2}$ defined above stabilises $X$ from finite perturbations in quadratic time.
\end{proposition}

\begin{proof}
Let $x\in\FPert{X}$. We prove that if $\Def(x)\not=\varnothing$, then $\CorNE(x)\not=\varnothing$. Let us indeed sweep the configuration~$x$ by NW-SE diagonals, from the NE to the SW. Since $\Def(x)$ is finite, we can consider the first diagonal which contains a NE-defect, and on this diagonal, we consider the leftmost NE-defect (which is also the uppermost). By definition of a NE-corner, this NE-defect is a NE-corner.

In order to end the proof, it is then sufficient to observe that while $\Def(x)\not=\varnothing$, the number of defects decreases strictly when applying $F$. Indeed, at least all the NE-corners become defect-free, and since the updated $\ell$-squares do not share any common edge, no new defect is created.
\end{proof}

\begin{figure}[!h]
\begin{center}
\begin{tikzpicture}[scale=0.5]
	\figBlockTwo
\end{tikzpicture}
\end{center}\vspace*{-2mm}
\caption{Illustration of the proof of Proposition~\ref{prop:ell-fillable} in the case $\ell=2$.  Defects are represented in red, and cells that are NE-corners are shaded.  For the other NE-defects, the constraints represented in green on Figure~\ref{fig:NeighbFourCol} are not all satisfied.  As represented, we know that by construction, there exists at least one NE-corner on the first NW-SE diagonal that contains a NE-defect.}
\label{fig:BlockTwo}
\end{figure}
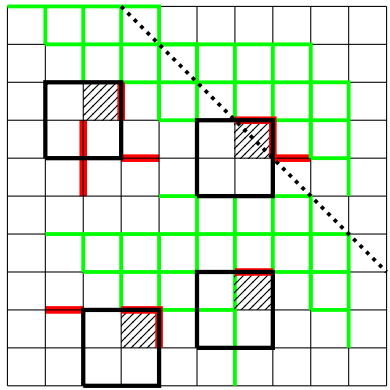

\subsection{Deterministic SFTs}
\label{sec:2d:deterministic}

We now extend the construction proposed in Section~\ref{sec:1d} for one-dimensional SFTs to a specific class of higher-dimensional SFTs: the \emph{deterministic} SFTs. To begin with, let us define this notion. Again, we focus on the two-dimensional case, but the ideas can also be applied in higher dimensions.

We say that a two-dimensional SFT $X\subseteq\Sigma^{\ZZ^2}$ is \emph{NE-deterministic} if it possible to describe it by a set~$\collection{F}$ of forbidden patterns in such a way that all the forbidden patterns of $\collection{F}$ have shape  $\{0,e_1,e_2\}$, and $\collection{F}$ has the additional property that for every $a,b\in \Sigma$, there exists at most one element $c\in \Sigma$ such that the pattern $p\colon\{0,e_1,e_2\}\to\Sigma$ defined by $p(e_1)=a$, $p(e_2)=b$ and $p(0)=c$ does not belong to~$\collection{F}$.

The constraints are therefore completely specified by the partial function $f\colon\Sigma\times \Sigma\to \Sigma$, that maps $(a,b)\in \Sigma^2$ to the unique symbol $c$ that is allowed, if any.

\begin{example}[Ledrappier SFT]
The SFT\vspace*{-2mm}
\begin{align}
	X &\isdef
		\{x\in \{\symb{0},\symb{1}\}^{\ZZ^2} : \text{$x_{k}=x_{k+e_1}+x_{k+e_2} \pmod*{2}$ for all $k\in\ZZ^2$}\}
\end{align}
known as the \emph{Ledrappier} SFT, or as the \emph{three-dot system}, is NE-deterministic.
\hfill\exampleqed
\end{example}

In case of SFTs identified by Wang tiles, NE-determinism means that for every two tiles $a,b\in \Sigma$, there is at most one tile $c=f(a,b)\in \Sigma$ such that the right edge of $c$ is compatible with the left edge of $a$ and the upper edge of $c$ matches the lower edge of $b$.

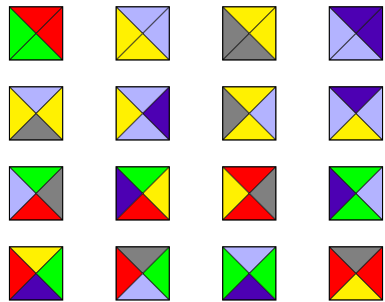
\begin{figure}[!b]
	\begin{center}
		{%
		\begin{tikzpicture}[x=20pt,y=15pt]
			\useasboundingbox (-5,-4) rectangle (5,4);
			
			\xxamtile{1}{(-3,3)}
			\xxamtile{2}{(-1,3)}
			\xxamtile{3}{(1,3)}
			\xxamtile{4}{(3,3)}
			
			\xxamtile{5}{(-3,1)}
			\xxamtile{6}{(-1,1)}
			\xxamtile{7}{(1,1)}
			\xxamtile{8}{(3,1)}
			
			\xxamtile{9}{(-3,-1)}
			\xxamtile{10}{(-1,-1)}
			\xxamtile{11}{(1,-1)}
			\xxamtile{12}{(3,-1)}
			
			\xxamtile{13}{(-3,-3)}
			\xxamtile{14}{(-1,-3)}
			\xxamtile{15}{(1,-3)}
			\xxamtile{16}{(3,-3)}
		\end{tikzpicture}
		}
	\end{center}\vspace*{-3mm}
\caption{%
	Ammann's set of Wang tiles is NE-deterministic and SW-deterministic.
}\label{fig:ammann}
\end{figure}

\begin{example}[Aperiodic deterministic tile sets]
The SFTs in two and higher dimensions can be significantly more complex than the one-dimensional SFTs.  For instance, in one dimension, every non-empty SFT contains periodic configurations, whereas in two dimensions, there are various examples of non-empty SFTs that do not contain periodic configurations.  Moreover, it is algorithmically undecidable whether the SFT identified by a given finite set of forbidden patterns is non-empty~\cite{Ber66,Rob71}.

Figure~\ref{fig:ammann} shows a set of $16$ Wang tiles, discovered by R.~Ammann, which is \emph{aperiodic}, meaning that it admits valid tilings but no periodic valid tilings~\cite[Chapter~11]{GrSh87}.  Remarkably, Ammann's tile set is both NE-deterministic and SW-deterministic.
The question of whether a given finite set of Wang tiles admits a valid tiling remains undecidable when restricted to NE-deterministic (or even four-way deterministic) tile sets~\cite{Kar92,Luk09}.
\hfill\exampleqed
\end{example}

As illustrated by the two examples above, the class of NE-deterministic SFTs is thus rich and multi-faceted.

In order to stabilise a NE-deterministic SFT, a natural first idea is to exploit the partial rule $f$ that defines the SFT.  Namely, the CA can apply the partial rule $f$ when $f$ is applicable and leave the state of the cell unchanged when $f$ is not applicable.
Figure~\ref{fig:ledrappier:toom-fails} illustrates why this approach fails on the Ledrappier SFT.  Indeed, the CA $F$ defined by $F(x)_k\isdef x_{k+e_1} + x_{k+e_2} \pmod*{2}$ cannot stabilise any finite perturbation of the all-$\symb{0}$ configuration.

\begin{figure}[!h]
\vspace*{2mm}
	\begin{center}
		\begin{tikzpicture}[x=2*\ltcellsize,y=-2*\ltcellsize,>=stealth']
			\ledrappiertoom
			\foreach \i in {0,1,...,9} {
				\draw[help lines] (\i-0.5,-0.8) -- +(0,9.6);
			}
			\foreach \j in {0,1,...,9} {
				\draw[help lines] (-0.8,\j-0.5) -- +(9.6,0);
			}
		\end{tikzpicture}
	\end{center}\vspace*{-4mm}
\caption{%
	Illustration of how the naive approach may fail on a NE-deterministic SFT such as the Ledrappier SFT.  In the depicted example, the initial configuration contains a single $\symb{1}$ (the red cell) in a sea of $\symb{0}$s.  Although each $\symb{1}$ is eventually turned into~$\symb{0}$, the two foremost $\symb{1}$s (the yellow cells) always lead to more $\symb{1}$s.
}\label{fig:ledrappier:toom-fails}\vspace*{-2mm}
\end{figure}

\begin{theorem}[Self-stabilisation of NE-deterministic SFTs]
\label{prop:robust:2d:deterministic}
For every two-dimensional NE-deterministic SFT $X\subseteq\Sigma^{\ZZ^2}$, there exists a CA $F\colon{\Sigma'}^{\ZZ^2}\to{\Sigma'}^{\ZZ^2}$ with ${\Sigma'}\supseteq\Sigma$ that stabilises $X$ from finite perturbations in linear time.
\end{theorem}

We construct a CA $T\colon\Sigma'^{\ZZ^2}\to \Sigma'^{\ZZ^2}$
with an extended alphabet $\Sigma'\supseteq \Sigma$ that corrects finite
islands of defects on $X$ in linear time.
The construction is similar to the one-dimensional case (Section~\ref{sec:1d}),
except that the patching rule is replaced by the map $f$ stemming from the fact that $X$ is a deterministic tiling space.

\medskip
Let us thus consider a NE-deterministic SFT $X\subseteq \Sigma^{\ZZ^2}$, identified by a partial map $f\colon\Sigma\times \Sigma\to \Sigma$.
For a configuration $y\in \Sigma^{\ZZ^2}$, we denote by
\begin{align}
	\Def(y) &\isdef \left\{
		(i,j): y_{i,j} \neq f(y_{i+1,j},y_{i,j+1})
	\right\}
\end{align}
the set of cells at which a \emph{defect} occurs.

As before, the extended alphabet will be $\Sigma'\isdef \Sigma\times\{\xnone,\xtrac,\xstop\}$,
in which $\{\xnone\}\times \Sigma$ is identified with $\Sigma$,
and $\xtrac$ and $\xstop$ represent, respectively, the \emph{trace} and \emph{stop} signals.
Again, for convenience, we identify the configurations with alphabet $\Sigma'$
with pairs $(y,\alpha)$ where $y\in \Sigma^{\ZZ^2}$ and $\alpha\in\{\xnone,\xtrac,\xstop\}^{\ZZ^2}$.
The configuration $(y,\alpha)$ has a \emph{defect} at cell $(i,j)$ if
either $y$ has a defect at $(i,j)$, or $(i,j)$ contains a signal symbol $\xtrac$ or $\xstop$.
Extending the notation $\Def(y)$, we denote the set of defects on $(y,\alpha)$ by
\begin{align}
	\Def(y,\alpha) &\isdef \{(i,j)\in\ZZ: \text{$(i,j)\in \Def(y)$ or $\alpha_{i,j}\neq\xnone$}\} \;.
\end{align}

The CA $T$ will be, as before, constructed as a composition
$T_2(T_gT_1^2T_0)^2$
of four CA maps $T_g,T_0,T_1,T_2\colon\allowbreak\Sigma'^\ZZ\to \Sigma'^\ZZ$.

\medskip

\noindent\emph{Patching}.
The map $T_g$ is defined by $T_g(y,\alpha)\isdef(y',\alpha')$, where
\begin{align}
	(y'_{i,j}, \alpha'_{i,j}) &\isdef \begin{cases}
		\left(f(y_{i+1,j},y_{i,j+1}),\xtrac\right)\qquad
			&\text{if $\alpha_{i+1,j}\neq\xstop$, $\alpha_{i,j+1}\neq\xstop$,} \\
			&\text{$(i,j)\in \Def(y,\alpha)$ and $f(y_{i+1,j},y_{i,j+1})$ exists,} \\
		\left(y_{i,j},\alpha_{i,j}\right)
			&\text{otherwise.}
	\end{cases}
\end{align}
It simply replaces the symbol $y_{i,j}$ with the one prescribed
by its North and East neighbours (if none of them contains a stop mark)
and leaves a trace mark behind.

\medskip

\noindent\emph{Generation of stop signals}.
The map $T_0$ is responsible for generating stop signals,
and is defined by $T_0(y,\alpha)\isdef(y,\alpha')$ (i.e., no change on $y$), where
\begin{align}
	\alpha'_{i,j} &\isdef \begin{cases}
		\xstop\qquad
			&\text{if $\alpha_{i,j}=\xtrac$ and $(i,j)\in\Def(y)$,} \\
		\alpha_{i,j}
			&\text{otherwise.}
	\end{cases}
\end{align}

\medskip

\noindent\emph{Propagation of stop signals}.
The propagation of the stop signals is governed by the map $T_1$,
which is defined by $T_1(y,\alpha)\isdef(y,\alpha')$ (i.e., no change on $y$), where
\begin{align}
	\alpha'_{i,j} &\isdef \begin{cases}
		\xstop\qquad
			&\text{if ($\alpha_{i+1,j}=\xstop$ or $\alpha_{i,j+1}=\xstop$) and $\alpha_{i,j}=\xtrac$,} \\
		\alpha_{i,j}
			&\text{otherwise.}
	\end{cases}
\end{align}

\medskip

\noindent\emph{Fading of the traces}.
Finally, the map $T_2$ handles the fading of the traces
and is defined by $T_2(y,\alpha)\isdef(y,\alpha')$ (i.e., no change on $y$), where
\begin{align}
	\alpha'_{i,j} &\isdef \begin{cases}
		\xnone\qquad
			&\text{if $\alpha_{i+1,j}=\alpha_{i,j+1}=\xnone$ and $\alpha_{i,j}=\xtrac$,} \\
		\alpha_{i,j}
			&\text{otherwise.}
	\end{cases}
\end{align}
The composition $T_2(T_gT_1^2T_0)^2$ ensures that the fading of the traces
is half as slow as the correction speed, while the stop signals propagate
twice as fast as the correction speed.

\subsubsection*{Proof of Theorem~\ref{prop:robust:2d:deterministic}:}
The proof is similar to the proof of Theorem~\ref{prop:robust:1d}. \QED

\section{Self-stabilisation of isotropic probabilistic CA}
\label{sec:PCA}

All the cellular automata discussed above provide directional solutions: the cells need to distinguish between the four directions North, South, East, West. In general, finding self-stabilising CA that respect the symmetries of the tiling space appears to be a difficult problem, and not always possible.
For instance, in case of the homogeneous space~$\Hom_2$, Pippenger has shown that self-stabilisation cannot be achieved by a monotone, self-dual (i.e., respecting the $\symb{0}\leftrightarrow\symb{1}$ symmetry), centrosymmetric rule~\cite{Pip94}.

\medskip
In this section, we will examine to which extent the use of randomness in the evolution of the cellular automata may provide us with a means to design {\em simpler} solutions.
More precisely, we now study probabilistic CA that achieve self-stabilisation with \emph{nearest-neighbour} \emph{isotropic} rules, that is, rules with von~Neumann neighbourhood which treat the neighbours ``equally''.  Our aim is to show that the use of randomness can extend the range of possibilities for designing self-stabilising processes.

Since this is a broad topic, we will restrict our scope to two examples: finite SFTs and single-cell fillable tiling spaces. In Section~\ref{sub:disc_proba}, we will discuss some further questions related to other families of tiling spaces.

\subsection{Setting}

To begin with, let us recall the notion of \emph{probabilistic CA}, and formulate the concept of self-stabilisation in this context.

\paragraph{Probabilistic cellular automata.}
The specificity of \emph{probabilistic CA} is that the outcome of the local rule is now a probability distribution on $\Sigma$, and the cells of the lattice are updated simultaneously and independently at each time step, according to the distributions prescribed by the local rule.

Formally, the local rule in this case is thus given by a function $\varphi\colon\Sigma^{\Neighb}\to\PS(\Sigma)$, where $\PS(\Sigma)$ denotes the set of probability distributions on $\Sigma$, and where we still denote by $\Neighb\subseteq \ZZ^d$ the (finite) neighbourhood of the rule.

\medskip
We describe the evolution of the system as a time-homogeneous Markov chain $(\RV{x}^t)_{t\in\NN}$ with values in $\Sigma^{\ZZ^d}$, such that for any finite $C\subseteq\ZZ^d$, and for any  $x^0, \ldots, x^t\in\Sigma^{\ZZ^d}$, we have:
\begin{align}
\PP(\RV{x}^{t+1}_C = x^{t+1}_C \; | \; \RV{x}^0=x^0, \dots, \RV{x}^t=x^t)
& = \PP(\RV{x}^{t+1}_C=x^{t+1}_C \; | \; \RV{x}^t_{C+\Neighb}=x^t_{C+\Neighb} ) \\
& = \prod_{c\in C}\varphi((x^t_{c+i})_{i\in\Neighb}) (\{ x^{t+1}_c \}). \vspace*{-7mm}
\end{align}

\paragraph{Self-stabilisation.} We say that a  probabilistic CA \emph{stabilises} a tiling space $X\subseteq\Sigma^{\ZZ^d}$ from finite perturbations if
\begin{enumerate}[label={\roman*)}]
	\item (\emph{consistency}) the configurations of $X$ are absorbing states, that is, if $\RV{x}^{t}\in X$, then $\RV{x}^{t+1}=\RV{x}^{t}$,
	\item (\emph{attraction}) finite perturbations of the elements of $X$ evolve almost surely to $X$ in finitely many steps, that is, if $\RV{x}^{0}\in \FPert{X}$, then there exists almost surely a time $\RV{t}\in\NN$ such that $\RV{x}^{\RV{t}}\in X$.
\end{enumerate}

The first such $\RV{t}$ is called the \emph{stabilisation time} (or the \emph{recovery time}) starting from $\RV{x}^{0}$. Note that $\RV{t}$ is now a random variable.  We say that $F$ stabilises $X$ from finite perturbations \emph{in time $\tau(n)$} if for each $n\in\NN$,
the maximum of the expected stabilisation time $\xExp[\RV{t}]$ among all possible (deterministic) initial configurations~$\tilde{x}\in\Sigma^{\ZZ^d}$ with $\delta(\tilde{x},X)=n$ is~$\tau(n)$.

\subsection{Finite tiling spaces}
\label{sec:pca:finite}

Let us first examine the case where the tiling space is the set $\Hom_2=\{\unifO,\unifI\}\subseteq\{\symb{0},\symb{1}\}^{\ZZ^2}$ of Example~\ref{ex:homog}. Toom's North-East-Center majority rule (see Example~\ref{exp:toom}) is a deterministic CA that stabilises~$\Hom_2$. We now present an isotropic self-stabilising probabilistic CA for this tiling space.

\medskip
Let $\Sigma\isdef\{\symb{0},\symb{1}\}$, and consider the probabilistic CA ${\MRIE}$ on $\Sigma^{\ZZ^2}$, defined on the von~Neumann neighbourhood $\vonNeumann\isdef\{0,e_1,-e_1,e_2,-e_2\}$  by the local rule $\varphi\colon\Sigma^{\Neighb}\to\PS(\Sigma)$ where
\begin{align}
	\varphi((x_{i})_{i\in\vonNeumann}) &\isdef
		\begin{cases}
			\delta_{\symb{1}}	& \text{ if $x_{e_1}+x_{e_2}+x_{-e_1}+x_{-e_2} > 2$,} \\
			\delta_{\symb{0}}	& \text{ if $x_{e_1}+x_{e_2}+x_{-e_1}+x_{-e_2} < 2$,} \\
			\Bern(1/2) 			& \text{otherwise,}
		\end{cases}
\end{align}
where $\delta_{\symb{1}}$ and $\delta_{\symb{0}}$ are the Dirac distributions on~$\symb{1}$ and $\symb{0}$, respectively, and $\Bern(1/2) $ denotes the Bernoulli random variable with parameter $1/2$.
In words, at every step, the state of each cell is changed to the state which is in majority among its four adjacent cells, and in case of a tie, the tie is broken with a flip of a fair coin, independently of the other cells.
A few sample snapshots from the evolution of {$\MRIE$} are shown in Figure~\ref{diag:MRIE}.
The continuous-time version of~{$\MRIE$} was studied by Fontes, Schonmann and Sidoravicious~\cite{FoScSi02}
(see Examples~\ref{exp:MRIE:continuous-time:finite} and~\ref{exp:MRIE:continuous-times} below).

\begin{figure}[!h]
\vspace*{3mm}
\hspace*{-2mm}\begin{tabular}{c c c c}
\includegraphics[width=0.215\linewidth]{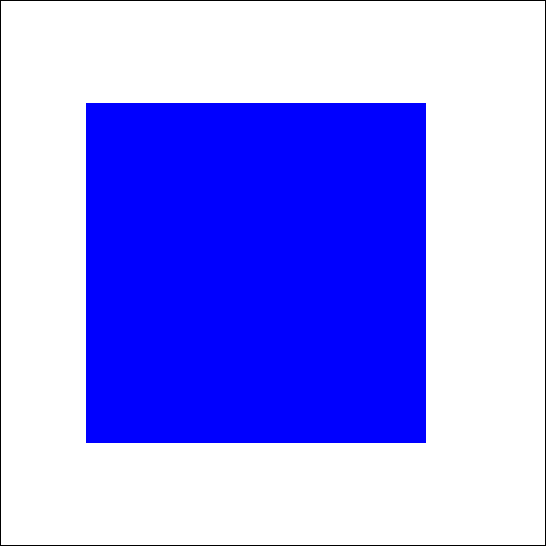} &
\includegraphics[width=0.215\linewidth]{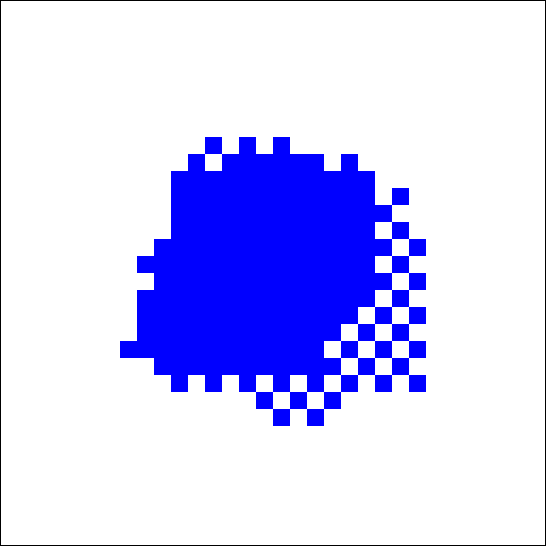} &
\includegraphics[width=0.215\linewidth]{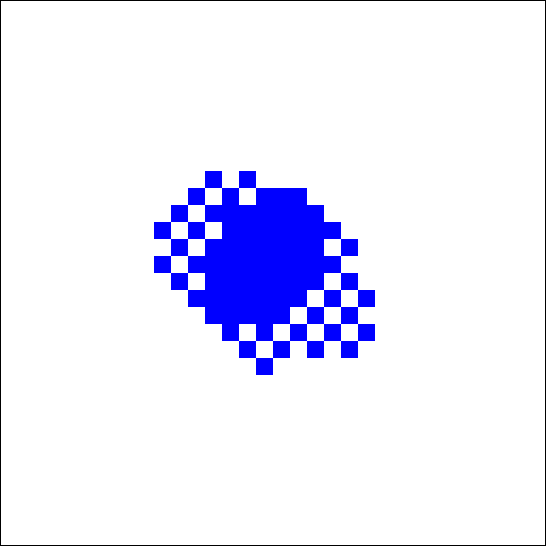} &
\includegraphics[width=0.215\linewidth]{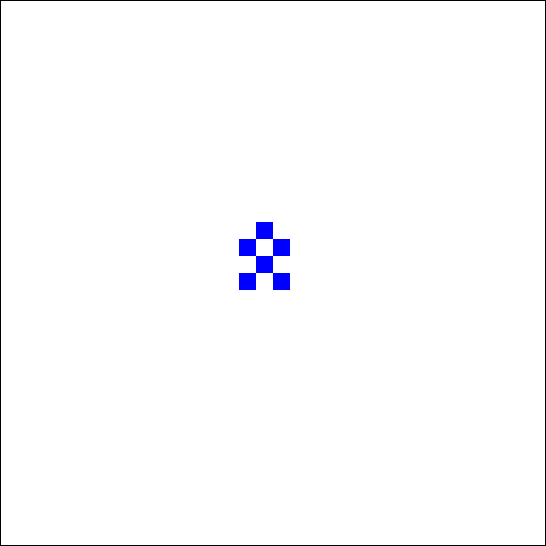} \\
$ t = 0 $ & $ t= 100 $ & $ t= 150 $ & $ t =200 $
\end{tabular}\vspace{-1mm}
\caption{Snapshots from the evolution of the $\MRIE$ rule.}
\label{diag:MRIE}\vspace*{-3mm}
\end{figure}

\begin{proposition}[Isotropic self-stabilisation of~$\Hom_2$]
\label{prop:hom_prob}
The probabilistic CA defined above stabilises $\Hom_2$ from finite perturbations in at most cubic time.
\end{proposition}

\begin{proof} Let $x\in \FPert{\Hom_2}$, and let us assume that the defects of $x$ are initially included in a rectangle~$R$.
By symmetry, we can consider that the defects are $\symb{1}$s and that the system needs to return to the all-$\symb{0}$ configuration.
Observe that over time, the defects always stay within $R$.

We first determine an \emph{upper bound} for the average time it takes for the $\symb{1}$s of the upper row of $R$ to disappear. Let us number from left to right by $1, \ldots, k$ the cells of the upper row of $R$. We also consider the cell $0$ which is on the left of cell $1$ and the cell ${k+1}$ which is on the right of cell $k$.

We bound the evolution of the cells $1,\ldots,k$ by a new process, designed by imagining that these cells evolve in an environment where for each cell $i$ ($1\leq i\leq k$), its North neighbour is in state $\symb{0}$ and its South neighbour in state $\symb{1}$ (see Figure~\ref{fig:majflipequal}).  Because of the monotonicity of the local rule~$\varphi$, the new process can be coupled with the original process in such a way that the state of the cells $1,\ldots, k$ in the original process remain dominated by their states in the new process (i.e., wherever the former has a $\symb{1}$, so does the latter).

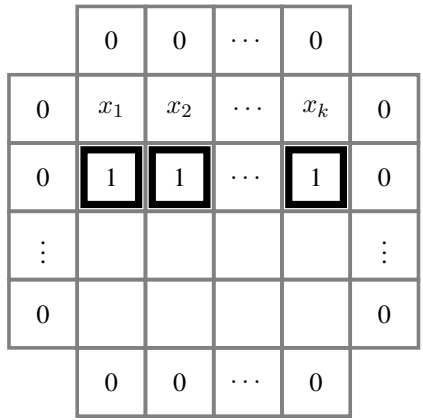
\begin{figure}[!h]
\vspace*{1mm}
\begin{center}
\scalebox{0.9}{\begin{tikzpicture}
\crownMaj
\end{tikzpicture} }
\end{center}\vspace*{-3mm}
\caption{Study of the evolution of the upper row of a rectange of defects under the $\MRIE$ probabilistic CA (see proof of Proposition~\ref{prop:hom_prob}).}
\label{fig:majflipequal}
\end{figure}

Since the two North and South neighbours have their state fixed, the evolution of the cells $1,\ldots,k$ can be modelled as a one-dimensional probabilistic CA $(\RV{y}^t)_{t\in\NN}$  with neighbourhood radius~$1$ and fixed boundary conditions $\RV{y}^t_{0}=\RV{y}^t_{k+1}=\symb{0}$.

Let us analyse the evolution of this one-dimensional probabilistic CA, whose behaviour can be observed in Figure~\ref{fig:diploid}.
The local rule of this one-dimensional CA is given in the following table:\medskip

\begin{center}
\scalebox{0.98}{\begin{tabular}{| r | c c c c c c c c |}
\hline
Value of the neighbourhood & $\symb{000}$ & $\symb{001}$ & $\symb{010}$ & $\symb{011}$ & $\symb{100}$ & $\symb{101}$ & $\symb{110}$ & $\symb{111}$\\
\hline
Probability of symbol $\symb{1}$ &  0  & \pa &  0  & \pa & \pa &  1  & \pa &  1 \\
\hline
\end{tabular} }
\end{center}

\medskip\noindent Observe that this local rule does not depend on the value of the cell itself, but only of its left and right neighbours. Indeed, if the neighbourhood is in state $(x,y,z)$, the new state of the central cell is equal to $x$ if $x=z$ and to a random value with distribution $\Bern(1/2)$ otherwise. This observation allows us to describe the evolution of this CA as the combination of two independent processes,
one on the even space-time positions and the other on the odd space-time positions.

Consider first the process on the \emph{odd} space-time positions.  In this process, the position of the leftmost cell in state~$\symb{1}$ is bounded from below by a symmetric random walk that is reflected on the left boundary (cell~$0$) and vanishes when it reaches the right boundary (cell~$k+1$).  One can show that the expected time this random walk needs to reach $k+1$ is of order $k^2$. Indeed, by a standard argument, if $T_i$ denotes the expected time needed to reach $k+1$ from cell $i$, then we have the recursion $T_i=1+{(T_{i-1}+T_{i+1})}/2$, with the boundary conditions $T_1=1+T_2$ and $T_{k+1}=0$.  It follows that $T_1$ is quadratic in $k$.

The same result holds for the process on the \emph{even} space-time positions.  Since the time needed for the one-dimensional CA to reach the all-$\symb{0}$ configuration is the maximum of the times of the two processes on the odd an even space-time positions, the former time is also quadratic.  In particular, the time it takes for the original two-dimensional CA to wipe out the first row of defects from rectangle~$R$ is quadratic in the diameter of the rectangle~$R$.
Finally, since we have a linear number of lines to wipe out, the stabilisation time is at most cubic.
\end{proof}

\begin{figure}[!t]
\begin{center}
\scalebox{1}[-1]{\includegraphics[width=0.25\linewidth]{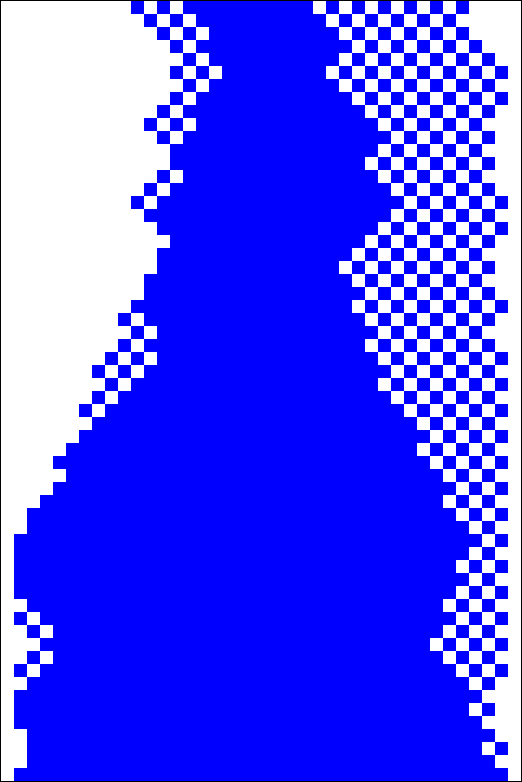}}
\scalebox{1}[-1]{\includegraphics[width=0.25\linewidth]{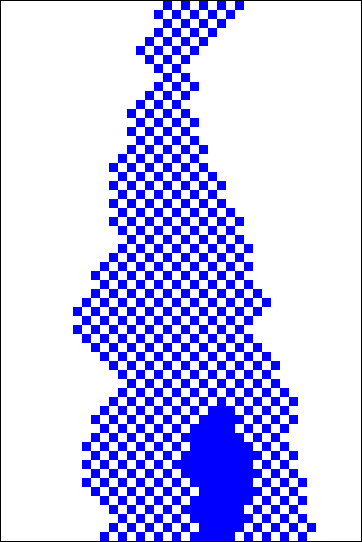}}
\scalebox{1}[-1]{\includegraphics[width=0.25\linewidth]{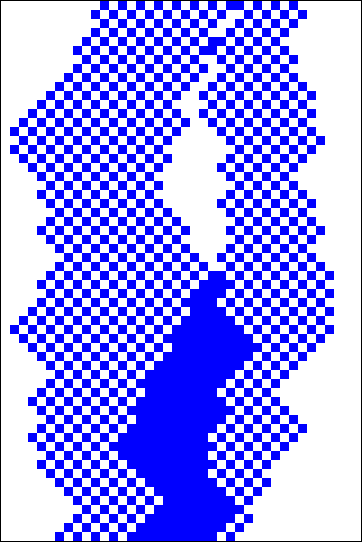}}
\end{center}
\caption{Space-time diagrams showing evolutions of the one-dimensional probabilistic CA appearing in the proof of Proposition~\ref{prop:hom_prob}, with fixed boundary conditions. Time goes from top to bottom. Blue and white squares respectively represent states $\symb{0}$ and $\symb{1}$. (left) evolution from the all-one configuration; (middle) evolution showing how the central homogeneous zone disappears; (right) evolution with various appearances and reappearances of the central homogeneous zone.}
\label{fig:diploid}
\end{figure}

Experimental evidence suggests that the stabilisation is faster than cubic.
Note that in the above argument, we did not use the fact that the different rows of the rectangle evolve simultaneously.  We conjecture that the true stabilisation time is in fact quadratic.  Our intuition is further supported by the following known result regarding the analogous model in continuous-time.

\begin{example}[Continuous-time~$\MRIE$]
\label{exp:MRIE:continuous-time:finite}
Fontes, Schonmann and Sidoravicious studied the continuous-time variant of ${\MRIE}$ in which the cells are updated asynchronously,$\,$triggered by independent$\,$Poisson clocks with rate$\,1$\cite{FoScSi02}.\footnote{%
In fact, Fontes et al.\ considered a more general family of local rules in which, in case of a neighbourhood tie, the current state of the cell is flipped with probability~$0<\alpha\leq 1$.  The ${\MRIE}$ rule corresponds to $\alpha=1/2$.
}
Their Theorem~1.3 states that, in dimension~$d\geq 2$, the system stabilises~$\Hom_2$ from finite perturbations,
and that the stabilisation occurs in time~$\bigo(n^d)$, in the sense that there exist constants $c,\gamma>0$ such that the stabilisation time~$\RV{t}$ starting from a configuration $\tilde{x}\in\{\symb{0},\symb{1}\}^{\ZZ^d}$ with $\delta(\tilde{x},\Hom_2)=n$ satisfies $\xPr(\RV{t}> c n^d)\leq\ee^{-\gamma n}$.~\hfill\exampleqed
\end{example}\vspace*{-4mm}

\paragraph{Extension to finite SFTs.} We can extend the above isotropic probabilistic rule to stabilise any finite SFT in at most cubic time. Indeed, if we define $N_1$ and $N_2$ as in Section~\ref{sub:finite}, then the following rule is suitable: if a state appears strictly more than twice among $x_{a+N_1,b},x_{a,b+N_2},x_{a-N_1,b},x_{a,b-N_2}$, then this state becomes the new value for $x_{a,b}$;
otherwise, choose one of the latter values uniformly at random (taking into account the multiplicities).
Again, all the defects stay within some enveloping rectangle and are eventually corrected, and a similar argument shows that the stabilisation time is at most cubic.

\subsection{Single-cell fillable tiling spaces}

The second case we examine is the one of single-cell fillable tiling spaces, such as $k$-colourings for $k\geq 5$ (see Section~\ref{sub:single_cell}). Again, a simple isotropic probabilistic rule can then be designed.  Interestingly, the new rule turns out to stabilise faster than the one in Section~\ref{sub:single_cell}.

\medskip
Let $X\subseteq\Sigma^{\ZZ^2}$ be a single-cell fillable tiling space, and let $\psi\colon \Sigma^4 \to \Sigma $ be a function as in Section~\ref{sub:single_cell}, assigning, for any possible choice $(a, b, c, d)\in \Sigma^4$ of symbols surrounding a cell, a consistent value $\psi(a,b,c,d)$.
For $x\in \Sigma^{\ZZ^2}$, recall that we denote by $\Def(x)$ the set of cells having a defect, that is,
\begin{align}
	\Def(x) &\isdef
		\{k\in\ZZ^2 : \text{$\exists e\in\{\pm e_1, \pm e_2\}$, $k$ has a defect in direction $e$}\} \;.
\end{align}
Given $\alpha\in(0,1)$, we define a probabilistic CA on $X$ which leaves the state of cell $k$ unchanged if $k\notin \Def(x)$ and changes it to $\psi(x_{k-e_1},x_{k-e_2}, x_{k+e_1},x_{k+e_2})$ with probability~$\alpha$ if $k\in \Def(x)$.
In other words, the CA has a von~Neumann neighbourhood $\vonNeumann\isdef\{0,e_1,-e_1,e_2,-e_2\}$ and local rule
\begin{align}
	\varphi((x_{i})_{i\in\vonNeumann}) &\isdef
		\begin{cases}
			\alpha \delta_{\psi(x_{-e_1},x_{-e_2}, x_{+e_1},x_{+e_2})}+(1-\alpha)\delta_{x_0}
								& \text{if $(x_{i})_{i\in\vonNeumann}$ contains a defect,} \\
			\delta_{x_0}		& \text{otherwise.}
		\end{cases}
\end{align}
Note that if $\psi$ is isotropic, then so is $\varphi$.

\medskip
To estimate the stabilisation time of this CA, we will use the following standard lemma.
Let us recall that a \emph{geometric} random variable with parameter $\varepsilon\in[0,1]$ is a discrete random variable $\RV{r}$ with possible values $k=1,2,3,\ldots$ such that $\xPr(\RV{r}=k)=(1-\varepsilon)^{k-1}\varepsilon$.
\begin{lemma}[Expectation of maximum of i.i.d.\ geometric random variables]
\label{lem:geom:max}
For every $\varepsilon\in (0,1)$,
there exist constants $a,b>0$ such that, whenever $\RV{r}_1,\RV{r}_2,\ldots,\RV{r}_n$ are independent geometric random variables with parameter~$\varepsilon$, we have $\xExp[\max\{\RV{r}_1,\RV{r}_2,\ldots,\RV{r}_n\}]\leq a\log n + b$.
\end{lemma}

\subsubsection*{Proof sketch:}
This can be shown via a simple comparison with the maximum of $n$ independent exponential random variables, for which the expected value can be calculated explicitly (see e.g.~\cite{Eis08}).\QED
\eject 

\begin{proposition}[Isotropic self-stabilisation of single-cell fillable tiling spaces]
\label{prop:ell-fillable:probabilistic}
Let $X\subseteq\Sigma^{\ZZ^2}$ be a single-cell fillable tiling space.
Then, for every $\alpha\in(0,1)$, the probabilistic CA defined above stabilises~$X$ from finite perturbations in at most logarithmic time.
\end{proposition}

\begin{proof}
Let $x\in\FPert{X}$ be an initial configuration. Let $\RV{x}^0,\RV{x}^1,\ldots$ denote the Markov process described by the probabilistic CA starting from $\RV{x}^0=x$.
Note that if a cell $k$ is non-defective at time~$t$ (i.e., $c\notin\Def(\RV{x}^t)$), then it will remain non-defective at time~$t+1$, because according to the local rule, the state of $k$ will not change, and by the property of~$\psi$, changes in the state of the neighbours of~$k$ cannot create a defect at~$k$.
Therefore, for every~$t$ we have $\Def(\RV{x}^{t+1})\subseteq \Def(\RV{x}^t)$.
On the other hand, at every step, every defective cell has probability at least $\alpha(1-\alpha)^4$ of becoming non-defective.  Indeed, consider a cell $k\in\Def(\RV{x}^t)$.
If at time~$t+1$, cell $k$ is updated according to~$\psi$ and none of its four neighbours are updated according to~$\psi$, then $k$ becomes non-defective, and this event occurs with probability at least $\alpha(1-\alpha)^4$.
It follows that inside~$\Def(x)$, the probabilistic CA behaves like an absorbing finite-state Markov chain that eventually reaches a configuration with no defects.

\medskip
To analyse the stabilisation time, let us imagine that the process is constructed using a collection $(\RV{z}_k^t)_{c\in\ZZ^2,t\in\NN}$ of independent Bernoulli random variables with parameter~$\alpha$, representing the random choices taken at every time step and each cell.  Namely,
for $t=0,1,2,\ldots$ and $k\in\ZZ^2$, we have
\begin{align}
	\RV{x}^{t+1}_k &\isdef
		\begin{cases}
			\psi(\RV{x}^t_{k-e_1},\RV{x}^t_{k-e_2}, \RV{x}^t_{k+e_1},\RV{x}^t_{k+e_2})
				& \text{if $\RV{z}_k^{t+1}=\symb{1}$ and $k\in\Def(\RV{x}^t)$,} \\
			\RV{x}^t_k	& \text{otherwise.}
		\end{cases}
\end{align}
For each $k\in\ZZ^2$, let
\begin{align}
	\RV{t}_k &\isdef \inf\big\{t>0: \text{$\RV{z}_k^t=\symb{1}$ and
		$\RV{z}^t_{k+e_1}=\RV{z}^t_{k-e_1}=\RV{z}^t_{k+e_2}=\RV{z}^t_{k-e_2}=\symb{0}$}\big\}
\end{align}
denote the first time that $k$ is updated according to~$\psi$ and none of its neighbours are updated according to~$\psi$.  This is a geometric random variable with parameter~$\varepsilon\isdef\alpha(1-\alpha)^4$.
As we observed above, $k\notin\Def(\RV{x}^t)$ for all $t\geq \RV{t}_k$.  In particular, defining $\RV{t}_A\isdef\max\{\RV{t}_k: k\in A\}$ for $A\subseteq\ZZ^2$, we have $\RV{x}^t\in X$ for all $t\geq\RV{t}_{\Def(x)}$.  We show that $\xExp[\RV{t}_{\Def(x)}]$ is logarithmic in~$\abs{\Def(x)}$, and hence also in $\delta(\tilde{x},X)$.

\begin{figure}[!b]
\begin{center}
\begin{tikzpicture}[scale=0.52] 
	\begin{scope}
		\clip (-0.2,-0.2) rectangle (10.2,6.2);
		\draw[help lines] (-1,-1) grid (11,7);
		
		\foreach \i in {-2,...,10}
			\foreach \j in {-1,...,7} {
				\node at ($(0.5,0.5)+(\i+\j,-3*\i+2*\j)$) {$0$};
				\node at ($(1,0)+(0.5,0.5)+(\i+\j,-3*\i+2*\j)$) {$1$};
				\node at ($(0,1)+(0.5,0.5)+(\i+\j,-3*\i+2*\j)$) {$2$};
				\node at ($(-1,0)+(0.5,0.5)+(\i+\j,-3*\i+2*\j)$) {$4$};
				\node at ($(0,-1)+(0.5,0.5)+(\i+\j,-3*\i+2*\j)$) {$3$};
			}
	\end{scope}
\end{tikzpicture}
\end{center}\vspace*{-4mm}
\caption{%
A partitioning of~$\ZZ^2$ into $5$ parts in such a way that the cells in each part do not share neighbours with one another.
}
\label{fig:lattice:partition}
\end{figure}

\medskip
The random variables $\RV{t}_k$ (for $k\in\ZZ^2$) are not independent. However, if $k_1,k_2,\ldots$ are cells that do not share neighbours, then $\RV{t}_{k_1},\RV{t}_{k_2},\ldots$ are independent.
Note that $\ZZ^2$ can be partitioned into five parts $Q_0,\ldots,Q_4$ in such a way that the cells in each part do not share neighbours with one another (see Figure~\ref{fig:lattice:partition}).

\medskip
Therefore, setting $\Def_i(x)\isdef\Def(x)\cap Q_i$, we have \vspace*{-1mm}
\begin{align}
	\RV{t}_{\Def(x)} &= \max\{\RV{t}_{\Def_i(x)}: i=0,1,\ldots,4\} \leq
		\sum_{i=0}^4 \RV{t}_{\Def_i(x)} \vspace*{-1mm}
\end{align}
and hence $\xExp[\RV{t}_{\Def(x)}]\leq\sum_{i=0}^4\xExp[\RV{t}_{\Def_i(x)}]$.
By Lemma~\ref{lem:geom:max}, $\xExp[\RV{t}_{\Def_i(x)}]=\bigo(\log\abs{\Def_i(x)})$.
The claim follows.

\medskip
An alternative argument for the logarithmic stabilisation time can be given by considering an appropriate martingale as in Ref.~\cite[Lemma~6]{Fat20}.
\end{proof}

\section{Complexity of self-stabilisation}
\label{sec:complexity}

In this section, we consider the complexity of self-stabilisation as a computational task.
One can consider at least three different measures of complexity: the speed of stabilisation,
the number of extra symbols, the size of the neighbourhood (Section~\ref{sec:complexity:measures}).
Our focus in this article is on the speed of stabilisation, with a preference for having no extra symbols.
In Section~\ref{sec:complexity:hardness}, we show that there is an SFT whose self-stabilisation problem is inherently hard, in the sense that it requires super-polynomial time (unless $\classP=\classNP$).
Interestingly, the optimal speed of stabilisation turns out to be a topological invariant:
if two SFTs are topologically isomorphic, then their stabilisation requires roughly the same amount of time (Section~\ref{sec:complexity:conjugacy}).  Here, we will be focusing on the deterministic setting, and leave it open whether the same results hold in the setting of probabilistic cellular \linebreak automata.

\subsection{Measures of complexity}
\label{sec:complexity:measures}

The problem of designing a cellular automaton that stabilises an SFT $X$ is an algorithmic problem, albeit a parallel one with extra requirements.
The examples discussed so far suggest that the complexity of this algorithmic problem could drastically vary with $X$.
The efficiency of a CA $F$ in stabilising $X$ can be judged based on the resources it uses:
\begin{description}
	\item[Speed of stabilisation $\tau_F(n)$]\leavevmode\\
		How fast does $\tau_F(n)$ grow with $n$?
		Recall that $\tau_F(n)$ denotes the maximum time it takes for $F$ to correct a finite perturbation $\tilde{x}$ with $\delta(\tilde{x},X)=n$.
	\item[Number of extra symbols $\kappa_F$]\leavevmode\\
		How many extra states per cell does $F$ have compared to the alphabet of~$X$?
	\item[Neighbourhood radius $r_F$]\leavevmode\\
		How far does the local rule of $F$ need to look in order to update the state of one cell?
\end{description}
The complexity of self-stabilisation for $X$ can be measured by the optimal values of $\tau_F(n)$, $\kappa_F$ and~$r_F$.  We let $\,\kappa_*\isdef\min_F\kappa_F\,$ and $\,r_*\isdef\min_F r_F,\,$ where the minimums are over all CA $F$ that stabilise~$X$.
\eject

\noindent We also allegorically use $\tau_*(n)$ to indicate the optimal (in order of magnitude) speed of stabilisation among all CA $F$ that stabilise~$X$.\footnote{%
	Note that this is not a rigorous notation.  For instance, it might be that for each $\varepsilon>0$, there exists an $F$ with $\tau_F(n)=\bigo(n^{1+\varepsilon})$ but no $F$ with $\tau_F(n)=\bigo(n)$, in which case $\tau_*(n)$ is not well defined.
}

\medskip
Let us make a couple of remarks about these measures:
\begin{enumerate}
	\item The stabilisation can be linearly sped up at the cost of increasing the neighbourhood radius.  Namely, if $F$ has neighbourhood radius $r_F$ and stabilises $X$ in time $\tau_F(n)$, then $F^k$ has neighbourhood radius $k r_F$ and stabilises $X$ in time $\lceil\frac{1}{k}\tau_F(n)\rceil$.
		In particular, the value of $\tau_*(n)$ is meaningful only up to a multiplicative constant.
	\item As we will see in Proposition~\ref{prop:conjugacy-invariant}, the speed of stabilisation $\tau_*(n)$ is (almost) invariant under topological isomorphisms.  Namely, if two SFTs $X$ and $Y$ are topologically isomorphic, then their minimum stabilisation speeds are roughly the same.
		We suspect that no such invariance holds for the minimum number of extra symbols $\kappa_*$ or for the minimum neighbourhood radius~$r_*$: for every SFT $X$, it should be possible to find an isomorphic SFT $Y$ for which $\kappa_*=0$ and $r_*=1$.
\end{enumerate}

In this paper, we have focused on the speed of stabilisation~$\tau(n)$ with preference towards having no extra symbols.

\subsection{Topologically isomorphic SFTs}
\label{sec:complexity:conjugacy}

In this section, we show that the optimal stabilisation time $\tau_*(n)$ for an SFT is an isomorphism invariant, meaning that topologically isomorphic SFTs have roughly the same optimal stabilisation times.

To prove this claim, let us first recall some terminology from symbolic dynamics (see the monograph by Lind and Marcus~\cite{LiMa95} for more details).

The space $\Sigma^{\ZZ^d}$ of all $d$-dimensional configurations with symbols from a finite alphabet $\Sigma$ is a compact metrisable space with the product topology.  The shift maps $\sigma^k\colon\Sigma^{\ZZ^d}\to\Sigma^{\ZZ^d}$ (for $k\in\ZZ^d$) are all continuous.
A closed subset $X\subseteq\Sigma^{\ZZ^d}$ is called a \emph{shift space} if it is invariant under all shifts, that is, $\sigma^k x\in X$ for all $x\in X$ and $k\in\ZZ^d$.  Clearly, every SFT is a shift space.
Given a shift space $X$ and a finite set $M\subseteq\ZZ^d$, we define $L_M(X)\isdef\{x_M: x\in X\}$ as the set of all patterns with shape $M$ that appear in $X$.

\medskip
A \emph{homomorphism} between two shift spaces $X\subseteq\Sigma^{\ZZ^d}$ and $Y\subseteq\Gamma^{\ZZ^d}$ is a continuous map $\Phi\colon X\to Y$ that commutes with the shifts, that is, $\Phi\oo\sigma^k = \sigma^k\oo\Phi$ for every $k\in\ZZ^d$.  It can be verified that a map $\Phi\colon X\to Y$ is a homomorphism if and only if it is realized by a local rule, that is, if and only if there exists a finite set $\Neighb\subseteq\ZZ^d$ and a map $\varphi\colon L_\Neighb(X)\to\Gamma$ such that $\Phi(x)_i=\varphi\big(\sigma^i(x)_\Neighb\big)$ for each $x\in X$ and $i\in\ZZ^d$.
A map $\Phi\colon X\to Y$ is an \emph{isomorphism} if it is bijective and both $\Phi$ and $\Phi^{-1}$ are homomorphisms.  Since every shift space is compact and Hausdorff, every bijective homomorphism is in fact an isomorphism.  Two shift spaces $X$ and $Y$ are \emph{topologically isomorphic} (or \emph{topologically conjugate}) if there is an isomorphism between them.

\begin{example}[Two isomorphic SFTs]
	Let $X\subseteq\{\symb{0},\symb{1}\}^{\ZZ}$ and $Y\subseteq\{\symb{0}_{\symb{0}},\symb{0}_{\symb{1}},\symb{1}\}^{\ZZ}$ be the $1$-step SFTs whose transition graphs are depicted in Figure~\ref{fig:vertex-shifts:isomorphic}.  It is easy to verify that the map $\Phi:Y\to X$ given by\vspace{2mm}
	\useshortskip
	\begin{align}
		\Phi(y)_i &\isdef
			\begin{cases}
				\symb{0}	& \text{if $y_i\in\{\symb{0}_{\symb{0}},\symb{0}_{\symb{1}}\}$,} \\
				\symb{1}	& \text{if $y_i=\symb{1}$,}
			\end{cases}
	\end{align}
	is an isomorphism.
	Observe that $\symb{0}$ is a safe symbol for $X$ (see Example~\ref{ex:hardcore}), whereas $Y$ does not have a safe symbol.  While the CA $F$ defined in~\eqref{eq:safe-symbol:stabiliser} (Section~\ref{sec:setting:self-stab}) stabilises~$X$ in one step, it is not immediately clear if $F$ can be ``translated'' into a CA that stabilises~$Y$.  Indeed, $Y$ has more symbols than~$X$, and as a result, a configuration in $Y$ can have many more finite perturbations than the corresponding configuration in~$X$.  Nevertheless, having the correspondence between $X$ and $Y$ in mind, one can re-implement the mechanism of stabilisation by~$F$ (i.e., replacing each defect with the safe symbol) to obtain a CA that stabilises~$Y$.  Namely, the CA $G:\{\symb{0}_{\symb{0}},\symb{0}_{\symb{1}},\symb{1}\}^{\ZZ}\to\{\symb{0}_{\symb{0}},\symb{0}_{\symb{1}},\symb{1}\}^{\ZZ}$ defined by
	\begin{align}
		G(y)_i &\isdef
			\begin{cases}
				\symb{0}_{\symb{0}}		&
					\text{if $y_{i-1}y_i=\symb{1}\symb{1}$ or $y_iy_{i+1}\in\{\symb{1}\symb{1},\symb{0}_{\symb{1}}\symb{0}_{\symb{0}}, \symb{0}_{\symb{1}}\symb{0}_{\symb{1}}\}$ or $y_iy_{i+1}y_{i+2}=\symb{0}_{\symb{1}}\symb{1}\symb{1}$,} \\
				\symb{0}_{\symb{1}}		&
					\text{if $y_iy_{i+1}y_{i+2}\in\{\symb{0}_{\symb{0}}1\symb{0}_{\symb{0}},\symb{0}_{\symb{0}}1\symb{0}_{\symb{1}}\}$,} \\
				y_i	& \text{otherwise,}
			\end{cases}
	\end{align}
	stabilises $Y$ (in one step), and is in some sense the ``translation'' of $F$.
	\hfill\exampleqed
\end{example}

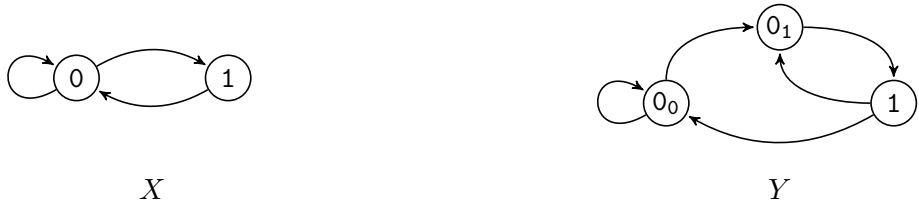
\begin{figure}[h]
\vspace*{-3mm}
\begin{center}
	\begin{tabular}{cc}
		\begin{minipage}{0.49\textwidth}
			\centering
			\begin{tikzpicture}[->,>=stealth',shorten >=1pt,semithick]
				\tikzstyle{every state}=[minimum size=17pt,inner sep=0pt]
				
				\node[state] (O) at (-1,0) {$\symb{0}$};
				\node[state] (I) at (1,0) {$\symb{1}$};
				
				\path (O) edge [out=-150,in=150,loop,overlay] (O);
				\path (O) edge [out=30,in=150] (I);
				\path (I) edge [out=-150,in=-30] (O);
			\end{tikzpicture}
		\end{minipage}
		&
		\begin{minipage}{0.49\textwidth}
			\centering
			\begin{tikzpicture}[->,>=stealth',shorten >=1pt,semithick]
				\tikzstyle{every state}=[minimum size=17pt,inner sep=0pt]
				
				\node[state] (OO) at (-1.5,0) {$\symb{0}_{\symb{0}}$};
				\node[state] (OI) at (0,1) {$\symb{0}_{\symb{1}}$};
				\node[state] (I) at (1.5,0) {$\symb{1}$};
				
				\path (OO) edge [out=-150,in=150,loop,overlay] (OO);
				\path (OO) edge [out=90,in=180] (OI);
				\path (OI) edge [out=0,in=90] (I);
				\path (I) edge [out=180,in=-90] (OI);
				\path (I) edge [out=-150,in=-30] (OO);
			\end{tikzpicture}
		\end{minipage}%
		\medskip \\
		$X$ & $Y$
	\end{tabular}
\end{center}\vspace*{-5mm}
\caption{%
	Graphs of allowed transitions for two isomorphic one-dimensional $1$-step SFTs.
}
\label{fig:vertex-shifts:isomorphic}\vspace*{-2mm}
\end{figure}

\begin{example}[Higher block presentation]
	Let $X\subseteq\Sigma^{\ZZ^d}$ be an SFT and $M\subseteq\ZZ^d$ be a finite set with at least two elements.
	Then, the map $\Phi:X\to(\Sigma^M)^{\ZZ^d}$ defined by $\Phi(x)_i\isdef x_{i+M}$ is an isomorphism between $X$ and $Y\isdef\Phi(X)$.  As in the previous example, the isomorphism between $X$ and $Y$ does not extend to a correspondence between the finite perturbations of $X$ and the finite perturbations of $Y$, and it is not immediately clear how a CA stabilising $X$ from finite perturbations can be translated into a CA stabilising~$Y$ from finite perturbations.
	\hfill\exampleqed
\end{example}

The following proposition shows that a CA stabilising an SFT $X$ from finite perturbations can always be translated into a CA stabilising an isomorphic SFT $Y$ from finite perturbations (if we appropriately extend the alphabet of $Y$) while keeping the stabilisation time roughly unchanged.

\begin{proposition}[Stabilisation time for isomorphic SFTs]
\label{prop:conjugacy-invariant}
	Let $X$ and $Y$ be two topologically isomorphic SFTs.
	If there exists a CA that stabilises $X$ from finite perturbations in time $\tau(n)$,
	then there also exists a CA that stabilises $Y$ from finite perturbations in time $\tau\big(n+\bigo(1)\big)$.
\end{proposition}
\begin{proof}
	Let $\Sigma$ and $\Gamma$ be two finite alphabets such that $X\subseteq\Sigma^{\ZZ^d}$ and $Y\subseteq\Gamma^{\ZZ^d}$.
	Let $\Phi\colon X\to Y$ be the isomorphism between $X$ and $Y$, and denote its inverse by~$\Psi$.
	Let $F\colon\Sigma^{\ZZ^d}\to\Sigma^{\ZZ^d}$ be a CA that stabilises $X$ from finite perturbations in time~$\tau(n)$.

	Let us first assume that $\Phi$ can be extended to a one-to-one shift-invariant continuous map $\hat{\Phi}\colon\Sigma^{\ZZ^d}\to\hat{\Gamma}^{\ZZ^d}$, where $\hat{\Gamma}$ is a finite alphabet including $\Gamma$.
	Let $\hat{\Psi}\colon\hat{\Phi}(\Sigma^{\ZZ^d})\to\Sigma^{\ZZ^d}$ denote the inverse of~$\hat{\Phi}$, and observe that $\hat{\Psi}$ is an extension of $\Psi$, that is, $\hat{\Psi}|_Y=\Psi$.  Let $\tilde{\Psi}\colon\hat{\Gamma}^{\ZZ^d}\to\Sigma^{\ZZ^d}$ be a shift-invariant continuous extension of $\hat{\Psi}$.  Such an extension can be constructed by (arbitrarily) completing the local rule of $\hat{\Psi}$.  Let us now define $G\colon\hat{\Gamma}^{\ZZ^d}\to\hat{\Gamma}^{\ZZ^d}$ by $Gy\isdef\hat{\Phi}F\tilde{\Psi}y$.  To see that $G$ stabilises $Y$ from finite perturbations, first note that for every $y\in\hat{\Gamma}^{\ZZ^d}$, the diagram
	\begin{equation}
		\begin{tikzpicture}[xscale=2,yscale=1.8,>=stealth',->,baseline=(current bounding box.center)]
			\node (x) at (0,1) {$x$};\node (Fx) at (1,1) {$Fx$};\node (F2x) at (2,1) {$F^2x$};
				\node (F3x) at (3,1) {$\ldots$};
			\node (y) at (0,0) {$y$};\node (Gy) at (1,0) {$Gy$};\node (G2y) at (2,0) {$G^2y$};
				\node (G3y) at (3,0) {$\ldots$};
			
			\draw (x) to node[above]{$F$} (Fx);
			\draw (Fx) to node[above]{$F$} (F2x);
			\draw (F2x) to node[above]{$F$} (F3x);
			\draw (y) to node[below]{$G$} (Gy);
			\draw (Gy) to node[below]{$G$} (G2y);
			\draw (G2y) to node[below]{$G$} (G3y);
			
			\draw (y) to node[right]{$\tilde{\Psi}$} (x);
			\draw (Gy) to[bend right=15] node[right]{$\tilde{\Psi}$} (Fx);
			\draw (G2y) to[bend right=15] node[right]{$\tilde{\Psi}$} (F2x);
			
			\draw (Fx) to[bend right=15] node[left]{$\hat{\Phi}$} (Gy);
			\draw (F2x) to[bend right=15] node[left]{$\hat{\Phi}$} (G2y);
		\end{tikzpicture}\vspace{-2mm}
	\end{equation}
	commutes.
	If $y\in Y$, then clearly $Gy=y$.
	Suppose that $y\in\FPert[\hat{\Gamma}]{Y}$, that is, a finite perturbation of an element of $Y$ in $\hat{\Gamma}^{\ZZ^d}$.
	Then, $x\isdef\hat{\Phi}y\in\FPert[\Sigma]{X}$.
	More specifically, let $\bar{y}\in Y$ be such that $\Delta(\bar{y},y)$ is finite with diameter $n$.
	Then, $\bar{x}\isdef\tilde{\Psi}\bar{y}=\Psi\bar{y}$ is in $X$ and the diameter of $\Delta(\bar{x},x)$ is at most $n+C$, where $C$ is the diameter of the neighbourhood of the local rule of $\tilde{\Psi}$.
	Following the above diagram, the iterations of $G$ on $y$ correspond to the iterations of $F$ on~$x$.
	Since $F$ stabilises $X$, we have $F^t x\in X$ for some $t\leq\tau(n+C)$.
	But as soon as $F^tx\in X$, we also obtain $G^t y = \hat{\Phi}F^t x = \Phi F^t x \in Y$.  It follows that $G$ stabilises $Y$ from finite perturbations in time $\tau(n+C)$.
	
\medskip
	It remains to construct a one-to-one shift-invariant continuous extension $\hat{\Phi}\colon\Sigma^{\ZZ^d}\to\hat{\Gamma}^{\ZZ^d}$ of $\Phi$.
	Let $\varphi\colon L_N(X)\to\Gamma$ and $\psi\colon L_M(Y)\to\Sigma$ be the local rules of $\Phi$ and $\Psi$ respectively.  Without loss of generality, we assume that $\Gamma$ and $\Sigma$ are disjoint.  Let $\hat{\Gamma}\isdef\Gamma\cup\Sigma\cup(\Gamma\times\Sigma)$.
	Define $\hat{\Phi}$ with the local rule $\hat{\varphi}\colon\Sigma^{M+N}\to\hat{\Gamma}$, given by
	\begin{align}
		\hat{\varphi}(p) &\isdef
			\begin{cases}
				\varphi(p_N)
				& \text{if $p\in L_{M+N}(X)$,} \\
				(\varphi(p_N),p_0)
				& \text{if $p\notin L_{M+N}(X)$ but $p_N\in L_N(X)$,} \\
				p_0	& \text{otherwise.}
			\end{cases}
	\end{align}
	The first case in the definition ensures that $\hat{\Phi}|_X=\Phi$.
	Furthermore, $p_0$ can always be recovered from~$\hat{\varphi}(p)$, either directly, or by applying $\psi$ on $\big(\varphi(\sigma^i(p)_N)\big)_{i\in M}$, which can be extracted from $\hat{\varphi}(p)$.
	This means that $\hat{\Phi}$ is indeed a one-to-one extension of $\Phi$ as required.
\end{proof}

\subsection{Super-polynomial hardness}
\label{sec:complexity:hardness}

In this section, we prove the following theorem:

\begin{theorem}[SFT with slow stabilisation]
\label{thm:hard-stabilisation}
	Let $d\geq 2$.
	Unless $\classP=\classNP$, there exists a $d$-dimensional SFT $X$ which is not stabilised from finite perturbations by any CA in polynomial time.
\end{theorem}

Let us emphasize that we do not know whether every SFT can be stabilised by some CA or not (see below, Problem~\ref{q:open:general-solution} of Section~\ref{sec:discussion:other-open}).  According to the above theorem, assuming $\classP\neq\classNP$, there exists an SFT for which either there is no solution whatsoever or every solution requires super-polynomial stabilisation time.  The proof of Proposition~\ref{prop:global-tiling-patching:NP-hard} below contains an explicit construction of such an SFT, but we do not prove that the constructed SFT is stabilised by some CA.

For concreteness, we present the proof for the two-dimensional case.  The proof of the higher-dimensional case carries through similarly.

\medskip
The \emph{square tiling problem} of a finite set of Wang tiles $\Theta$ is the decision problem of whether a square of size $n$ can be tiled admissibly using~$\Theta$ in such a way as to achieve a prescribed colouring of its boundary.
The square tiling problem of every tile set is clearly in class $\classNP$.  The existence of a set of Wang tiles for which the square tiling problem is $\classNP$-complete is folklore.

The idea of the proof of Theorem~\ref{thm:hard-stabilisation} is that a CA that stabilises the SFT $X$ associated to a set of Wang tiles $\Theta$ can be used to solve a variant of the square tiling problem for~$\Theta$.  Namely, suppose that $F$ is a CA that stabilises $X$ in time $\tau(n)$, and let $r$ be the neighbourhood radius of~$F$.  Then, given a configuration $x\in\FPert{X}$ and a finite set $A\subseteq\ZZ^2$ such that $x_{\ZZ^2\setminus A}$ is globally admissible, the cellular automaton is able to ``patch'' the defects of $x$ in $\tau(\diam(A))$ steps and turn it into a valid tiling by only changing the states of the cells that are no farther than $r\tau(n)$ from $A$.  Note that only the states of the cells within distance $r\tau(n)$ from $A$ can possibly change during this computation, and only the states of the cells within distance $2r\tau(\diam(A))$ from $A$ are relevant for such changes.  Thus, the relevant parts of the computation can be simulated by a Turing machine.
As we will see, such a Turing machine can then be used to solve the standard square tiling problem.

To be specific, let us state the latter problem more explicitly.
Let $\Moore\isdef\{-1,0,1\}^2$ denote the Moore neighbourhood.  The \emph{(Moore) boundary} of a set $A\subseteq\ZZ^2$ is the set $\partial\Moore(A)\isdef\Moore(A)\setminus A$.
For each~$n\in\NN$, let $S_n\isdef\{0,1,\ldots,n-1\}^2$.  Thus, for $k\geq 0$, $\Moore^k(S_n)$ represents the square $\{-k,-k+1\ldots,n+k-1\}^2$.
\begin{algorithmicproblem}[Global tiling patching]\leavevmode
\label{aprob:tiling-patching}
	\begin{description}[style=multiline, topsep=0pt, itemsep=0pt, parsep=0pt, font=\normalfont\sffamily, leftmargin=\widthof{Parameters:~~~~~~~~~}]
		\item[Parameters:] A finite set of Wang tiles $\Theta$ and two functions $\alpha,\beta\colon\NN\to\NN$.
		\item[Input:] A globally admissible tiling $q$ of $\Moore^{\alpha(n)+\beta(n)+1}(S_n)\setminus S_n$.
		\item[Task:] Find a globally admissible pattern $\tilde{q}$ on $\Moore^{\alpha(n)+1}(S_n)$ that agrees with $q$ on the band $\partial\Moore(\Moore^{\alpha(n)}(S_n))$.
	\end{description}
\end{algorithmicproblem}
\noindent See Figure~\ref{fig:global-tiling-patching:problem} for an illustration.

\begin{figure}[h]
\vspace*{-1mm}
	\begin{center}
		\begin{tabular}{ccc}
			\begin{minipage}{0.5\textwidth}
				\centering
				\begin{tikzpicture}[x=10pt,y=10pt]
					\globalpatchinginput
				\end{tikzpicture}
			\end{minipage}%
			& $\mapsto$ &
			\begin{minipage}{0.4\textwidth}
				\centering
				\begin{tikzpicture}[x=10pt,y=10pt]
					\globalpatchingoutput
				\end{tikzpicture}
			\end{minipage}%
			\\ [8em]
			(a) & & (b)
		\end{tabular}
	\end{center}\vspace*{-5mm}
\caption{%
	Illustration of the global tiling patching problem.
	(a) The input is a globally admissible pattern on the shaded region.
	(b) The output is a globally admissible pattern on the shaded region which agrees with the input on the region which is shaded in blue.
}\label{fig:global-tiling-patching:problem}\vspace*{-2mm}
\end{figure}
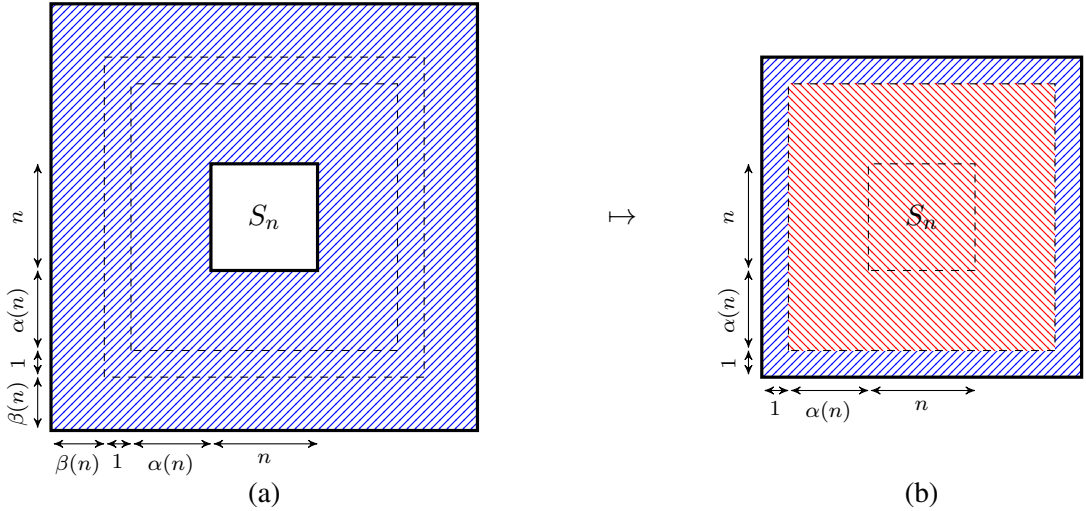

\begin{proposition}[Serial simulation]
\label{prop:serial-simulation}
	Let $X$ be the SFT of the valid tilings of a finite set of Wang tiles~$\Theta$.  Suppose $X\neq\varnothing$ and that there exists a CA with neighbourhood radius $r$ that stabilises $X$ in time~$\tau(n)$.  Then, there exists a Turing machine with a two-dimensional tape that solves the global tiling patching problem for~$\Theta$ and $\alpha(n)\isdef\beta(n)\isdef r\tau(n)$ in time $\bigo\big([n+4r\tau(n)]^2\tau(n)\big)$.
\end{proposition}
\begin{proof}
	The Turing machine simulates the CA on its two-dimensional tape.  Once $\tau(n)$ steps of the CA simulation have been carried out, the Turing machine stops and returns the current configuration of the CA.
	
	A serial simulation of the CA can be done in a standard fashion.  The tape has two layers (numbered~$0$ and~$1$) for storing the configurations at even and odd time steps.  The input is initially provided on layer~$0$.  In an initialization stage, the input pattern is extended by writing an arbitrary symbol from~$\Theta$ at every position in $S_n$.
	When simulating time step~$t$ of the CA, the Turing machine reads the configuration at time $t-1$ from layer~$(t-1)\bmod 2$ and writes the result on layer~$t\bmod 2$.  At positions in which not enough information is available
	(i.e., outside $\Moore^{2r\tau(n)+1-rt}(S_n)\setminus\Moore^{rt}(S_n)$)
	the blank symbol is written instead.
	
	In order to keep of track of the number of simulated steps, the cells in
	$\Moore^{2r\tau(n)+1}(S_n)\setminus S_n$
	are, in the initialization stage, marked with $\star$.  At every stage of the simulation, the marks are also updated: a mark is kept if all its $(2r+1)^2$ neighbours are marked and is erased otherwise.  After $\tau(n)$ stages of the simulation, only the cells in the band
	$\Moore^{r\tau(n)+1}(S_n)\setminus\Moore^{r\tau(n)}(S_n)$
	are marked and thus the Turing machine easily recognizes that it has to enter its final stage of computation.  In the final stage, the marks and the distinction between the two layers are erased so that only the configuration at time $\tau(n)$ is left on the tape.
	
	The simulation of each time step of the CA takes $\bigo\big([n+4r\tau(n)]^2\big)$ time steps of the Turing machine.
	Likewise, the initialization stage and the final stage take only $\bigo\big([n+4r\tau(n)]^2\big)$ time steps each.
	Thus, in overall, the Turing machine accepts or rejects its input in $\bigo\big([n+4r\tau(n)]^2\tau(n)\big)$ steps.	
\end{proof}

\begin{proposition}[$\classNP$-hardness]
\label{prop:global-tiling-patching:NP-hard}
	There exists a finite set of Wang tiles $\Theta$ such that for every two functions $\alpha,\beta\colon\NN\to\NN$ with polynomial growth, the global tiling patching problem associated to $\Theta$, $\alpha$ and~$\beta$ is $\classNP$-hard.
\end{proposition}
\begin{proof}
	Let $\Theta_0$ be a finite set of Wang tiles for which the square tiling problem is $\classNP$-complete, and let $C_0$ denote the set of colours that appear in the tiles of $\Theta_0$.
	Let $\Theta_1$ be the set of Wang tiles depicted in Figure~\ref{fig:global-tiling-patching:NP-hard:extra-tiles} and their four-fold rotations.  Without loss of generality, we assume
	\begin{align}
		C_1 &\isdef
			\{\blank,
				\colrarrow, \collarrow, \coluarrow, \coldarrow,
				\colrtick, \colltick, \colutick, \coldtick
			\} \cup\{(c,\bullet): c\in C_0\}
	\end{align}
	is disjoint from $C_0$.
	Let $\Theta=\Theta_0\cup\Theta_1$.
	
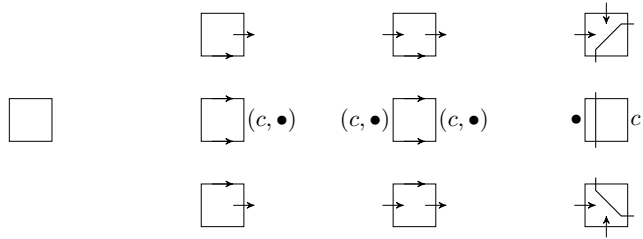
\begin{figure}[!h]
	\begin{center}
		{%
		\begin{tikzpicture}[x=90pt,y=40pt,scale=0.8,every node/.style={scale=0.8}]
								\wpatchLlu[0]{(-1,1)}\wpatchLcu[0]{(0,1)}\wpatchLru[0]{(1,1)}
			\wpatchblank{(-2,0)}\wpatchLlc[0]{(-1,0)}{c}\wpatchLcc[0]{(0,0)}{c}\wpatchLrc[0]{(1,0)}{c}
								\wpatchLld[0]{(-1,-1)}\wpatchLcd[0]{(0,-1)}\wpatchLrd[0]{(1,-1)}
		\end{tikzpicture}
		}
	\end{center}\vspace*{-5mm}
\caption{The extra tiles used in the proof of Proposition~\ref{prop:global-tiling-patching:NP-hard}.
	The tile set $\Theta_1$ consists of all these tiles (for each $c\in C$) and their rotations.
	The unlabeled edges are coloured with the blank symbol~$\blank$.
}\label{fig:global-tiling-patching:NP-hard:extra-tiles}
\end{figure}

\medskip
	Let $\alpha,\beta\colon\NN\to\NN$ be any two functions with polynomial growth.
	We show that the square tiling problem for $\Theta_0$ can be reduced in polynomial time to the global tiling patching problem for $\Theta$, $\alpha$ and~$\beta$.  It will then follow that the latter problem is $\classNP$-hard.
	
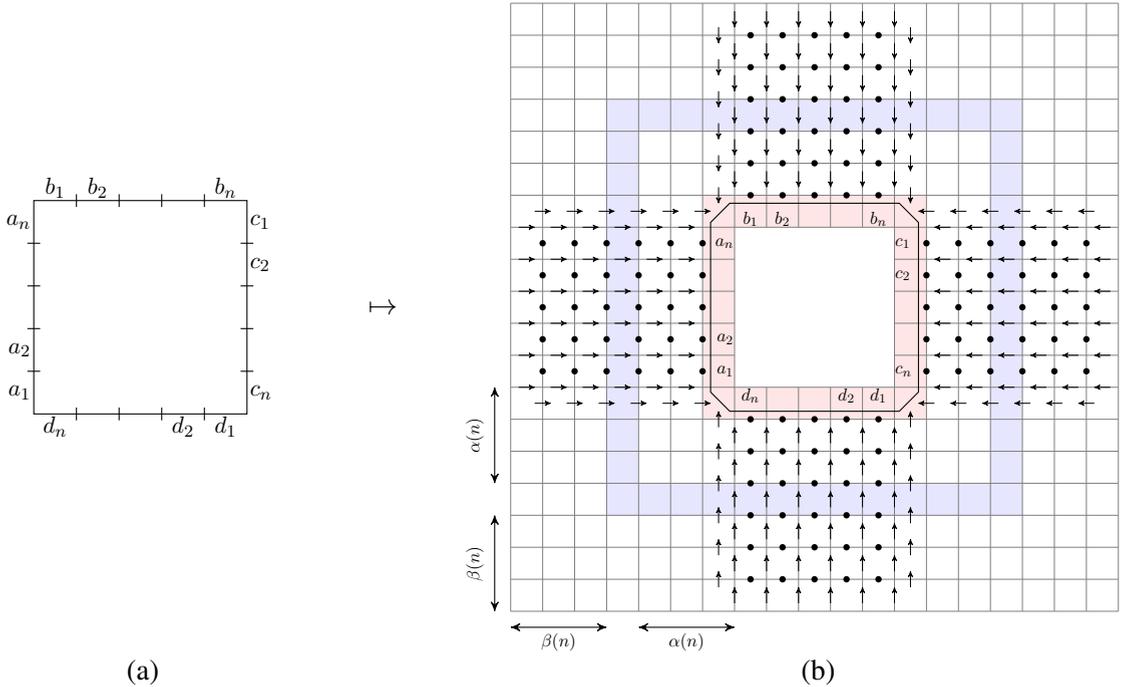
\begin{figure}[!h]
\vspace*{2mm}
	\begin{center}
		\begin{tabular}{ccc}
			\begin{minipage}{0.3\textwidth}
				\centering
				{%
				\begin{tikzpicture}[x=20pt,y=20pt,scale=0.8,every node/.style={scale=0.8}]
					\squaretilinginstance
				\end{tikzpicture}
				}
			\end{minipage}
			& $\mapsto$ &
			\begin{minipage}{0.63\textwidth}
				\centering
				{%
				\begin{tikzpicture}[x=20pt,y=20pt,scale=0.6,every node/.style={scale=0.6}]
					\squaretilinginstancetransformed
				\end{tikzpicture}
				}
			\end{minipage}
			\\ [11em]
			(a) & & (b)
		\end{tabular}
	\end{center}\vspace*{-5mm}
\caption{The reduction in the proof of Proposition~\ref{prop:global-tiling-patching:NP-hard}.
	(a) An instance of the square tiling problem, where $a_i,b_i,c_i,d_i$ are colours from $C$.
	(b) An instance of the global tiling patching problem.  The colour of the edges marked with $\bullet$ (not depicted) match accordingly.
}\label{fig:global-tiling-patching:NP-hard:reduction}\vspace*{-2mm}
\end{figure}

	The reduction is illustrated in Figure~\ref{fig:global-tiling-patching:NP-hard:reduction}.
	An instance of the square tiling problem for $\Theta_0$ is given by prescribing colours from $C$ for the boundary of the square $S_n$ for some $n$ as in Figure~\ref{fig:global-tiling-patching:NP-hard:reduction}a.
	This is transformed into a locally admissible tiling of
	the region $\Moore^{\alpha(n)+\beta(n)+1}(S_n)\setminus S_n$ with tiles from~$\Theta$ as depicted in Figure~\ref{fig:global-tiling-patching:NP-hard:reduction}b.  Let us call the latter tiling $q$.  Observe that the restriction of $q$ to the band
	$\partial\Moore(\Moore^{\alpha(n)}(S_n))$ (the region shaded in blue) enforces the tiling of the band
	$\partial\Moore(S_n)$ (the region shaded in red),
	which in turn fixes the colouring of the boundary of the square $S_n$.  Also observe that if $\tilde{q}$ is a locally admissible extension of $q$ to
	$\Moore^{\alpha(n)+\beta(n)+1}(S_n)$,
	then $\tilde{q}$ is in fact globally admissible (can be extended to the entire plane).
	
	Let $G(q)$ be an answer to the global tiling patching problem for $\Theta$ with the tiling of Figure~\ref{fig:global-tiling-patching:NP-hard:reduction}b as an input.
	
	First, suppose that the there is a locally admissible tiling of $S_n$ using the tiles from~$\Theta_0$ that achieves the prescribed colouring in Figure~\ref{fig:global-tiling-patching:NP-hard:reduction}a.
	It is easy to see that in this case, the partial tiling given in Figure~\ref{fig:global-tiling-patching:NP-hard:reduction}b is globally admissible.  Thus, $G(q)$ will be globally admissible and will agree with $q$ on $\partial\Moore(\Moore^{\alpha(n)}(S_n))$.  In particular, by the above remark, $G(q)_{S_n}$ will be a locally admissible tiling of $S_n$ with $\Theta_0$ compatible with the prescribed colouring of Figure~\ref{fig:global-tiling-patching:NP-hard:reduction}a.
	Conversely, assume that the prescribed colouring of Figure~\ref{fig:global-tiling-patching:NP-hard:reduction}a cannot be achieved by a locally admissible tiling of $S_n$ using~$\Theta_0$.  Then, either $G(q)$ is not locally admissible, or $G(q)$ and $q$ disagree on the band $\partial\Moore(S_n)$.
	
\medskip
	We see that, in either case, the answer to the square tiling problem for $\Theta_0$ on Figure~\ref{fig:global-tiling-patching:NP-hard:reduction}a can easily (in polynomial time) be extracted from $G(q)$.
\end{proof}

The proof of Theorem~\ref{thm:hard-stabilisation} is obtained by putting together Propositions~\ref{prop:serial-simulation} and~\ref{prop:global-tiling-patching:NP-hard}.

\subsubsection*{Proof of Theorem~\ref{thm:hard-stabilisation}.}\vspace*{-1mm}
	Let $X\neq\varnothing$ be the SFT of the valid tilings of the tile set $\Theta$ in Proposition~\ref{prop:global-tiling-patching:NP-hard}.  Suppose there exists a CA $F$ that stabilises $X$ in time $\tau(n)=\bigo(n^k)$ for some $k\in\NN$.  By Proposition~\ref{prop:serial-simulation}, this implies that there exists a Turing machine solving the global tiling patching problem associated to $\Theta$ and $\alpha(n)\isdef\beta(n)\isdef r\tau(n)$ (where $r$ is the neighbourhood radius of $F$) in polynomial time. Unless $\classP=\classNP$, this is a contradiction.\QED
\eject

\begin{remark}
	In principle, Proposition~\ref{prop:serial-simulation} may also lead to other hardness results.  For instance, if we have a tile set for which solving the global tiling patching problem with a Turing machine (with a two-dimensional tape) requires more than $\bigo(n^3)$ time steps, then it follows from Proposition~\ref{prop:serial-simulation} that there is no CA that stabilises the valid tilings of that tile set in linear time.
	\hfill\remarkqed
\end{remark}

\section{Stability against random noise}
\label{sec:random-noise}

A true fault-tolerant system should involve efficient error-correction mechanisms so as to maintain its structure and functionality in presence of random noise.  In this scenario, noise is present everywhere and throughout the evolution of the system.  Such level of stability is achieved by Toom's {$\toom$} CA and its variants~\cite{Too80} (Examples~\ref{exp:toom} and~\ref{exp:toom:noise}) and G\'acs's (very sophisticated) reliable CA~\cite{Gac86,Gac01}, but appears very challenging in general.

As a modest intermediate step, in this section we consider the problem of self-stabilisation starting from tilings that are perturbed by random noise. (Hence, noise is present only in the initial configuration but not at every time step.)
We show that if a CA stabilises from finite perturbations in linear time, then it also stabilises from Bernoulli random perturbations with a sufficiently low density of errors.
The argument is based on the idea of sparseness due to G\'acs~\cite{GaRe88,Gac86,Gac01} and Durand, Romashchenko and Shen~\cite{DuRoSh12}.
In this section, we restrict ourselves to the case of stabilisation by deterministic CA and leave the corresponding problem for probabilistic CA open.

\subsection{Formulation and examples}

Let us start by extending the notion of self-stabilisation to the case where the tiling is perturbed with random noise.  This is not as straightforward as one might hope for, and there are indeed several variants depending on the requirements and the type of noise.  Here we use one such possible notion.

\medskip
Consider a configuration $x\colon\ZZ^d\to\Sigma$.
Let $\varepsilon\geq 0$.
An \emph{$\varepsilon$-perturbation} of $x$ in $\Sigma^{\ZZ^d}$ is a random configuration $\tilde{\RV{x}}$ in $\Sigma^{\ZZ^d}$ with the property that for each finite set $I\subseteq\ZZ^d$, we have
\begin{align}
	\label{eq:eps-perturb}
	\xPr\big(\Delta(x,\tilde{\RV{x}})\supseteq I\big) =
	\xPr(\text{$\tilde{\RV{x}}_i\neq x_i$ for each $i\in I$}) &\leq \varepsilon^{\abs{I}} \;.
\end{align}
We think of $\tilde{\RV{x}}$ as a ``noisy version'' of $x$ where random errors have occurred leading to a change in the state of some cells.
A special type of an $\varepsilon$-perturbation is a \emph{Bernoulli $\varepsilon$-perturbation} for which the set $\Delta(x,\tilde{\RV{x}})$ is a Bernoulli random set with parameter~$\varepsilon$, that is, a random subset of $\ZZ^d$ in which each cell $k\in\ZZ^d$ is included with probability~$\varepsilon$ independently of the other cells.
The notion of $\varepsilon$-perturbation is however much more general.
It turns out that the usual arguments regarding noise-resilience in computational models often work equally well with this more general notion of noise.
This was first observed by Toom~\cite{Too80}.

\medskip
The \emph{(Besicovitch) density} of a set $A\subseteq\ZZ^d$ is defined as
\begin{align}
	\updensity(A) &\isdef \limsup_{n\to\infty} \frac{A\cap\{-n,-n+1,\ldots,n\}^d}{(2n+1)^d} \;.
\end{align}

\paragraph{Self-stabilisation from random perturbations.}
Let $\varepsilon,\delta\geq 0$.
We say that a CA $F\colon\Sigma^{\ZZ^d}\to\Sigma^{\ZZ^d}$ \emph{$\delta$-stabilises} an SFT $X\subseteq\Sigma^{\ZZ^d}$ \emph{from $\varepsilon$-perturbations}~if
\begin{enumerate}[label={\textrm{(\roman*)}},ref=\roman*]
	\item (\emph{consistency}) the configurations of $X$ are fixed points, that is, $F(x)=x$ for every $x\in X$,
	\item for every $x\in X$ and every $\varepsilon$-perturbation $\tilde{\RV{x}}$ of $x$,
		\begin{enumerate}[label={(\theenumi.\alph*)},leftmargin=3.5em]
			\item (\emph{attraction})
				$F^t(\tilde{\RV{x}})$ converges almost surely in the product topology to a (random) configuration $\RV{y}$ from $X$,
			\item (\emph{stability})
				The density of the disagreements between $\RV{y}$ and $x$
				is almost surely less than~$\delta$, that is,
				$\updensity\big(\Delta(x,\RV{y})\big)<\delta$.
		\end{enumerate}
\end{enumerate}
We say that $F$ stabilises $X$ \emph{from random perturbations} if for every $\delta>0$, there exists an $\varepsilon>0$ such that $F$ $\delta$-stabilises $X$ from $\varepsilon$-perturbations.
The stability condition may sound unnecessary at first sight.  The following example illustrates why it is a desired property.

\begin{example}[Unstable attraction]
	Let $F\colon\{\symb{0},\symb{1}\}^\ZZ\to\{\symb{0},\symb{1}\}^\ZZ$ be the one-dimensional CA with neighbourhood $\Neighb=\{-1,0,1\}$ defined by
	\begin{align}
		F(x)_k &\isdef
			\begin{cases}
				\symb{0}	& \text{if $x_{k-1}=x_k=x_{k+1}=\symb{0}$,} \\
				\symb{1}	& \text{otherwise.}
			\end{cases}
	\end{align}
	As usual, let $\Hom_2\isdef\{\unifO,\unifI\}$ be the finite tiling space consisting only of the two homogeneous configurations in~$\{\symb{0},\symb{1}\}^\ZZ$.
	Then, $F$ satisfies the conditions of consistency and attraction in the definition of stabilisation from random perturbations.
	However, note that if $\RV{x}$ is a non-trivial Bernoulli perturbation of~$\unifO$, then $F^t(\RV{x})$ converges almost surely to~$\unifI$ rather than to~$\unifO$.
	In other words, the system recovers from the initial noise but the distinction between the elements of $X$ is completely lost.
	\hfill\exampleqed
\end{example}

Before stating our result, let us mention some examples.

\begin{example}[Toom's CA; Continuation of Example~\ref{exp:toom}]
\label{exp:toom:noise}
	Toom's CA has strong forms of stability against noise.
	
\medskip
	Bu\v{s}i\'{c}, Fat\`es, Mairesse and Marcovici have shown that Toom's CA \emph{classifies} the Bernoulli random configurations according to their \emph{density}~\cite{BuFaMaMa13}.  Namely, if we start from a Bernoulli random configuration $\RV{x}$ with parameter~$p$ (i.e., the state of different cells are chosen independently at random with probability~$p$ of choosing~$\symb{1}$ and probability~$1-p$ of choosing~$\symb{0}$), then
	\begin{align}
		\toom^t(\RV{x}) &\to \unifO		\qquad\text{if $p<\nicefrac{1}{2}$,} \\
		\toom^t(\RV{x}) &\to \unifI		\qquad\text{if $p>\nicefrac{1}{2}$.}
	\end{align}
	In particular, for $\delta=0$ and every $\varepsilon<\nicefrac{1}{2}$, the CA $\delta$-stabilises the homogeneous tiling space $\Hom_2\isdef\{\unifO,\unifI\}$ from Bernoulli $\varepsilon$-perturbations.
	
\medskip
	Moreover, the two homogeneous configurations $\unifO$ and $\unifI$ remain stable under $\toom$ even if there is a small independent noise at every time step.
	In his original paper~\cite{Too80}, Toom showed that for every $\delta>0$, there exists an~$\varepsilon>0$ such that if $\big(\RV{x}^{(t)}_k\big)_{k\in\ZZ^2,t\in\NN}$ is an $\varepsilon$-perturbed trajectory of $\toom$ starting from $\unifO$ in the sense that for every finite set $I\subseteq\ZZ^2\times\NN$ of space-time positions, we have
	\begin{align}
		\xPr\big\{\text{$\RV{x}^{(s)}_i$ deviates from the local rule of $\toom$ at every $(i,s)\in I$}\big\}
			&\leq \varepsilon^{\abs{I}} \;,
	\end{align}
	then this trajectory remains $\delta$-close to its initial configuration~$\unifO$, that is,
	\begin{align}
		\xPr\big(\RV{x}^{(t)}_k\neq \symb{0}\big) &\leq\delta \qquad\text{for each $(k,t)\in\ZZ^2\times\NN$.}
	\end{align}
	In fact, Toom proved the same kind of stability for random perturbations of all monotonic eroders.
	Note that for $\toom$, a similar stability property holds by symmetry for the trajectories starting from~$\unifI$.~\hfill\exampleqed
\end{example}

\begin{example}[GKL and modified~traffic; Continuation of Example~\ref{exp:gkl-mtraffic}]
	Using the sparseness result of Durand, Romashchenko and Shen~\cite{DuRoSh12}, Taati showed that $\gkl$ and modified traffic stabilise the homogeneous tiling space $\Hom_2\isdef\{\unifO,\unifI\}$ from random perturbations~\cite{Taa15}.  Below, we use the same sparseness result to obtain a more general stabilisation result.
	\hfill\exampleqed
\end{example}

\begin{example}[Continuous-time~$\MRIE$; Continuation of Example~\ref{exp:MRIE:continuous-time:finite}]
\label{exp:MRIE:continuous-times}
	The notion of stabilisation from random perturbations extends naturally to probabilistic CA and to continuous-time models.
	Fontes, Schonmann and Sidoravicius~\cite{FoScSi02} proved that, in dimension two and higher, the continuous-time~$\MRIE$ of Example~\ref{exp:MRIE:continuous-time:finite} stabilises the homogeneous tiling space $\Hom_2\isdef\{\unifO,\unifI\}$ from random perturbations.
	We conjecture that a similar result holds for the discrete-time version discussed in Section~\ref{sec:pca:finite}.
	\hfill\exampleqed
\end{example}

\subsection{Sparseness and the general result}

Our objective is to prove the following result.

\begin{theorem}[Self-stabilisation from random perturbations]
\label{thm:stabilisation:random-perturbation:linear}
	Let $F\colon\Sigma^{\ZZ^d}\to\Sigma^{\ZZ^d}$ be a CA and $X\subseteq\Sigma^{\ZZ^d}$ an SFT.
	If $F$ stabilises $X$ from finite perturbations in linear time, then $F$ also stabilises~$X$ from random perturbations.
\end{theorem}

A proof of Theorem~\ref{thm:stabilisation:random-perturbation:linear} will be provided in the next subsection.  Here we discuss the idea of the proof and its main ingredient, namely the notion of sparseness.
Let us remark that a recent result of G\'acs~\cite{Gac20} would allow us to improve the above result by replacing the ``linear time'' condition with the ``sub-quadratic time'' condition.
See Remark~\ref{rem:sparseness:sub-quadratic} below for more details.

A basic idea in the proof of Theorem~\ref{thm:stabilisation:random-perturbation:linear} is that the speed of the propagation of information in a cellular automaton is bounded.

\begin{observation}[Speed of light]
\label{obs:speed-of-light}
	Let $F\colon\Sigma^{\ZZ^d}\to\Sigma^{\ZZ^d}$ be a CA with neighbourhood $\Neighb\isdef\{-r,\ldots,r\}^d$.
	Then, for every $A\subseteq\ZZ^d$ and $t\geq 0$, and every two configurations $x,y\in\Sigma^{\ZZ^d}$, we have
	\begin{enumerate}[label={\textup{(\alph*)}}]
		\item If $x_{\Neighb^t(A)}=y_{\Neighb^t(A)}$, then $F^t(x)_A=F^t(y)_A$.
		\item If $\Delta(x,y)\subseteq A$, then $\Delta\big(F^t(x),F^t(y)\big)\subseteq \Neighb^t(A)$.
	\end{enumerate}
\end{observation}
\eject

Suppose that a CA $F\colon\Sigma^{\ZZ^d}\to\Sigma^{\ZZ^d}$ stabilises an SFT $X\subseteq\Sigma^{\ZZ^d}$ from finite perturbations.
Since information propagates at bounded speed, this immediately implies that $F$ stabilises starting from any perturbation in which the errors are ``sufficiently sparse''.  Indeed, if the set of errors in a perturbation~$\tilde{x}$ can be decomposed into well-separated finite islands, then under the iterations of $F$, each island will disappear before sensing or affecting the other islands (see Figure~\ref{fig:sparse-error-islands}).

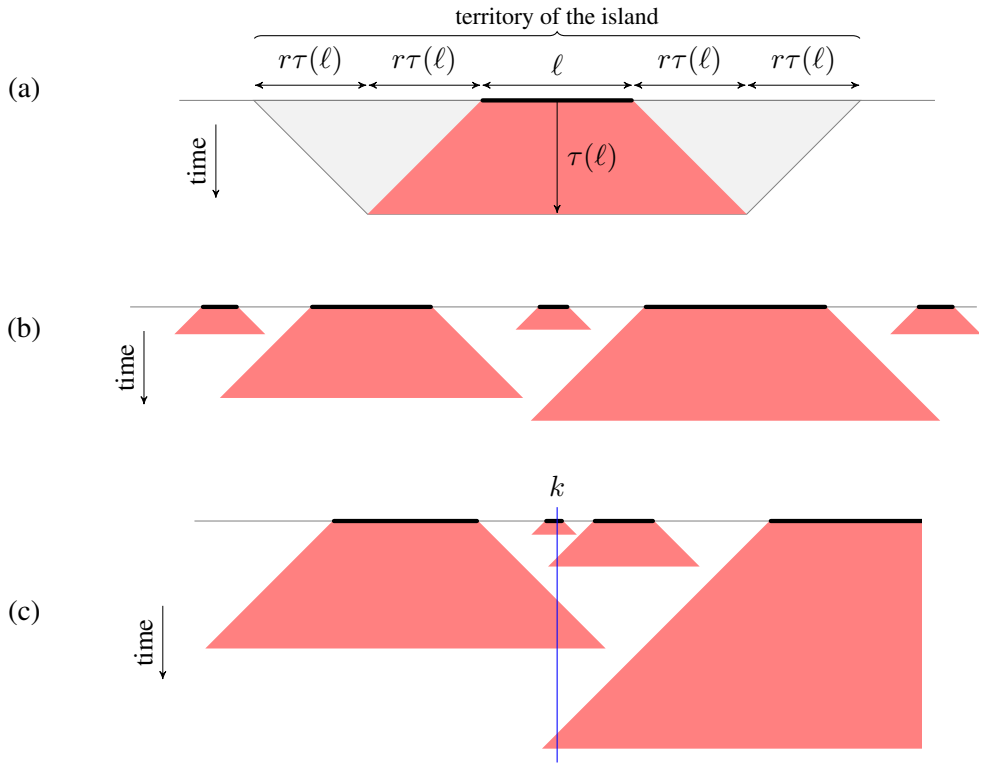
\begin{figure}[!h]
\vspace*{-2mm}
	\begin{center}
		\begin{tabular}{c@{\quad\qquad}c}
		(a) &
				{
		\begin{tikzpicture}[scale=1,>=stealth',shorten >=0.5pt,shorten <=0.5pt,baseline=(current bounding box.center)]
			\useasboundingbox (-5,-1.5) rectangle (5,1.6);
			\fill[fill=gray!10] (-4,0) -- (-2.5,-1.5) -- (2.5,-1.5) -- (4,0) -- cycle;
			\draw[help lines] (-4,0) -- (-2.5,-1.5) -- (2.5,-1.5) -- (4,0) -- cycle;
			\fill[fill=red!50] (-1,0) -- (1,0) -- (2.5,-1.5) -- (-2.5,-1.5) -- cycle;
			\draw[help lines] (-5,0) -- (5,0);
			\draw[ultra thick, line cap=round] (-1,0) -- (1,0);
			\draw[->,very thin,overlay] (0,0) -- (0,-1.5) node[midway,right] {$\tau(\ell)$};
			\draw[<->,very thin] (-1,0.2) -- (1,0.2) node[midway,above] {$\ell$};
			\draw[<->,very thin] (-4,0.2) -- (-2.5,0.2) node[midway,above] {$r\tau(\ell)$};
			\draw[<->,very thin] (-2.5,0.2) -- (-1,0.2) node[midway,above] {$r\tau(\ell)$};
			\draw[<->,very thin] (1,0.2) -- (2.5,0.2) node[midway,above] {$r\tau(\ell)$};
			\draw[<->,very thin] (2.5,0.2) -- (4,0.2) node[midway,above] {$r\tau(\ell)$};
			
			\draw[->,overlay,very thin] (-4.5,-0.3) -- node[above,rotate=90] {\small time} (-4.5,-1.3);
			
			\draw[decorate, decoration={brace,amplitude=3pt},thin]
				(-4,0.8) -- node[above] {\footnotesize territory of the island} (4,0.8);
		\end{tikzpicture}
		}
		\bigskip \\
		
		(b) &
		
		{%
		\begin{tikzpicture}[scale=0.8,>=stealth',shorten >=0.5pt,shorten <=0.5pt,baseline=(current bounding box.center)]
			\begin{scope}
			\clip (-7,-2) rectangle (7,1);
			\useasboundingbox (-7,-2) rectangle (7,0.5);
			
			\draw[help lines] (-7,0) -- (7,0);
			\fill[fill=red!50] (-4,0) -- (-5.5,-1.5) -- (-0.5,-1.5) -- (-2,0) -- cycle;
			\draw[ultra thick, line cap=round] (-4,0) -- (-2,0);

			\fill[fill=red!50] (1.5,0) -- (-0.375,-1.875) -- (6.375,-1.875) -- (4.5,0) -- cycle;
			\draw[ultra thick, line cap=round] (1.5,0) -- (4.5,0);
			
			\fill[fill=red!50] (-0.25,0) -- (-0.625,-0.375) -- (0.625,-0.375) -- (0.25,0) -- cycle;
			\draw[ultra thick, line cap=round] (-0.25,0) -- (0.25,0);
			
			\fill[fill=red!50] (-5.8,0) -- (-6.25,-0.45) -- (-4.75,-0.45) -- (-5.2,0) -- cycle;
			\draw[ultra thick, line cap=round] (-5.8,0) -- (-5.2,0);
			
			\fill[fill=red!50] (6,0) -- (5.55,-0.45) -- (7.05,-0.45) -- (6.6,0) -- cycle;
			\draw[ultra thick, line cap=round] (6,0) -- (6.6,0);
			\end{scope}
			
			\draw[->,overlay,very thin] (-6.75,-0.375) --
				node[above,rotate=90] {\small time} (-6.75,-1.625);
			
		\end{tikzpicture}
		}
		\bigskip \\
		
		(c) &
		
		{%
		\begin{tikzpicture}[scale=0.8,>=stealth',shorten >=0.5pt,shorten <=0.5pt,baseline=(current bounding box.center)]
			\begin{scope}
			\clip (-6,-4.25) rectangle (6,1);
			
			\draw[help lines] (-6,0) -- (6,0);
			
			\fill[fill=red!50] (-0.2,0) -- (-0.425,-0.225) -- (0.325,-0.225) -- (0.1,0) -- cycle;
			\draw[ultra thick, line cap=round] (-0.2,0) -- (0.1,0);
			
			\fill[fill=red!50] (0.6,0) -- (-0.15,-0.75) -- (2.35,-0.75) -- (1.6,0) -- cycle;
			\draw[ultra thick, line cap=round] (0.6,0) -- (1.6,0);
			
			\fill[fill=red!50] (-3.7,0) -- (-5.8,-2.1) -- (0.8,-2.1) -- (-1.3,0) -- cycle;
			\draw[ultra thick, line cap=round] (-3.7,0) -- (-1.3,0);
			
			\fill[fill=red!50] (3.5,0) -- (-0.25,-3.75) -- (12.25,-3.75) -- (8.5,0) -- cycle;
			\draw[ultra thick, line cap=round] (3.5,0) -- (8.5,0);
			
			\draw[thin,blue] (0,0.25) node[above,black] {$k$} -- (0,-4);
			\end{scope}
			
			\draw[->,overlay,very thin] (-6.5,-1.375) -- node[above,rotate=90] {\small time} (-6.5,-2.625);
		\end{tikzpicture}
		}
		\end{tabular}
	\end{center}\vspace*{-3mm}
\caption{Stabilisation from a sparse set of errors.
	(a)~An isolated island of errors and its territory.
	(b)~Well-separated islands of errors disappear before sensing or affecting one another.
	(c)~The disappearance of all the islands of error is not sufficient for stabilisation.
}\label{fig:sparse-error-islands}\vspace*{-1mm}
\end{figure}

More specifically, suppose that $F$ stabilises $X$ in time $\tau(n)$.
Suppose that $F$ has a neighbourhood radius~$r$ and $X$ has an interaction range~$m$.
Let $\tilde{x}\in\Sigma^{\ZZ^d}$ be a (possibly infinite) perturbation of a configuration $x\in X$, and let $S\isdef\Delta(x,\tilde{x})$ denote the set of positions at which an error has occurred.
Let us call a subset $A\subseteq S$ with diameter~$\ell$ an \emph{isolated island} if $A$ is at distance more than $2r\tau(\ell)+m$ from $S\setminus A$, where $r$ denotes the neighbourhood radius of~$F$ (Figure~\ref{fig:sparse-error-islands}a).  Note that under the iterations of~$F$ on~$\tilde{x}$, every isolated island of $S$ is corrected before interacting with the rest of~$S$.  Let $S'\subseteq S$ be the subset of $S$ obtained by removing all its isolated islands.  Then, again under the iterations of~$F$ on $\tilde{x}$, every isolated island of $S'$ is corrected before interacting with the rest of~$S'$.  This reasoning can be repeated to identify an infinite hierarchy of islands in~$S$, each of which is corrected without interacting with the rest of~$S$.  If every element of $S$ is included in one of these islands, then this means that every error on~$\tilde{x}$ is eventually patched (Figure~\ref{fig:sparse-error-islands}b).  However, note that this alone does not guarantee attraction towards an element of $X$, for it is possible that a cell~$k\in\ZZ^d$ is within the interaction range of infinitely many islands of~$S$ (Figure~\ref{fig:sparse-error-islands}c) and hence never stabilises.

To make this idea precise, we need to introduce the notion of sparseness.

\medskip
Let $\rho\colon\NN\to\NN$ be an arbitrary function, and let $\Moore\isdef\{-1,0,1\}^d$ be the Moore neighbourhood.  We define the \emph{$\rho$-territory} of a finite set $A\subseteq\ZZ^d$ as the set $\Moore_\rho(A)\isdef\Moore^{\rho(\diam(A))}$ of all cells $k$ that are within $\rho(\diam(A))$ from $A$.
(We use the Moore neighbourhood and the $\ell^\infty$ distance on~$\ZZ^d$, but this is an arbitrary choice.)
We say that a set $S\subseteq\ZZ^d$ is \emph{$\rho$-sparse} if there is a partitioning $\family{C}(S)$ of $S$ into finite sets, called the \emph{$\rho$-islands} of $S$, such that
\begin{enumerate}[label={(\roman*)}]
	\item (separation) every two distinct islands $A,B\in\family{C}(S)$ are \emph{well separated} from each other, that is, either $A\cap \Moore_\rho(B)=\varnothing$ or $\Moore_\rho(A)\cap B=\varnothing$.
	\item (thinness) every cell $a\in\ZZ^d$ is in the territory of no more than finitely many islands, that is, $\{C\in\family{C}(S): \Moore_\rho(C)\ni a\}$ is finite.
\end{enumerate}

The significance of the concept of sparseness in the context of fault-tolerant computation was noticed by G\'acs~\cite{GaRe88,Gac86,Gac01}, who used more sophisticated variants of it
in the scenario in which errors due to noise occur at every time step.
The above definition of sparseness is equivalent to the one introduced by Durand, Romashchenko and Shen~\cite{DuRoSh12}, who used it in a context very similar to (but different from) ours.

\medskip
Let $\varepsilon>0$.  By an \emph{$\varepsilon$-random} subset of $\ZZ^d$, we shall mean a random set $\RV{S}\subseteq\ZZ^d$ with the property that for each finite $I\subseteq\ZZ^d$,
\begin{align}
	\xPr(\RV{S}\supseteq I) &\leq \varepsilon^{\abs{I}} \;.
\end{align}
Durand, Romashchenko and Shen~\cite[Section~9.2]{DuRoSh12} proved the following.

\begin{theorem}[Linear sparseness of $\varepsilon$-random sets]
\label{thm:sparseness:linear}
	Let $\rho\colon\NN\to\NN$ be such that $\rho(\ell)=\bigo(\ell)$ as $\ell\to\infty$.
	For every $\delta>0$, there exists an $\varepsilon>0$ such that
	\begin{enumerate}[label={\textup{(\roman*)}}]
		\item \label{item:thm:sparseness:linear:sparseness} every $\varepsilon$-random set $\RV{S}\subseteq\ZZ^d$ is almost surely $\rho$-sparse,
		\item \label{item:thm:sparseness:linear:rareness} for every $\varepsilon$-random set $\RV{S}\subseteq\ZZ^d$,
			the density of the union of the $\rho$-territories of all the $\rho$-islands in~$\RV{S}$
			is almost surely less than~$\delta$.
	\end{enumerate}
\end{theorem}

Let us clarify that Durand, Romashchenko and Shen proved their result for Bernoulli $\varepsilon$-random sets (for all sufficiently small~$\varepsilon$).  However, their proof goes through without modification for arbitrary $\varepsilon$-random sets.

\begin{remark}
\label{rem:sparseness:sub-quadratic}
	The proof of Theorem~\ref{thm:stabilisation:random-perturbation:linear} relies on the linear sparseness of $\varepsilon$-random sets (Theorem~\ref{thm:sparseness:linear}).
	Recently, G\'acs has strengthened the latter sparseness result by proving the sub-quadratic sparseness of $\varepsilon$-random sets~\cite{Gac20}.  More precisely, he has shown that if $\rho\colon\NN\to\NN$ is any sub-quadratic function (i.e., $\rho(\ell)=\bigo(\ell^\beta)$ for some $\beta<2$), then for $\varepsilon>0$ sufficiently small, every $\varepsilon$-random set is almost surely $\rho$-sparse.
	Applying the result of G\'acs in place of Theorem~\ref{thm:sparseness:linear}, we immediately obtain that $F$~stabilises~$X$ from random perturbations as long as it stabilises $X$ from finite perturbations in sub-quadratic time.~%
	\hfill\remarkqed
\end{remark}

\subsection{Proof of the result}
In this section, we prove Theorem~\ref{thm:stabilisation:random-perturbation:linear}

\medskip
Suppose that $F$ stabilises $X$ from finite perturbations in time~$\tau(\ell)$.  There is no loss of generality to assume that $\tau(\ell)$ is an increasing function with $\tau(\ell)\to\infty$ as $\ell\to\infty$.  If not, we can replace $\tau(\ell)$ with $\tau'(\ell)\isdef\max\{\ell,\tau(k): k\leq\ell\}$ in the argument below.

Suppose that $F$ has neighbourhood radius~$r$ so that $\Moore^r=\{-r,\ldots,r\}^d$ is a neighbourhood for the local rule of $F$.  Suppose that $X$ has interaction range~$m$ so that it can be identified by a set $\collection{F}$ of forbidden patterns with shape $S_m\isdef\{0,1,\ldots,m-1\}^d$.
Define $\rho(\ell)\isdef 3 r\tau(\ell)+m$.  (We use this choice rather than $\rho(\ell)\isdef 2 r\tau(\ell)+m$ used in the intuitive explanation above to simplify the proof.)
By assumption, $\tau(\ell)=\bigo(\ell)$, hence $\rho(\ell)=\bigo(\ell)$.
Let $\delta>0$, and choose $\varepsilon>0$ as in Theorem~\ref{thm:sparseness:linear}.
We claim that $F$ $\delta$-stabilises $X$ from $\varepsilon$-perturbations.

\medskip
The consistency property of $F$ is satisfied by assumption.
Let $x\in X$, and let $\tilde{\RV{x}}$ be an $\varepsilon$-perturbation of $x$.
Let $\RV{Q}\isdef\Delta(x,\tilde{\RV{x}})$ denote the set of positions where an error has occurred.
Note that $\RV{Q}$ is an $\varepsilon$-random set.  Thus, by the choice of $\varepsilon$, we know that $\RV{Q}$ is almost surely $\rho$-sparse.  Let $\family{C}(\RV{Q})$ be a partition of $\RV{Q}$ into $\rho$-islands, witnessing the $\rho$-sparseness of~$\RV{Q}$.
For each $\ell\geq 0$, let us define
\begin{align}
	\family{C}_{\ell}(\RV{Q}) &\isdef \big\{C\in\family{C}(\RV{Q}): \diam(C)=\ell\big\} \;, \\
	\family{C}_{>\ell}(\RV{Q}) &\isdef \big\{C\in\family{C}(\RV{Q}): \diam(C)>\ell\big\}
		= \bigcup_{k>\ell}\family{C}_{k}(\RV{Q}) \;, \\
	\family{C}_{\geq\ell}(\RV{Q}) &\isdef \big\{C\in\family{C}(\RV{Q}): \diam(C)\geq\ell\big\}
		= \family{C}_{\ell}(\RV{Q})\cup\family{C}_{>\ell}(\RV{Q}) \;.
\end{align}
Note that by the separation property of~$\RV{Q}$,
each $C\in\family{C}_\ell(S)$ is at distance more than $\rho(\ell)=3r\tau(\ell)+m$
from every $C'\in\family{C}_{\geq\ell}(\RV{Q})$ distinct from~$C$.

\begin{lemma}[Recursive correction]
\label{lem:stabilisation:random-perturbation:main}
	For all $\ell\geq 0$, we can construct configurations $\tilde{\RV{x}}^{(\ell)}\in\Sigma^{\ZZ^d}$ and $\RV{y}^{(\ell)}\in X$ such that the following conditions are satisfied:
	\begin{enumerate}[label={\textup{(\alph*)}}]
		\item \label{item:lem:stabilisation:random-perturbation:main:forgetting} \textup{(forgetting)}
			$F^t(\tilde{\RV{x}})=F^t\big(\tilde{\RV{x}}^{(\ell)}\big)$ for all $t\geq\tau(\ell)$, that is, the distinction between $\tilde{\RV{x}}$ and $\tilde{\RV{x}}^{(\ell)}$ is forgotten in $\tau(\ell)$ time steps.
		\item \label{item:lem:stabilisation:random-perturbation:main:patching} \textup{(patching)}
			$\RV{y}^{(0)}=x$ and $\Delta\big(\RV{y}^{(\ell-1)},\RV{y}^{(\ell)}\big)\subseteq\bigcup_{C\in\family{C}_{\ell}(\RV{Q})} \Moore^{r\tau(\ell)}(C)\subseteq\bigcup_{C\in\family{C}_{\ell}(\RV{Q})} \Moore_\rho(C)$ for $\ell\geq 1$, that is, $\RV{y}^{(\ell-1)}$ and $\RV{y}^{(\ell)}$ differ only in the territory of the $\rho$-islands with diameter $\ell$ in~$\family{C}(\RV{Q})$.
		\item \label{item:lem:stabilisation:random-perturbation:main:correcting} \textup{(correction)}
			$\Delta\big(\RV{y}^{(\ell)},\tilde{\RV{x}}^{(\ell)}\big)\subseteq\bigcup_{C\in\family{C}_{>\ell}(\RV{Q})} C$, that is, the $\rho$-islands of errors on $\tilde{\RV{x}}^{(\ell)}$ are simply the $\rho$-islands of diameter larger than $\ell$ in~$\family{C}(\RV{Q})$.
	\end{enumerate}
\end{lemma}
\begin{proof}
	We start with $\tilde{\RV{x}}^{(0)}\isdef\tilde{\RV{x}}$ and $\RV{y}^{(0)}\isdef x$.
	
	\medskip
	In step $\ell\geq 1$, assuming we have constructed $\tilde{\RV{x}}^{(\ell-1)}$ and $\RV{y}^{(\ell-1)}$,
	we construct $\tilde{\RV{x}}^{(\ell)}$ as follows.
	Let $C_1,C_2,\ldots$ be an enumeration of the elements of~$\family{C}_\ell(\RV{Q})$.  Let $\RV{z}^{(n)}$ be the configuration that agrees with $\tilde{\RV{x}}^{(\ell-1)}$ on $C_n$ and with $\RV{y}^{(\ell-1)}$ outside $C_n$.  Note that $\RV{z}^{(n)}$ is a finite perturbation of $\RV{y}^{(\ell-1)}$ with $\Delta\big(\RV{y}^{(\ell-1)},\RV{z}^{(n)}\big)\subseteq C_n$.  Since $F$ stabilises $X$ from finite perturbations in time~$\tau(\cdot)$, we know that $F^{\tau(\ell)}(\RV{z}^{(n)})\in X$. By Observation~\ref{obs:speed-of-light}, the two configurations $\RV{y}^{(\ell-1)}$ and $F^{\tau(\ell)}(\RV{z}^{(n)})$ agree everywhere except possibly in $\Moore^{r\tau(\ell)}(C_n)$, that is, the set of positions within distance $r\tau(\ell)$ from~$C_n$.  Define
	\begin{align}
		\tilde{\RV{x}}^{(\ell)}_k &\isdef
			\begin{cases}
				F^{\tau(\ell)}(\RV{z}^{(n)})_k
					& \text{if $k\in\Moore^{r\tau(\ell)}(C_n)$ for some $n\in\NN$,} \\
				\tilde{\RV{x}}^{(\ell-1)}_k		& \text{otherwise.}
			\end{cases}
		\\
		\RV{y}^{(\ell)}_k &\isdef
			\begin{cases}
				F^{\tau(\ell)}(\RV{z}^{(n)})_k
					& \text{if $k\in\Moore^{r\tau(\ell)}(C_n)$ for some $n\in\NN$,} \\
				\RV{y}^{(\ell-1)}_k		& \text{otherwise.}
			\end{cases}
	\end{align}
	
	We use induction to show that $\RV{y}^{(\ell)}\in X$ and conditions~\ref{item:lem:stabilisation:random-perturbation:main:forgetting}--\ref{item:lem:stabilisation:random-perturbation:main:correcting} are satisfied.
	The case $\ell=0$ is trivial.  Assuming that these conditions are satisfied in the first $\ell-1$ steps of the construction, we show that they remain satisfied in step~$\ell$.
	\begin{enumerate}[label={\textup{(\alph*)}}]
		\item[(--)] Let us first verify that $\RV{y}^{(\ell)}\in X$.
			To see this, note that for each $n\in\NN$, the configuration $\RV{y}^{(\ell)}$ agrees with $F^{\tau(\ell)}(\RV{z}^{(n)})$ not only on $\Moore^{r\tau(\ell)}(C_n)$ but on $\Moore^{r\tau(\ell)+m}(C_n)$.
			Namely, by Observation~\ref{obs:speed-of-light}, $F^{\tau(\ell)}(\RV{z}^{(n)})$ agrees with $\RV{y}^{(\ell-1)}$ on $\Moore^{r\tau(\ell)+m}(C_n)\setminus\Moore^{r\tau(\ell)}(C_n)$, and on the same set, $\RV{y}^{(\ell-1)}$ agrees with $\RV{y}^{(\ell)}$ by construction.
			Now, for every $i\in\ZZ^d$, the set $i+S_m$ is either entirely in $\Moore^{r\tau(\ell)+m}(C_n)$ for some $n$ or entirely in $\ZZ^d\setminus\bigcup_{n=1}^\infty\Moore^{r\tau(\ell)}(C_n)$.  In the first case, $\RV{y}^{(\ell)}_{i+S_m} = F^{\tau(\ell)}(\RV{z}^{(n)})_{i+S_m}\notin\collection{F}$ and in the second case $\RV{y}^{(\ell)}_{i+S_m} = \RV{y}^{(\ell-1)}_{i+S_m}\notin\collection{F}$.  Thus, $\RV{y}^{(\ell)}\in X$ as claimed.
		
		\item By the induction hypothesis, $F^t(\tilde{\RV{x}})=F^t\big(\tilde{\RV{x}}^{(\ell-1)}\big)$ for $t\geq\tau(\ell-1)$.  Since $\tau(\ell-1)<\tau(\ell)$, it is enough to show that $F^t\big(\tilde{\RV{x}}^{(\ell)}\big)=F^t\big(\tilde{\RV{x}}^{(\ell-1)}\big)$ for $t\geq\tau(\ell)$.
			
			Note that for each $n\in\NN$, the configuration $\RV{x}^{(\ell)}$ agrees with $F^{\tau(\ell)}(\RV{z}^{(n)})$ not only on $\Moore^{r\tau(\ell)}(C_n)$ but on $\Moore^{3r\tau(\ell)}(C_n)$.
			Namely, by construction, Observation~\ref{obs:speed-of-light}, and the fact that the elements of~$X$ are fixed points, $F^{\tau(\ell)}(\RV{z}^{(n)})$ agrees with $\RV{y}^{(\ell-1)}$ outside $\Moore^{r\tau(\ell)}(C_n)$.  Furthermore, by the induction hypothesis, $\Delta\big(\RV{y}^{(\ell-1)},\tilde{\RV{x}}^{(\ell-1)}\big)\subseteq\bigcup_{C\in\family{C}_{\geq\ell}(\RV{Q})} C$.  Since $C_n$ has distance more than $\rho(\ell)=3r\tau(\ell)+m$ from the rest of $\family{C}_{\geq\ell}(\RV{Q})$, we find that $\RV{y}^{(\ell-1)}$ agrees with $\tilde{\RV{x}}^{(\ell-1)}$ on $\Moore^{3r\tau(\ell)}(C_n)\setminus \Moore^{r\tau(\ell)}(C_n)$.  Lastly, again by construction and the fact that $C_n$ has distance more than $\rho(\ell)=3r\tau(\ell)+m$ from the rest of $\family{C}_{\ell}(\RV{Q})$, the two configurations $\tilde{\RV{x}}^{(\ell-1)}$ and $\tilde{\RV{x}}^{(\ell)}$ also agree on $\Moore^{3r\tau(\ell)}(C_n)\setminus \Moore^{r\tau(\ell)}(C_n)$.
			
			Now, Observation~\ref{obs:speed-of-light} and the fact that $F^{\tau(\ell)}(\RV{z}^{(n)})\in X$ is a fixed point imply that for each $n\in\NN$, the configurations $F^{\tau(\ell)}(\tilde{\RV{x}}^{(\ell-1)})$, $F^{\tau(\ell)}(\RV{z}^{(n)})$ and $F^{\tau(\ell)}(\tilde{\RV{x}}^{(\ell)})$ agree on $\Moore^{2r\tau(\ell)}(C_n)$.  On the other hand, by construction, $\Delta\big(\tilde{\RV{x}}^{(\ell-1)},\tilde{\RV{x}}^{(\ell)}\big)\subseteq\bigcup_{n=1}^\infty\Moore^{r\tau(\ell)}(C_n)$.  Therefore, according to Observation~\ref{obs:speed-of-light}, the two configurations $F^{\tau(\ell)}(\tilde{\RV{x}}^{(\ell-1)})$ and $F^{\tau(\ell)}(\tilde{\RV{x}}^{(\ell)})$ agree outside $\bigcup_{n=1}^\infty\Moore^{2r\tau(\ell)}(C_n)$.  We conclude that $F^{\tau(\ell)}(\tilde{\RV{x}}^{(\ell-1)})$ and $F^{\tau(\ell)}(\tilde{\RV{x}}^{(\ell)})$ agree everywhere.  That $F^t(\tilde{\RV{x}}^{(\ell-1)})=F^t(\tilde{\RV{x}}^{(\ell)})$ for all $t\geq\tau(\ell)$ follows immediately.
			
		\item That $\Delta\big(\RV{y}^{(\ell-1)},\RV{y}^{(\ell)}\big)\subseteq\bigcup_{C\in\family{C}_{\ell}(\RV{Q})} \Moore^{r\tau(\ell)}(C)\subseteq\bigcup_{C\in\family{C}_{\ell}(\RV{Q})} \Moore_\rho(C)$ is immediate from the construction.		
		
		\item By the induction hypothesis,
			\begin{align}
				\Delta\big(\RV{y}^{(\ell-1)},\tilde{\RV{x}}^{(\ell-1)}\big) &\subseteq \bigcup_{C\in\family{C}_{\geq\ell}(\RV{Q})} C \;,
				&
				\Delta\big(\RV{y}^{(\ell-1)},\RV{y}^{(\ell)}\big) &\subseteq\bigcup_{C\in\family{C}_{\ell}(\RV{Q})}\Moore^{r\tau(\ell)}(C) \;.
			\end{align}
			On the other hand, by construction, $\tilde{\RV{x}}^{(\ell)}$ and $\RV{y}^{(\ell)}$ agree on $\bigcup_{C\in\family{C}_{\ell}(\RV{Q})}\Moore^{r\tau(\ell)}(C)$.
			It follows that $\Delta\big(\RV{y}^{(\ell)},\tilde{\RV{x}}^{(\ell)}\big)\subseteq\bigcup_{C\in\family{C}_{>\ell}(\RV{Q})} C$.			
	\end{enumerate}
	This completes the proof of the lemma.
\end{proof}

Item~\ref{item:lem:stabilisation:random-perturbation:main:patching} in Lemma~\ref{lem:stabilisation:random-perturbation:main} ensures that $\RV{y}^{(\ell)}\to\RV{y}$ for some $\RV{y}\in X$.  Indeed, let $k\in\ZZ^d$.  By the thinness condition in the definition of $\rho$-sparseness, we have $k\in \Moore_\rho(C)$ for at most finitely many $C\in\family{C}(\RV{Q})$.  It follows that $\RV{y}^{(\ell-1)}_k\neq\RV{y}^{(\ell)}_k$ for no more than finitely many values of~$\ell$.  Define $\RV{y}_k$ as the eventual value of~$\RV{y}^{(\ell)}_k$.  Then, $\RV{y}^{(\ell)}\to\RV{y}$ in the product topology.  Furthermore, we have $\RV{y}\in X$ because $X$ is closed.

\begin{lemma}[Attraction]
\label{lem:stabilisation:random-perturbation:attraction}
	$F^t(\tilde{\RV{x}})\to\RV{y}$ as $t\to\infty$.
\end{lemma}
\begin{proof}
	Let $t\geq 0$, and choose $\ell$ such that $\tau(\ell)\leq t\leq\tau(\ell+1)$.
	By item~\ref{item:lem:stabilisation:random-perturbation:main:forgetting} in Lemma~\ref{lem:stabilisation:random-perturbation:main}, $F^t(\tilde{\RV{x}})=F^t\big(\tilde{\RV{x}}^{(\ell)}\big)$.
	By item~\ref{item:lem:stabilisation:random-perturbation:main:correcting} in Lemma~\ref{lem:stabilisation:random-perturbation:main}, $\Delta\big(\RV{y}^{(\ell)},\tilde{\RV{x}}^{(\ell)}\big)\subseteq\bigcup_{C\in\family{C}_{>\ell}(\RV{Q})} C$.  This, together with Observation~\ref{obs:speed-of-light}, implies that $\Delta\Big(F^t\big(\RV{y}^{(\ell)}\big),F^t\big(\tilde{\RV{x}}^{(\ell)}\big)\Big)\subseteq\bigcup_{C\in\family{C}_{>\ell}(\RV{Q})}\Moore^{rt}(C)$.  Recall that $\RV{y}^{(\ell)}$ is an element of $X$ and thus $F^t\big(\RV{y}^{(\ell)}\big)=\RV{y}^{(\ell)}$.  Note further that for $C\in\family{C}_{>\ell}(\RV{Q})$, we have $2rt\leq\rho(\ell+1)\leq\rho\big({\diam(C)}\big)$, and thus $\Moore^{rt}(C)\subseteq \Moore_\rho(C)$.  It follows that
	\begin{align}
		\Delta\Big(\RV{y}^{(\ell)},F^t\big(\tilde{\RV{x}}\big)\Big) &=
			\Delta\Big(F^t\big(\RV{y}^{(\ell)}\big),F^t\big(\tilde{\RV{x}}\big)\Big) \\
		&=
			\Delta\Big(F^t\big(\RV{y}^{(\ell)}\big),F^t\big(\tilde{\RV{x}}^{(\ell)}\big)\Big) \\
		&\subseteq \bigcup_{C\in\family{C}_{>\ell}(\RV{Q})}\Moore^{rt}(C)
			\subseteq\bigcup_{C\in\family{C}_{>\ell}(\RV{Q})}\Moore_\rho(C) \;.
	\end{align}
	From item~\ref{item:lem:stabilisation:random-perturbation:main:patching} in Lemma~\ref{lem:stabilisation:random-perturbation:main}, it also follows that $\Delta\big(\RV{y},\RV{y}^{(\ell)}\big)\subseteq\bigcup_{C\in\family{C}_{>\ell}(\RV{Q})} \Moore_\rho(C)$.  Thus,
	\begin{align}
		\Delta\Big(\RV{y},F^t\big(\tilde{\RV{x}}\big)\Big)
		&\subseteq\bigcup_{C\in\family{C}_{>\ell}(\RV{Q})}\Moore_\rho(C) \;.
	\end{align}
	Let $k\in\ZZ^d$.  By the thinness condition in the definition of $\rho$-sparseness, we know that $k\in \Moore_\rho(C)$ for no more than finitely many $C\in\family{C}(\RV{Q})$.  Thus, the value of $F^t\big(\tilde{\RV{x}}\big)_k$ will eventually stabilise at $\RV{y}_k$.  We conclude that $F^t\big(\tilde{\RV{x}}\big)\to\RV{y}$ as $t\to\infty$.
\end{proof}

\begin{lemma}[Stability]
\label{lem:stabilisation:random-perturbation:stability}
	$\updensity\big(\Delta(x,\RV{y})\big)<\delta$ almost surely.
\end{lemma}
\begin{proof}
	In order to have $\RV{y}_k\neq x_k$, we must have $\RV{y}^{(\ell)}_k\neq\RV{y}^{(\ell-1)}_k$ for some $\ell\in\ZZ$.
	By item~\ref{item:lem:stabilisation:random-perturbation:main:patching} in Lemma~\ref{lem:stabilisation:random-perturbation:main}, this means that $k\in\bigcup_{C\in\family{C}(\RV{Q})} \Moore_\rho(C)$.  However, by property~\ref{item:thm:sparseness:linear:rareness} in Theorem~\ref{thm:sparseness:linear},
	\begin{align}
		\updensity\left(\bigcup_{C\in\family{C}(\RV{Q})} \Moore_\rho(C)\right) &<\delta \;.
	\end{align}
	The claim follows.

\medskip
This concludes the proof of Theorem~\ref{thm:stabilisation:random-perturbation:linear}.
\end{proof}

\section{Discussion and open problems}
\label{sec:discussion}

\subsection{Stabilising $3$-colourings}
\label{sec:discussion:3col}

In Section~\ref{sec:2d}, we presented self-stabilising CA for different families of two-dimensional tiling spaces, including $k$-colourings for $k\not=3$.
When $k\neq 3$, the $k$-colourings have the property that any configuration with a finite island of defects can be corrected in a purely local manner, by modifying the configuration only in a bounded region around the defects.
In contrast, as remarked in Example~\ref{ex:3col}, this is not always possible for two-dimensional $3$-colourings.
This makes the self-stabilisation problem more challenging.

\begin{question}[Self-stabilisation of $3$-colourings]
Is there a deterministic CA that stabilises $3$-colourings from finite perturbations in polynomial time, ideally without additional symbols?
\end{question}

\noindent We conjecture that it is not possible to stabilise $3$-colourings in linear time, but that it could be possible in quadratic time, at least if we allow additional symbols.

In order to shed some light on this problem, let us describe a representation of the two-dimensional $3$-colourings based on the configurations of the so-called \emph{six-vertex} model, and in the process, explain Figure~\ref{fig:SWerror}.
The correspondence between $3$-colourings and the six-vertex model was first discovered by
A.~Lenard (see~\cite[Section~8.13]{Bax82}).\vspace*{-2mm}

\paragraph{Connection with the six-vertex model.}
Let us start with a valid $3$-colouring in two dimensions.
For each pair of neighbouring cells (horizontal or vertical), let us draw an arrow on the boundary between the two cells according to the following rule.
Let $q$ and $q'$ be the colours of the two neighbouring cells. Since $ q' \neq q$, we either have $ q'= q + 1 \pmod{3} $ or $ q'= q - 1 \pmod{3} $. Depending on this, we draw the arrow in one direction or the other:
\begin{itemize}
\itemsep=0.95pt
\item The arrow on a vertical boundary is directed upwards if the colour on its right is one more than the colour on its left (modulo~$3$), and downwards otherwise.
\item The arrow on a horizontal boundary is directed towards the right if the colour below it is one more than the colour above it (modulo~$3$), and towards the left otherwise.
\end{itemize}
These conventions are depicted in Figure~\ref{fig:sixVertex}.

\begin{figure}[!h]
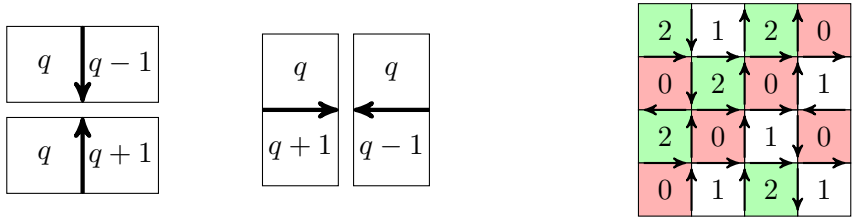

\begin{center}
	\begin{tabular}{m{7em}m{7em}@{\qquad\qquad\qquad}m{10em}}
		\verC &
		\horC &
		\figExampleSV
	\end{tabular}
\end{center}\vspace*{-5mm}
\caption{The convention used for encoding 3-colouring configurations in the six-vertex model and an example of a configuration with its associated six-vertex image.}
\label{fig:sixVertex}\vspace*{-3mm}
\end{figure}

\medskip
One can then check that starting from a $3$-colouring, the resulting arrow configuration is such that at each vertex, there are exactly two incoming arrows and two outgoing arrows. Conversely, from a six-vertex configuration (i.e., a configuration of arrows with equal number of incoming and outgoing arrows at each vertex), there are exactly three $3$-colourings giving rise to that arrow configuration. (Once we choose the colour of one cell, all the other colours can be deduced.)

\medskip
Figure~\ref{fig:sixVertex} shows an example of such an encoding of a valid $3$-colouring. By contrast, Figure~\ref{fig:SWerror} depicts a finite perturbation of a 3-colouring.
Notice that, in Figure~\ref{fig:SWerror}, the arrows pointing to the South and the ones pointing to the West are drawn in bold.
In a valid six-vertex configuration, the knowledge of these two types of arrows is sufficient to fully describe the configuration: indeed, the other horizontal or vertical arrows have to be East or North arrows, respectively.
In Figure~\ref{fig:SWerror}, we have also shaded the unperturbed portion of the configuration, thus the perturbed region consists in the inner (unshaded) square.
Therefore, in order to correct the tiling, we need to fill in the inner square with an admissible configuration.
One can verify that the only way to do so corresponds to a six-vertex configuration that has a direct downward vertical path.  Indeed, there is only one bold incoming arrow and one bold outgoing arrow at the boundary of the square, and these two arrows have to be connected.

The example in Figure~\ref{fig:SWerror} shows that we can construct finite perturbations of a $3$-colouring which contain only two defective cells (i.e., a single interface with same colour), but for which in order to obtain a valid configuration, one needs to modify an arbitrarily large domain.
In fact, taking a limit of such configurations, one obtains a configuration with only two defective cells which is not a finite perturbation of any valid  $3$-colouring.\vspace*{-2mm}

\paragraph{A sequential correction procedure.}
In order to tackle the problem of self-stabilisation for $3$-colour\-ings, it may be helpful to address the simpler question of how to correct a perturbed $3$-colouring with a conventional, sequential, non-local algorithm.

Suppose we are given a $3$-colouring with a finite number of defects and a square region~$S$ of size~$\ell$ containing all the defective cells.  Is there a simple (sequential, non-local) algorithm to decide if the tiling can be corrected by modifying only the colour of the cells inside~$S$?
The $3$-colouring problem on general graphs is $\classNP$-complete, and an exhaustive search takes an exponential amount of time in~$\ell^2$.
Nevertheless, in this case, the representation in terms of the six-vertex configurations allows us to solve the decision problem in linear time, and to find the actual correction (when it exists) in quadratic time.

More precisely, consider the pattern of incoming and outgoing arrows on the boundary of~$S$.
In order to decide if the colouring of~$S$ can be corrected, we just need to know if it is possible to pair the incoming and outgoing arrows on the boundary of~$S$ in an admissible way.
This is easy to do sequentially in linear time.
Starting from the NE-corner, let us enumerate the incoming arrows on the North and the West sides from 1 to $n_{\incoming}$ counter-clockwise, and the outgoing arrows on the East and the South sides from 1 to $n_{\outgoing}$ clockwise.
The square can be coloured if we can match each incoming arrow number~$k$ with the outgoing arrow number~$k$ by a SE-path of arrows.  (In particular, this would imply $n_{\incoming}=n_{\outgoing}$.)  In order to know if this can be done, we try to match successively the incoming and outgoing arrows from $1$ to $n_{\incoming}=n_{\outgoing}$ with disjoint paths, by moving East if the edge has not already been selected, and South otherwise. As an additional condition, we need to ensure that at each step, the path does not go beyond the corresponding outgoing arrow or come across another path. This procedure succeeds if and only if there is at least one admissible matching.  See Figure~\ref{fig:SWcontour} for an illustration.

\begin{figure}[h]
\begin{center}
\begin{tikzpicture}[scale=0.5,>=stealth']
\figSVcontour
\end{tikzpicture}
\end{center}\vspace*{-5mm}
\caption{Example of valid matching of the arrows of the contour, with the procedure that match successively the incoming and outgoing arrows, by moving East if possible, and South otherwise.}
\label{fig:SWcontour}
\end{figure}
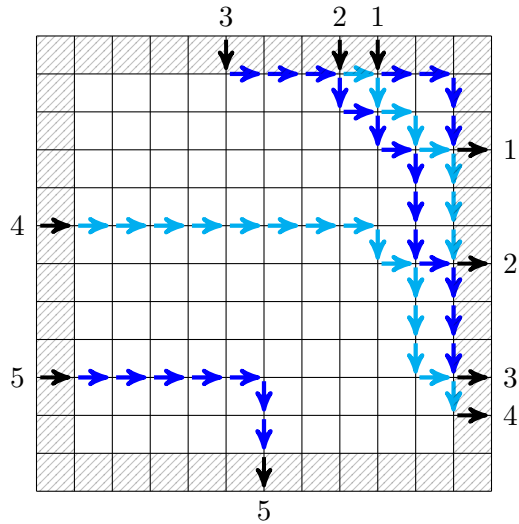

At this point, we may want to try to turn the above sequential procedure into a self-stabilising CA.
The CA would start by marking squares encompassing the defects.  On each such square, the CA would then simulate a Turing machine that performs the above sequential procedure.  If the procedure is successful, a signal is propagated throughout the square in order to erase the markings.  If the procedure is unsuccessful, another signal is sent to increase the size of the square, and to repeat the simulation of the Turing machine on this larger square.  If two squares collide, they merge and form a larger square.
Such a CA can indeed be constructed.
However, the main difficulty in turning such a construction into a self-stabilising CA is that the initial perturbation could now involve the extra states.  The construction must ensure stabilisation starting from any such perturbation, not just the perturbations involving the three colours.

\subsection{Probabilistic self-stabilisation}
\label{sub:disc_proba}

In this paper, we have barely touched the topic of self-stabilisation with probabilistic rules.  There are several questions left to be answered, and much more to be explored.

\paragraph{An isotropic candidate for stabilising $4$-colourings.}

In Section~\ref{sec:PCA}, we have presented isotropic probabilistic cellular automata achieving self-stabilisation on two-dimensional $k$-colourings, with $k=2$ and $k\geq 5$.
We now propose a rule which we believe does the job for $k=4$, but for which we have no formal proof of convergence.
The idea is to modify the method used for the case $k\geq 5$, and make an exception when there is no colour available to directly correct a defective cell.
More specifically, for $\alpha\in (0,1)$, we define a probabilistic CA with the following rule.
If a cell~$k$ is not defective, its state is kept unchanged.  Otherwise, the state of~$k$ is changed, with probability~$\alpha$, to a colour which is chosen at random from among all colours consistent with the current colours of its four neighbours.  If no consistent colour exists, a colour is chosen at random from among all possible colours.
As usual, different cells are updated independently.
Experimentally, we have observed that this rule rapidly corrects the defects.
However, unlike the case $k\geq 5$, for $k=4$, we cannot ensure with the above rule that the defects stay in some bounded area.
See Figure~\ref{diag:fourstatepca} for an example of stabilisation in this probabilistic CA.

\begin{figure}[!h]
\vspace*{1mm}
\hspace*{-2mm}\begin{tabular}{c c c c}
\includegraphics[width=0.218\linewidth]{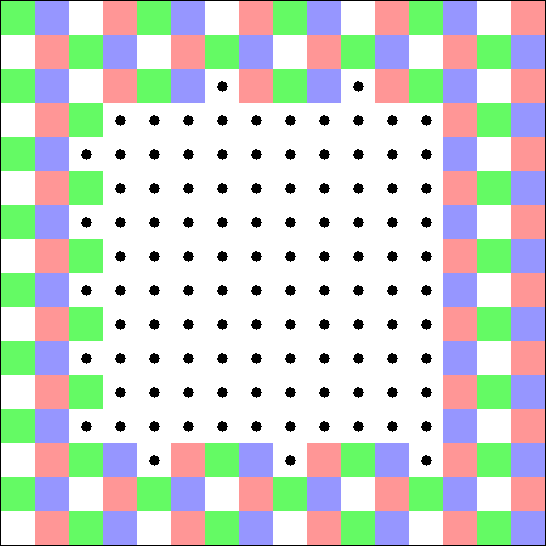} &
\includegraphics[width=0.218\linewidth]{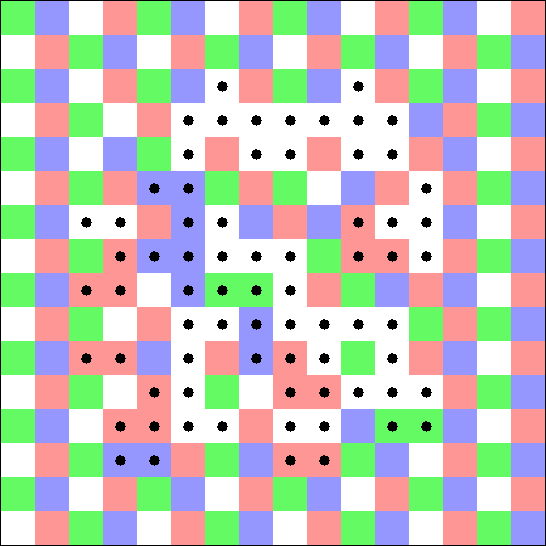} &
\includegraphics[width=0.218\linewidth]{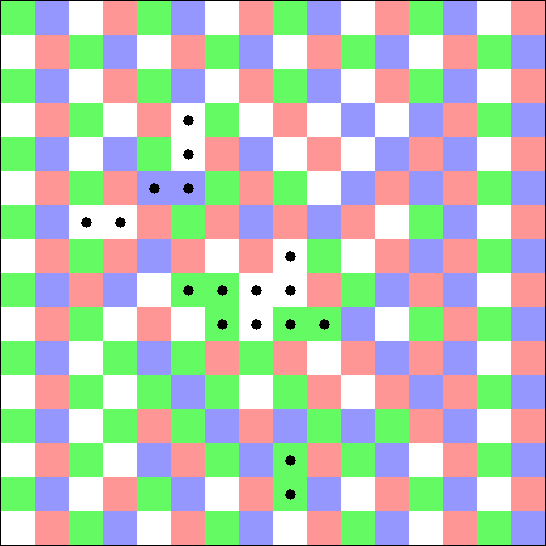} &
\includegraphics[width=0.218\linewidth]{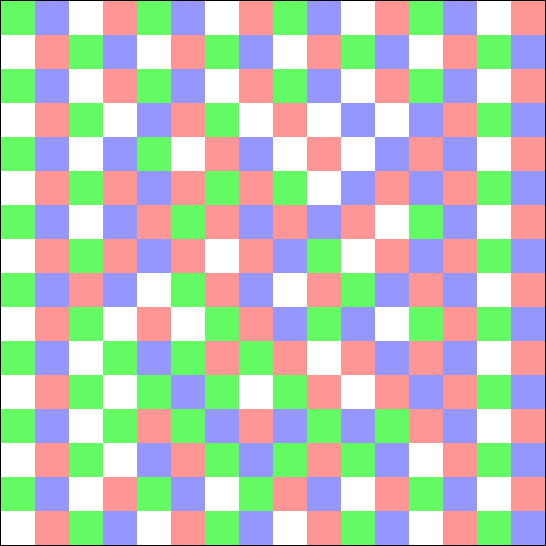} \\
$ t = 0 $ & $ t= 1$ & $ t= 3 $ & $ t = 8  $
\end{tabular}\vspace*{-2mm}
\caption{Illustration of the evolution of the probabilistic CA proposed in Section~\ref{sub:disc_proba} for stabilising $4$-colourings. The dots indicate the defective cells.}
\label{diag:fourstatepca}\vspace*{-3mm}
\end{figure}

\begin{question}[Probabilistic self-stabilisation of $4$-colourings]
Does the probabilistic CA defined above stabilise $4$-colourings from finite perturbations, for any or for some values of $\alpha \in (0,1)$?
\end{question}

We suspect that the answer is positive.
To support this claim, one can try to look for configurations for which this rule could potentially fail to stabilise.
Consider the configuration depicted in Figure~\ref{fig:diabolic}.
It has only two defective cells in the center.  Furthermore, all cells (including the defective ones), see the three other colours in their neighbourhoods.
Consequently, if one of the defective cells changes its state alone, it will remain defective.
For this specific configuration, some kind of coordination is thus necessary, which cannot here occur by a specific mechanism as in the deterministic case.
This at first might suggest that the defects may propagate arbitrary far from their origin.
However, we have experimentally observed that this is not the case: defects have a tendency to stay in the same area, and the correcting process is more rapid than the diffusion of the defects. Surprisingly enough, even when the cells are updated successively at random (i.e., the fully asynchronous case), we also noticed that the rule succeeds in correcting the defects. Indeed, when the defects propagate, they modify the configuration in such a way that the property of seeing three different colours in the neighbourhood is lost, which finally enables a correction to take place. This observation supports the idea that when we use parallel updates, as we do in our rule, we can only increase the possibilities of correction.

\begin{figure}[h]
\vspace*{3mm}
\begin{center}
\begin{tikzpicture}[scale=0.7]
\pgfmathsetmacro{\ln}{4}
\foreach \i in {3,7,11} {\sqaa{\i}{\ln}} ;
\foreach \i in {2,6,10} {\sqdd{\i}{\ln}} ;
\foreach \i in {1,5,9} {\sqcc{\i}{\ln}} ;
\foreach \i in {4,8,12} {\sqbb{\i}{\ln}} ;
\pgfmathsetmacro{\ln}{3}
\foreach \i in {3,7,11} {\sqcc{\i}{\ln}} ;
\foreach \i in {2,6,10} {\sqbb{\i}{\ln}} ;
\foreach \i in {1,5,9} {\sqaa{\i}{\ln}} ;
\foreach \i in {4,8,12} {\sqdd{\i}{\ln}} ;
\pgfmathsetmacro{\ln}{2}
\foreach \i in {1,5,10} {\sqdd{\i}{\ln}} ;
\foreach \i in {4,9} {\sqcc{\i}{\ln}} ;
\foreach \i in {2,11}{\sqaa{\i}{\ln}} ;
\foreach \i in {3,8,12}{\sqbb{\i}{\ln}} ;
\pgfmathsetmacro{\ln}{1}
\foreach \i in {3,7,11} {\sqdd{\i}{\ln}} ;
\foreach \i in {2,6,10} {\sqcc{\i}{\ln}} ;
\foreach \i in {1,5,9} {\sqbb{\i}{\ln}} ;
\foreach \i in {4,8,12} {\sqaa{\i}{\ln}} ;
\pgfmathsetmacro{\ln}{0}
\foreach \i in {3,7,11} {\sqbb{\i}{\ln}} ;
\foreach \i in {2,6,10} {\sqaa{\i}{\ln}} ;
\foreach \i in {1,5,9} {\sqdd{\i}{\ln}} ;
\foreach \i in {4,8,12} {\sqcc{\i}{\ln}} ;
\sqee 6 2
\sqee 7 2
\end{tikzpicture}
\end{center}\vspace*{-3mm}
\caption{A $3$-colouring with a single defect (i.e., two defective cells).  The configuration has the property that every cell sees the other three colours in its neighbourhood.}
\label{fig:diabolic}\vspace*{-4mm}
\end{figure}

\paragraph{No isotropic candidate for stabilising $3$-colourings.}
In contrast with the previous cases, when we have three colours, we cannot use the method of taking an available colour or a random colour when no colour is available.  We noticed experimentally that the errors diffuse and we could not find any rule that keeps them confined, even in statistical terms.

\begin{question}[Probabilistic self-stabilisation of $3$-colourings]
Is there a probabilistic CA that stabilises $3$-colourings from finite perturbations, ideally rapidly and having the same symmetries (isotropy, symmetry of colours) as the tiling space?
\end{question}

\paragraph{Self-stabilisation of the {$\MRIE$} rule.}
In Section~\ref{sec:pca:finite}, we showed that the {$\MRIE$} probabilistic CA stabilises the homogeneous space~$\Hom_2$ in no more than cubic time.

\begin{question}[Speed of stabilisation in the {$\MRIE$} CA]
What is the precise speed of stabilisation in the \textup{$\MRIE$} probabilistic CA?
\end{question}

\noindent We conjecture that the stabilisation occurs in quadratic time.
See the discussion after the proof of Proposition~\ref{prop:hom_prob}.

\paragraph{Self-stabilisation from random perturbations.}

Does the result of Theorem~\ref{thm:stabilisation:random-perturbation:linear} extend in some form to probabilistic CA?
\begin{question}[Stabilisation from random perturbations]
Suppose that a probabilistic CA stabilises an SFT $X$ from finite perturbations in linear (or sub-quadratic) time.  Does the CA also stabilise $X$ from random perturbations?
\end{question}

\noindent As a special case, it would be interesting to know if the {$\MRIE$} rule stabilises from random perturbations, as in the case of its continuous-time counterpart (see Example~\ref{exp:MRIE:continuous-times}).

\subsection{Other open questions}
\label{sec:discussion:other-open}

\paragraph{Stabilising any tiling space.}

In the current paper, we have searched for \emph{efficient} solutions to the self-stabilisation problem for specific classes of SFTs.  However, even if we drop the efficiency ambition, it is not clear if one can always find a CA stabilising any given SFT.

\begin{question}[General solution]
\label{q:open:general-solution}
Is it true that for every SFT~$X$, there exists a deterministic CA that stabilises~$X$ from finite perturbations, possibly in exponential time in the number of errors?  What if we require the solution to have no extra symbols compared to the alphabet of the SFT?
\end{question}

\noindent We conjecture that an approach similar to the one suggested for $3$-colourings in the last paragraph of Section~\ref{sec:discussion:3col} should be possible for a general SFT.  In particular, one should be able to come up with a general construction for translating a Turing machine that solves the functional version of the square tiling problem for a finite set of Wang tiles (see Section~\ref{sec:complexity:hardness}) into a self-stabilising CA for the valid tilings.
This approach, if successful, stabilises the SFT in at most exponential time by testing all the possible ways to fill in a given region.
As before, the main difficulty would be to handle the appearance of the extra symbols in the initial perturbation.

\paragraph{Stabilisation in presence of temporal noise.}

Recall from Example~\ref{exp:toom:noise} that Toom's {$\toom$} CA maintains a form of stability on~$\Hom_2$ even in presence of temporal noise.
A comparison argument with Toom's CA can be used to show that the two CA discussed in Sections~\ref{sub:finite} and~\ref{sub:single_cell} have the same stability property.

\begin{question}[Self-stabilisation in presence of temporal noise]
Under which conditions does a CA stabilise an SFT in presence of temporal noise?
\end{question}

\noindent In particular, do the CA discussed in Sections~\ref{sec:1d} and~\ref{sec:2d:deterministic} self-stabilise in presence of temporal noise?

\section*{Acknowledgments}
\addcontentsline{toc}{section}{Acknowledgments}
We would like to thank Peter G\'acs and Ilkka T\"orm\"a for helpful discussions and for communicating their results with us.
We are grateful to Daniel Martins Fiebich for pointing out some inaccuracies in an earlier version of the paper, in particular in the construction of Section~\ref{sec:1d}.
We would also like to thank the two anonymous referees for their careful reading of the text and for providing numerous helpful suggestions.

\newcommand{\noopsort}[1]{}



\end{document}